\RequirePackage[l2tabu,orthodox]{nag}
\documentclass
[letterpaper,11pt,]
{article}

\usepackage{float}
\usepackage{nag}
\usepackage{lmodern}
\usepackage{etex}
\usepackage{verbatim}
\usepackage{xspace,enumerate}
\usepackage[dvipsnames]{xcolor}
\usepackage[T1]{fontenc}
\usepackage[full]{textcomp}
\usepackage[american]{babel}
\usepackage{mathtools}
\usepackage{amsthm}
\usepackage{thmtools}
\usepackage{thm-restate}

\usepackage[
letterpaper,
top=1in,
bottom=1in,
left=1in,
right=1in]{geometry}
\usepackage{mdframed}

\usepackage{newpxtext} %
\usepackage{textcomp} %
\usepackage[varg,bigdelims]{newpxmath}
\usepackage[scr=rsfso]{mathalfa}%
\usepackage{bm} %
\let\mathbb\varmathbb
\usepackage{microtype}
\usepackage[pagebackref,colorlinks=true,urlcolor=blue,linkcolor=blue,citecolor=OliveGreen]{hyperref}
\usepackage[capitalise,nameinlink]{cleveref}
\crefname{lemma}{Lemma}{Lemmas}
\crefname{fact}{Fact}{Facts}
\crefname{theorem}{Theorem}{Theorems}
\crefname{corollary}{Corollary}{Corollaries}
\crefname{claim}{Claim}{Claims}
\crefname{example}{Example}{Examples}
\crefname{algorithm}{Algorithm}{Algorithms}
\crefname{problem}{Problem}{Problems}
\crefname{definition}{Definition}{Definitions}
\crefname{exercise}{Exercise}{Exercises}
\usepackage{amsthm}

\newtheorem{theorem}{Theorem}[section]
\newtheorem*{theorem*}{Theorem}
\newtheorem{lemma}[theorem]{Lemma}
\newtheorem*{lemma*}{Lemma}
\newtheorem{fact}[theorem]{Fact}
\newtheorem*{fact*}{Fact}
\newtheorem{proposition}[theorem]{Proposition}
\newtheorem*{proposition*}{Proposition}
\newtheorem{corollary}[theorem]{Corollary}
\newtheorem*{corollary*}{Corollary}

\newtheorem*{hypothesis*}{Hypothesis}
\newtheorem{conjecture}[theorem]{Conjecture}
\newtheorem*{conjecture*}{Conjecture}
\theoremstyle{definition}
\newtheorem{definition}[theorem]{Definition}
\newtheorem*{definition*}{Definition}

\newtheorem*{construction*}{Construction}

\newtheorem*{example*}{Example}

\newtheorem*{question*}{Question}
\newtheorem{algorithm}[theorem]{Algorithm}
\newtheorem*{algorithm*}{Algorithm}

\newtheorem*{assumption*}{Assumption}
\newtheorem{problem}[theorem]{Problem}
\newtheorem*{problem*}{Problem}

\newtheorem*{openquestion*}{Open Question}
\theoremstyle{remark}

\newtheorem*{claim*}{Claim}
\newtheorem{remark}[theorem]{Remark}
\newtheorem*{remark*}{Remark}

\newtheorem*{observation*}{Observation}
\usepackage{paralist}
\let\originalleft\left
\let\originalright\right
\renewcommand{\left}{\mathopen{}\mathclose\bgroup\originalleft}
\renewcommand{\right}{\aftergroup\egroup\originalright}
\usepackage{turnstile}
\usepackage{mdframed}
\usepackage{tikz}
\usepackage{caption}
\DeclareCaptionType{Algorithm}
\usepackage{newfloat}
\usepackage{xparse}
\usepackage{amsthm} %
\makeatletter
\let\latexparagraph\paragraph
\RenewDocumentCommand{\paragraph}{som}{%
  \IfBooleanTF{#1}
    {\latexparagraph*{#3}}
    {\IfNoValueTF{#2}
       {\latexparagraph{\maybe@addperiod{#3}}}
       {\latexparagraph[#2]{\maybe@addperiod{#3}}}%
  }%
}
\newcommand{\maybe@addperiod}[1]{%
  #1\@addpunct{.}%
}
\makeatother

\usepackage{boxedminipage}

\newcommand{\norm}[1]{\lVert#1\rVert}

\newcommand{\Iprod}[1]{\left\langle#1\right\rangle}

\newcommand{\Esymb}{\mathbb{E}}
\newcommand{\Psymb}{\mathbb{P}}
\newcommand{\Vsymb}{\mathbb{V}}
\DeclareMathOperator*{\E}{\Esymb}
\DeclareMathOperator*{\Var}{\Vsymb}
\DeclareMathOperator*{\ProbOp}{\Psymb}
\renewcommand{\Pr}{\ProbOp}

\newcommand{\defeq}{\stackrel{\mathrm{def}}=}
\newcommand{\seteq}{\mathrel{\mathop:}=}

\newcommand\bdot\bullet

\DeclareMathOperator{\poly}{poly}

\DeclareMathOperator{\supp}{supp}

\newcommand{\Z}{\mathbb Z}
\newcommand{\N}{\mathbb N}
\newcommand{\R}{\mathbb R}

\newcommand{\cA}{\mathcal A}

\newcommand{\cD}{\mathcal D}
\newcommand{\cE}{\mathcal E}
\newcommand{\cF}{\mathcal F}
\newcommand{\cG}{\mathcal G}

\newcommand{\cL}{\mathcal L}

\newcommand{\cN}{\mathcal N}

\newcommand{\cP}{\mathcal P}
\newcommand{\cQ}{\mathcal Q}

\newcommand{\cS}{\mathcal S}

\newcommand{\cZ}{\mathcal Z}

\renewcommand{\leq}{\leqslant}
\renewcommand{\le}{\leqslant}
\renewcommand{\geq}{\geqslant}
\renewcommand{\ge}{\geqslant}
\let\epsilon=\varepsilon
\numberwithin{equation}{section}
\newcommand\MYcurrentlabel{xxx}
\newcommand{\MYstore}[2]{%
  \global\expandafter \def \csname MYMEMORY #1 \endcsname{#2}%
}
\newcommand{\MYload}[1]{%
  \csname MYMEMORY #1 \endcsname%
}
\newcommand{\MYnewlabel}[1]{%
  \renewcommand\MYcurrentlabel{#1}%
  \MYoldlabel{#1}%
}
\newcommand{\MYdummylabel}[1]{}
\newcommand{\torestate}[1]{%
  \let\MYoldlabel\label%
  \let\label\MYnewlabel%
  #1%
  \MYstore{\MYcurrentlabel}{#1}%
  \let\label\MYoldlabel%
}
\newcommand{\restatetheorem}[1]{%
  \let\MYoldlabel\label
  \let\label\MYdummylabel
  \begin{theorem*}[Restatement of \cref{#1}]
    \MYload{#1}
  \end{theorem*}
  \let\label\MYoldlabel
}
\newcommand{\restatelemma}[1]{%
  \let\MYoldlabel\label
  \let\label\MYdummylabel
  \begin{lemma*}[Restatement of \cref{#1}]
    \MYload{#1}
  \end{lemma*}
  \let\label\MYoldlabel
}
\newcommand{\restateprop}[1]{%
  \let\MYoldlabel\label
  \let\label\MYdummylabel
  \begin{proposition*}[Restatement of \cref{#1}]
    \MYload{#1}
  \end{proposition*}
  \let\label\MYoldlabel
}
\newcommand{\restatefact}[1]{%
  \let\MYoldlabel\label
  \let\label\MYdummylabel
  \begin{fact*}[Restatement of \cref{#1}]
    \MYload{#1}
  \end{fact*}
  \let\label\MYoldlabel
}
\newcommand{\restate}[1]{%
  \let\MYoldlabel\label
  \let\label\MYdummylabel
  \MYload{#1}
  \let\label\MYoldlabel
}

\newcommand{\eps}{\epsilon}

\allowdisplaybreaks
\DeclarePairedDelimiterX{\infdivx}[2]{(}{)}{#1\;\delimsize\|\;#2}

\newcommand*{\Lowner}{L\"owner\xspace}

\title{
  On efficient robust regression with subquadratic samples
}

\renewcommand{\epsilon}{\ensuremath\varepsilon}

\renewcommand{\phi}{\ensuremath{\varphi}}

\newcommand{\sn}{\frac{1}{n}\sum\limits_{i=1}^n}

\newcommand{\vnorm}[1]{\left\lVert#1\right\rVert}

\newcommand{\D}{\mathcal{D}}

\author{
  Deeksha Adil\thanks{Department of Computer Science, ETH Zurich.}
  \and
  Jarosław Błasiok\thanks{Department of Computing Sciences, Bocconi University.}
  \and
  Hongjie Chen\footnotemark[1]
  \and 
  Deepak Narayanan Sridharan\footnotemark[1]
}

\date{}

\begin{document}

\maketitle

\begin{abstract}%
We revisit the problem of robust linear regression under Gaussian covariates with an unknown covariance matrix of condition number $\kappa$.
For this fundamental problem, significant gaps remain in our understanding of the trade-offs among sample complexity, condition number, runtime, and prediction error for efficient algorithms.
Our first result is a near-linear-time algorithm that uses $\widetilde{O}(d/\epsilon^4)$ samples, where $d$ is the dimension and $\epsilon$ is the corruption rate, and achieves prediction error $O(\sqrt{\epsilon\kappa})$ under the condition $\epsilon\kappa \lesssim 1$, improving over all prior works.
We complement this result with a Statistical Query (SQ) lower bound showing that efficient SQ algorithms achieving error $o(\sqrt{\epsilon\kappa})$ when $\eps \kappa \lesssim 1$ require queries that take $\Omega(d^2)$ samples to simulate. 
Finally, we prove a low-degree polynomial lower bound that gives fine-grained evidence that, without assumptions such as $\eps \kappa \lesssim 1$, efficient algorithms may require $\tilde{\Omega}\left(\min\{d\eps^{2}\kappa^{2},\ \eps^{2}d^{2}\}\right)$ samples to significantly outperform the trivial estimator that always guesses $0$.

\end{abstract}

\newpage
\tableofcontents
\newpage

\section{Introduction}

Linear regression is one of the most fundamental tools in statistics, optimization, and machine learning; however, its standard formulations are extremely sensitive to corrupted observations \cite{rousseeuw2003robust}. Even a small fraction of outliers---whether in the labels or covariates---can arbitrarily bias the least squares estimator. This has motivated a long line of work in robust statistics \cite{huber1964robust, hampel1974influence, rousseeuw1984lms} on designing statistically efficient estimators that can effectively handle outliers.  

We study robust linear regression under \emph{strong} contamination. %
The strong contamination model imposes no structural or distributional assumptions on the corruptions; it only assumes that at most a small fraction of the dataset is corrupted (potentially adversarially). 
This makes it a flexible framework and subsumes several other commonly studied data corruption models, such as Huber contamination (see \cite{diakonikolas2023algorithmic} for a detailed discussion).
We now give the formal definition.
\begin{definition}[Robust linear regression]\label{def:epscorr}
    Given a regression vector $\beta \in \R^d$, let $D_{X, Y}$ be the joint distribution on $\R^{d} \times \R$ defined by the linear model 
    \[
        Y = \langle X,\beta \rangle + \eta,
    \] 
    where $X \sim D_X$ is a mean-zero distribution with covariance $\Sigma$, and $\eta$ denotes observation noise independent from $X$.
    We say $\{(X_i', Y_i')\}_{i=1}^n$ are $\eps$-corrupted samples from $D_{X, Y}$ if $\{(X_i,Y_i)\}_{i=1}^n \overset{\mathrm{iid}}{\sim} \D_{X, Y}$ and $(X'_i, Y'_i) = (X_i, Y_i)$ for at least $(1-\eps) n$ indices $i$.
    Given $\epsilon$-corrupted samples, the goal is to output an estimate $\widehat{\beta}$ with small error $\|\widehat{\beta} - \beta\|_{\Sigma}$.\footnote{This error is equivalent to prediction error under the covariate distribution, as $\E\limits_{X \sim D_X}|\langle X, \beta\rangle - \langle X, \widehat{\beta}\rangle|^2 = \|\beta - \widehat{\beta}\|^2_{\Sigma}$, where $\|v\|_{\Sigma}^2 := v^T \Sigma v$.}
\end{definition}

While robust regression has been studied for several decades under a variety of settings (see e.g. \cite{rousseeuw2003robust}), polynomial-time algorithms robust to strong contamination were developed only recently, beginning with the works of \cite{klivans2018efficient, diakonikolas2019sever, diakonikolas2019efficient, prasad2020robust}.
Since the initial papers, a growing body of work has proposed polynomial-time algorithms under various distributional assumptions\footnote{Some form of distributional assumption on the covariates that is stronger than bounded covariance is information-theoretically necessary~\cite{bakshi2021robust}.} \cite{cherapanamjeri2020optimalrobustlinearregression,jambulapati2021robust, bakshi2021robust, zhu2022robust, pensia2025robust}. 
Despite this progress, even in the basic setting of Gaussian covariates with unknown covariance, the optimal trade-off between sample complexity, robustness, and prediction error---whether it can be achieved by polynomial-time algorithms---remains unknown.

\paragraph{Fast algorithm for Gaussian covariates}
In this paper, we focus on this basic setting when $X$ is Gaussian with unknown covariance $\Sigma$.
Under this setting, \cite{gao2020robust} showed that given ${O}(d/\eps^2)$ samples, one can in exponential time find an estimate achieving the information-theoretically optimal error ${O}(\sigma\epsilon)$.
The work by~\cite{diakonikolas2019efficient} gave a polynomial-time algorithm that uses $\widetilde{O}(d^2/\epsilon^2)$ samples and achieves a near-optimal error of ${O}(\sigma\epsilon\log(1/\eps))$. 
Subsequently, one line of work has focused on designing \emph{fast} robust regression algorithms under broader distributional assumptions that encompass Gaussians, such as bounded fourth-moment distributions \cite{diakonikolas2019sever, cherapanamjeri2020optimalrobustlinearregression, jambulapati2021robust}. 

When specialized to Gaussian covariates,~\cite{cherapanamjeri2020optimalrobustlinearregression} designed an algorithm that uses $n = \widetilde{O}(d/\eps)$ samples, runs in time $\widetilde{O}(nd\kappa/\epsilon^6)$, and achieves an error $O(\sigma\sqrt{\epsilon} \kappa)$, where $\kappa$ is a condition number of the covariance matrix $\Sigma$.
Note that with $n$ samples in $\R^d$, the size of the input is $nd$, so in the well-conditioned case, this algorithm runs in near-linear time in its input. 
However, their algorithm requires the additional assumption $\epsilon\kappa^2 \lesssim 1$ to succeed. 
Under the same assumption, ~\cite{jambulapati2021robust} improved the runtime to $\widetilde{O}(nd\sqrt{\kappa})$ while maintaining the sample complexity and error guarantee. %
In the same work, \cite{jambulapati2021robust} gave another algorithm with runtime $\widetilde{O}(nd\sqrt{\kappa}/\epsilon)$, while achieving a smaller error $O(\sigma\sqrt{\eps \kappa})$, under the \emph{milder} assumption $\eps \kappa \lesssim 1$, at the expense of increasing the sample complexity to $\widetilde{O}(d^2/\eps^3)$. A summary of the above discussion is in~\cref{tab:UpperBound}.

Therefore, it remains widely open to determine what the optimal trade-off is among the condition number $\kappa$, corruption rate $\eps$, sample complexity, and runtime.
In particular, \cite{jambulapati2021robust} asked the following question:
\begin{center}
  \emph{
  Does achieving error $O(\sigma\sqrt{\eps \kappa})$ under the milder assumption of $\eps \kappa \lesssim 1$ necessarily come at the cost of increased sample complexity for fast algorithms?
  }
\end{center}
Our first result answers this question for Gaussian covariates with unknown covariance. We show that there is an algorithm with runtime $\widetilde{O}(nd\sqrt{\kappa}/\epsilon)$ that achieves the sharper error bound of $O(\sigma\sqrt{\epsilon\kappa})$ under the milder condition $\epsilon\kappa \lesssim 1$, and requires only $\widetilde{O}(d/\eps^4)$ samples. 

\begin{theorem}[Fast, robust linear regression for Gaussians]\label{thm:GaussSampleInf}
Consider the setting in \cref{def:epscorr} where $D_X = \cN(0,\Sigma)$ with $\mu \cdot I_d \preceq \Sigma \preceq L \cdot I_d$, and the noise distribution has mean zero, variance $\sigma^2$, and fourth moment $O(\sigma^4)$. 
Let $\kappa \seteq L / \mu$. 
Then there is an algorithm that, for any $\eps \lesssim 1/\kappa$, given 
\[
    n \;=\; O\!\left(\frac{d\log d}{\epsilon^4} + \frac{d\log \frac{Ld}{\epsilon^7}}{\epsilon^2}\right)
\]
$\epsilon$-corrupted samples from $D_{XY}$, runs in time $\widetilde{O}(nd\sqrt{\kappa}/\eps)$, and returns an estimate $\widehat{\beta}$ satisfying $\|\widehat{\beta} - \beta\|_{\Sigma} \le O\big(\sigma\sqrt{\kappa\epsilon}\big)$ with probability at least $9/10$.  
\end{theorem}
 
We provide an overview of the proof of our result in~\cref{sec:techOverSampleComp}, and present the formal proof in~\cref{sec:UpperBound}.
\renewcommand{\arraystretch}{1.3} %

\begin{table}[h]
\centering
\caption{Selected algorithms for Robust Regression with Gaussian Covariates and Gaussian Noise}
\label{tab:UpperBound}
\small
\setlength{\tabcolsep}{3pt}
\begin{tabular}{|l|c|c|c|c|c|}
\hline
Paper  & Covariance & Samples & Running Time & Error & Assumptions \\ \hline

\cite{diakonikolas2019efficient} & Identity &
$\widetilde{O}(d/\eps^2)$ & $\mathrm{poly}(d)$ &
$\sigma\eps \log 1/\eps$ & --- \\ \hline

\cite{diakonikolas2019efficient} & Unknown &
$\widetilde{O}(d^2/\eps^2)$ & $\mathrm{poly}(d)$ &
$\sigma\eps \log 1/\eps$ & --- \\ \hline

\cite{anderson2025sample} & Unknown &
$\widetilde{O}(d^2/\eps^2)$ & $\mathrm{poly}(d)$ &
$\sigma\eps \log 1/\eps$ & --- \\ \hline

\cite{klivans2018efficient} & Unknown &
$\widetilde{O}(d^2/\eps^2)$ & $\mathrm{poly}(d)$ &
$\sigma\sqrt{\eps}$ & --- \\ \hline

\cite{bakshi2021robust} & Unknown &
$\widetilde{O}(d^2/\eps^2)$ & $\mathrm{poly}(d)$ &
$\sigma \epsilon^{3/4}$ & --- \\ \hline

\cite{cherapanamjeri2020optimalrobustlinearregression} & Unknown &
$\widetilde{O}(d/\eps)$ &
$\widetilde{O}(nd\kappa/\epsilon^6)$ &
$\sigma\sqrt{\epsilon}\kappa$ & $\epsilon\kappa^2 \lesssim 1$ \\ \hline

\cite{jambulapati2021robust} & Unknown &
$\widetilde{O}(d/\eps)$ &
$\widetilde{O}(nd\sqrt{\kappa})$ &
$\sigma\sqrt{\epsilon}\kappa$ & $\epsilon\kappa^2 \lesssim 1$ \\ \hline

\cite{jambulapati2021robust} & Unknown &
$\widetilde{O}(d^2/\eps^3 + d/\eps^4)$ &
$\widetilde{O}(nd\sqrt{\kappa}/\epsilon)$ &
$\sigma\sqrt{\epsilon \kappa}$ & $\epsilon\kappa \lesssim 1$ \\ \hline

{\hypersetup{linkcolor=red}%
\hyperref[thm:GaussSampleInf]{This work}}
 & Unknown &
$\widetilde{O}(d/\eps^4)$ &
$\widetilde{O}(nd\sqrt{\kappa}/\epsilon)$ &
$\sigma\sqrt{\epsilon \kappa}$  & $\epsilon\kappa \lesssim 1$\\ \hline

\end{tabular}
\end{table}

\paragraph{Impact of ill-conditioning on robust linear regression}
Comparing the results in \Cref{tab:UpperBound}, we observe two common limitations of all linear time algorithms, which are using $\widetilde{O}(d)$ samples:
\begin{enumerate}
    \item For fixed noise level $\sigma$ and corruption rate $\eps$, the estimation error grows with the condition number $\kappa$.
    \item The algorithms work only if the corruption rate is smaller than a threshold that depends on the condition number~--- otherwise, no meaningful guarantees are provided.
\end{enumerate}

These two limitations are not inherent to the problem; there exist polynomial-time algorithms that use $\Omega(d^2)$ samples and avoid either of these drawbacks, e.g., \cite{diakonikolas2019efficient, anderson2025sample}. 
One may justify that as long as there is any efficient robust covariance estimation algorithm
applicable when $n \gtrsim d^2$ samples, then there is no need for the error rate and tolerated fraction of corruptions of the algorithm to degrade with the condition number. One way to see this is that we can first use $\Theta(d^2)$ samples to robustly estimate the covariance matrix $\Sigma$, and use it to reduce to a well-conditioned instance ($\kappa = O(1)$). We can then run a robust linear-regression algorithm with $\kappa = O(1)$ (See~\cref{sec:covariancetoreg}).

We would like to point out an apparent parallel between the error term appearing in the algorithms in~\Cref{tab:UpperBound} and the state-of-the-art results for the related problem of {\em robust covariance-aware mean estimation}. For the latter problem, when the distribution of the underlying random variable is ill-conditioned ($\kappa = \omega(1)$), the error term grows proportionally to $\sqrt{\varepsilon \kappa}$ when $n = \widetilde{O}(d)$ but an efficient algorithm is nevertheless applicable, even when $\varepsilon$ is constant and $\kappa$ superconstant (see~\cref{sec:meanestimation}). \emph{A priori} it is not unreasonable to hope for a similar behavior in the robust linear regression.

\paragraph{Error rate}
A natural question to ask is whether fast algorithms can achieve error $o(\sigma\sqrt{\eps \kappa})$ when $\eps \kappa \lesssim 1$ while using $\widetilde{O}(d)$ samples. 
We provide a negative answer to this question in the Statistical Query (SQ) model. At a high level, our lower bound shows that any efficient SQ algorithms achieving error $o(\sigma\sqrt{\eps \kappa})$ must make either exponentially many queries, or at least one query which requires roughly $\Omega(d^2)$ samples to simulate (up to $\eps$, $\kappa$ factors).

\begin{theorem}[SQ Lower Bound (informal, see~\cref{thm:SQFull})]\label{thm:MainSQ}
Let $\eps \kappa \lesssim 1$. For robust linear regression with Gaussian covariates of unknown covariance $\Sigma$ with condition number $\kappa$ 
and unknown label noise variance $\sigma^2 \leq 1$, no efficient Statistical Query algorithm can achieve error $o(\sqrt{\eps \kappa})$ on all instances. In particular, any such algorithm must either make at least $2^{d^{\Omega(1)}}$ queries, or make at least one query that requires $\Omega\left(d^2 / 
\left(\sqrt{\kappa} e^{O(1/\eps)}\right)\right)$ samples to be simulated.
\end{theorem}

Our result can be viewed as a generalization of the SQ lower bound of \cite{diakonikolas2019efficient} to the setting in which the condition number $\kappa$ may be arbitrary. In particular,~\cite{diakonikolas2019efficient} show that when $\kappa = O(1)$, no efficient SQ algorithm can achieve prediction error $o(\sqrt{\eps})$ unless it uses queries of accuracy that require $\Omega(d^2)$ samples to simulate. When $\kappa = O(1)$, our lower bound recovers their result as a special case, in both the achievable error and query tolerance. Similar to \cite{diakonikolas2019efficient}, our lower bound yields meaningful quantitative guarantees for $\eps = \Omega(1/\log d)$, and mildly growing values of $\kappa$. We present an overview of our proof in~\cref{sec:techoverviewSQ} and complete formal proofs in~\cref{sec:SQLB}.

Moving further, we now turn to a more fine-grained study of the trade-off between sample complexity and error guarantees. Building on the question posed by \cite{jambulapati2021robust}, we study our problem as the sample size varies from $n = \widetilde{O}(d)$ to $n = \widetilde{O}(d^2)$ (See ~\cref{sec:TechOverLowDeg} and \cref{sec:LowDegree}). 

\paragraph{Corruption rate and condition number.}
As noted earlier, when $\widetilde{O}(d^2)$ samples are available, there exist efficient algorithms that achieve near-optimal prediction error $\widetilde{O}(\eps)$ whenever $\eps$ is at most a sufficiently small constant, \emph{regardless} of $\kappa$~\cite{diakonikolas2019efficient, anderson2025sample}. 
In contrast, known fast algorithms that use $\widetilde{O}(d)$ samples---including our result~\cref{thm:GaussSampleInf}---impose much stronger restrictions on the corruption rate $\eps$.
Specifically, they all require $\varepsilon$ to be smaller than a quantity that depends on $\kappa$, ranging from $\varepsilon \lesssim 1/\kappa^2$ to $\varepsilon \lesssim 1/\kappa$ for the stronger results.

At a high level, the $\kappa$-dependence comes from the fact that these
methods avoid robustly preconditioning the covariates. Robust gradient
descent methods, such as~\cite{cherapanamjeri2020optimalrobustlinearregression},
use robust mean estimation at the current iterate $\beta_t$ on the
corrupted gradient samples $(\langle X_i,\beta_t\rangle-y_i)X_i$, to estimate the population gradient $\Sigma(\beta_t-\beta)$. The robust
mean estimation guarantee gives an additive error of size roughly
$\sqrt{\eps}(\Vert \beta_t-\beta \Vert_2)$ in Euclidean norm.
However, the descent step needs this error to be small compared to the
true gradient $\Sigma(\beta_t-\beta)$. After normalizing
$\Vert \Sigma \Vert_2\leq 1$, we have
\[
  \Vert \beta_t-\beta \Vert_2 \leq \Vert \Sigma^{-1} \Sigma (\beta_t - \beta ) \Vert 
  \leq
  \kappa \Vert \Sigma(\beta_t-\beta) \Vert_2.
\]
Thus the relative error in the robust gradient estimate can be as large as
$\sqrt{\eps}\kappa$. For robust gradient descent to converge, it has to absorb the bias and and can thus guarantee convergence under a condition
of the form $\sqrt{\eps}\kappa \lesssim 1$, equivalently
$\eps \lesssim 1/\kappa^2$. 
The natural geometry for robust regression is the Mahalanobis norm
induced by $\Sigma$. Inspired by the identifiability proof
of~\cite{bakshi2021robust}, \cite{jambulapati2021robust} do their analysis directly in the Mahalanobis norm, and as a result are able to achieve a stronger condition of $\epsilon \lesssim 1/\kappa$.

Several existing algorithms that avoid joint assumptions on $\eps$ and $\kappa$, either implicitly or explicitly, reduce the problem to a well-conditioned setting, i.e., to an instance with condition number $O(1)$, via a robust preconditioning step. 
This preconditioner is obtained via robust covariance estimation in the relative spectral norm (see, e.g.,~\cite{kothari2018robust}).
When $n = O(d^{2 - \Omega(1)})$,~\cite{diakonikolas2017statistical} showed that for $\eps = \Omega(1/\log d)$, robust covariance estimation in this norm is computationally hard in the SQ model.

More recently, \cite{diakonikolas2025sos} gave an efficient robust covariance estimation algorithm using $\widetilde{O}\left(\eps^2 d^{2 + 2\delta} + d^{1 + \delta} \right)$  samples for any constant $\delta > 0$, which can be significantly below $d^2$ when $\eps$ is small (e.g. $\eps \approx d^{- 0.3}$). They further complemented this result with a low-degree lower bound, suggesting that $\tilde{\Omega}(\eps^2 d^2)$ samples might be necessary for efficient covariance estimation in the relative spectral norm. 

To summarize, existing efficient algorithms that achieve bounded prediction error without joint assumptions on $\eps$ and $\kappa$ appear to require at least $\tilde{\Omega}(\eps^2 d^2)$ samples.
In contrast, existing efficient algorithms that use $\widetilde{O}(d)$ samples can only achieve bounded error when $\eps \kappa \lesssim 1$.
However, neither existing algorithmic results nor current computational lower bounds shed light on whether the joint dependence on $\eps$ and $\kappa$ is inherent for efficient algorithms operating in the regime $n = o(\eps^2 d^2)$. 
Indeed, without such an assumption, we are not aware of \emph{any} efficient algorithm for Gaussian robust regression that provably outperforms the trivial estimator \(\widehat{\beta}=0\). This motivates the following question:
\begin{center}
  \emph{
    Is the condition \(\eps \kappa \lesssim 1\) merely an artifact of existing algorithmic techniques, or is it inherent to efficient algorithms in the low-sample regime \(n = o(\eps^2 d^2)\)?
  }
\end{center}
In this paper, we provide the first formal evidence that such a joint assumption may be necessary for efficient algorithms in the regime $n = o(\eps^2 d^2)$.
Specifically, we establish lower bounds against low-degree polynomial tests, showing that without a joint assumption on $\eps$ and $\kappa$, efficient algorithms for robust regression may be unable to outperform the trivial estimator significantly unless they use $\tilde{\Omega}(\min\{ d \eps^2 \kappa^2, \eps^2 d^2\})$ samples.
We show this by a reduction from the following hypothesis testing problem with input $\{(X_i, Y_i)\}_{i=1}^n$  (see~\cref{prob:mainlineartesting} for a formal statement). The problem requires distinguishing between the following two distributions, the null $H_0$: $\{(X_i, Y_i)\}_{i=1}^n$ sampled i.i.d. from $\cN(0, I_d) \times \cN(0, \sigma_y^2)$, and the alternate $H_1$: For $v$ sampled uniformly from the unit sphere, $\{(X_i, Y_i)\}_{i=1}^n$ sampled i.i.d. from $(1 - \eps) D_v(X, Y) + \varepsilon E_v(X, Y)$ where $E_v(X,Y)$ is arbitrary, and $D_v(X, Y)$ is the following linear model -- 
    $X \sim \cN(0, \Sigma_v)$, 
    $\eta \sim \cN(0,\sigma^2)$, and
    $Y = \langle X, \beta_v \rangle + \eta$ where $\beta_v, \Sigma_v$ depend on $v$.
We remark that under the null hypothesis $H_0$, there is no linear relationship between $X$ and $Y$, and the alternative hypothesis $H_1$ is an instance of robust linear regression under the weaker Huber contamination model.

\begin{theorem}[Informal, see~\cref{thm:mainlb}]\label{thm:LowDegreeMainInf}
Let $\epsilon \gg \frac{1}{\sqrt{d}}$, $\kappa \geq 1$ and $\epsilon \kappa \geq C$ for some constant $C > 0$ sufficiently large. Then, there exists a choice of $\beta_v, \Sigma_v, E_v(X, Y), \sigma_y^2$ such that the null and alternative above are indistinguishable for polynomials of degree $\mathrm{poly}(\log n)$  
for 
\[
    n \ll \frac{1}{\mathrm{poly}( \log d)}\min \left(d \epsilon^2 \kappa^2, \epsilon^2 d^2  \right).
\]

\end{theorem}
On the other hand, we show in~\cref{corollary:hardnessofestimation} that if the true parameter has signal strength $\Vert \Sigma^{1/2} \beta \Vert = \alpha$ for $\alpha = \Omega(1)$, an efficient regression algorithm achieving prediction error $0.1 \alpha$ solves the testing problem. This provides evidence that when $n \ll \min \left(d \epsilon^2 \kappa^2, \epsilon^2 d^2  \right)$, outperforming the trivial estimator by more than constant factors might be computationally hard whenever $\eps \kappa \gg 1$.

A concrete setting of parameters for which this bound is easily interpretable is to focus on the case where $\varepsilon, \delta$ are small constants (for example $0.01$), and $\kappa$ is moderately growing with the dimension, e.g., $\kappa = d^{\delta}$.
In this case, one could \emph{a priori} hope for an algorithm using $\widetilde{O}(d)$ samples, and returning an estimate with error $\sqrt{ \eps \kappa} \approx d^{\delta/2}$. Our lower bound suggests that, unless one uses $\widetilde{\Omega}(d^{1+2\delta})$ samples, this problem might be computationally hard.
Our computations in~\cref{sec:LowDegree} illustrate that the hardness is driven both by ill-conditioning and the underlying linear dependence of $X$ and $Y$.

We remark that when $\eps \kappa = \omega(1)$ and $\kappa \leq \sqrt{d}$, our lower bound does not rule out the existence of a low-degree polynomial algorithm that uses $o(\eps^2 d^2)$ samples and achieves small error. However, we are not aware of efficient algorithms that use $n = \widetilde{O}(d \eps^2 \kappa^2)$ samples and achieve error even of the order of $\sqrt{\epsilon \kappa}$. We believe that it is an interesting open problem to understand this gap of whether there exists an efficient algorithm matching our lower bounds in the regime where $\kappa \leq \sqrt{d}$.
We describe our hard instance in detail in~\cref{sec:TechOverLowDeg}. Our lower bound also has consequences for {\em differentially private} regression. In the regime $n=O(d^{2-\Omega(1)})$, it suggests an information-computation gap for efficient private algorithms, consistent with the best-known efficient algorithms in this regime \cite{brown2024insufficient}. We refer to~\cref{sec:privateReg} for details.

\section{Related Work}
\paragraph{Robust regression}
Beyond the algorithms described earlier, robust regression with unknown covariance has been studied under a variety of additional settings. \cite{oliveira2022spectral} studied it in this setting with noise that could depend on the covariates and obtained results that are closely comparable to those of~\cite{cherapanamjeri2020optimalrobustlinearregression}. ~\cite{depersin2020spectral} studied the case of known covariance and possibly dependent noise, and designed a spectral algorithm with a sub-Gaussian error rate and $O(\sqrt{\eps})$ error under bounded fourth moment assumptions. \cite{pensia2025robust} studied the problem in the well-conditioned setting. An important line of work based on the Sum-of-Squares hierarchy studies robust regression under \emph{certifiable} moment-bounded distributions~\cite{klivans2018efficient, bakshi2021robust, zhu2022robust}, achieving information-theoretically optimal error guarantees under broader distributional assumptions at the cost of increased sample complexity and runtime. Another widely studied direction considers milder corruption models in which the covariates remain uncorrupted while the responses are corrupted. The responses may be corrupted adaptively or non-adaptively, and a non-exhaustive list of papers that study this model includes~\cite{bhatia2015robust, bhatia2017consistent, suggala2019adaptive, dorsi2021consistent, chen2022online}. In some of these settings, consistent estimation is information-theoretically possible even when the fraction of corrupted responses approaches one, since the covariates remain uncorrupted and thus preserve sufficient structure, which is in sharp contrast to the strong contamination model where any procedure has to incur error scaling with $\eps$. There is also work studying robust regression with both covariate and response corruption under weaker models, such as Huber’s contamination; see, e.g.,~\cite{diakonikolas2023near}. 

\paragraph{Lower bounds for statistical problems} 
Standard approaches to proving computational lower bounds for average-case statistical problems establish hardness against restricted classes of algorithms, such as Statistical Query algorithms~\cite{kearns1998efficient, feldman2017statistical}, low-degree polynomials~\cite{ hopkins2017efficient, hopkins2017power, hopkins2018statistical}, which capture a large class of spectral methods, or the Sum-of-Squares hierarchy~\cite{barak2019nearly}. %
A complementary line of work, which has received recent attention, proves reduction-based hardness, similar to classical computational complexity theory, via reductions from problems believed to be computationally hard; see e.g.~\cite{brennan2020reducibility, bruna2021continuous}.

\paragraph{Algorithmic robust statistics}
Since the breakthrough works of \cite{lai2016agnostic, diakonikolas2019robust} that designed efficient algorithms to robustly estimate the mean and covariance of a Gaussian, there has been a plethora of work that has designed efficient estimators for a large number of related problems such as moment estimation, clustering and regression \cite{diakonikolas2018list, kothari2018robust, hopkins2018mixture, klivans2018efficient}. Earlier works also studied efficient learning under outliers~\cite{klivans2009learning}. We refer the reader to \cite{diakonikolas2023algorithmic} for a comprehensive overview of the developments.

\section{Technical Overview}
In this section, we provide a detailed overview of our proof techniques. We begin by discussing our sample complexity result (\cref{thm:GaussSampleInf}), and then turn to our SQ lower bound (\cref{thm:MainSQ}) and low-degree lower bound (\cref{thm:LowDegreeMainInf}).

\subsection{Improved Sample Complexity for Gaussians}\label{sec:techOverSampleComp}

In this section, we outline the main ideas and techniques used to prove~\cref{thm:GaussSampleInf}. Our algorithm is the same as that of~\cite{jambulapati2021robust}, and our result follows by improving their analysis for the case of Gaussian distributions. 

They propose an alternating algorithm that alternates between updating the regression parameter (ERM step), and removing outliers (filtering step). Specifically, the algorithm of \cite{jambulapati2021robust} is able to achiever error $O(\sigma \sqrt{\eps\kappa})$, improving on \cite{cherapanamjeri2020optimalrobustlinearregression}'s error of $O(\sigma\sqrt{\eps}\kappa)$ under the milder condition $\eps \kappa \lesssim 1$.
To achieve this smaller error, their algorithm requires stronger regularity conditions for the filtering step to succeed, necessitating $d^2$ samples. The following statement is the only bottleneck in their approach, which leads to $\Omega(d^2)$ sample complexity.

\begin{lemma}\cite[Lemma 19]{jambulapati2021robust}\label{lemma:JLSTnew}
    Let $\eps > 0$ be sufficiently small.
	Let $X_1, \ldots, X_n$ be $n$ samples from a $2$-to-$4$ hypercontractive distribution $\D$ with parameter $C$ and second moment $\Sigma$ with $\mu I_d \preceq \Sigma \preceq L I_d$.
	Then, there exist universal constants $c, C_{est} > 0$ so that if 
	\[
	n \geq 	c \left( \frac{d \log d}{\eps^4} + 	 \frac{ d^2 \log (d / \eps)}{\eps^3} \right),
	\] then with probability $0.99$, for every $u \in \R^d$, there exists an $G_u \subseteq [n]$ satisfying $|G_u| \geq (1 - \eps^2) n$,  and
	\begin{equation}
		\vnorm{\frac{1}{|G_u|} \sum_{i \in G_u} \Iprod{X_i, u}^2 X_i X_i^\top}_{\textup{op}} \leq C_{est} L \norm{u}^2_{\Sigma} \; .
        \label{eq:concentration}
	\end{equation}
\end{lemma}

In this work, we prove a similar statement under a weaker sample complexity requirement, $n \gtrsim d \log d / \varepsilon^{4}$, when $\mathcal{D}$ is a Gaussian distribution (\cref{thm:MainUB}). To highlight our core ideas, we sketch the proof in the simpler case where $\Sigma = I$.

We proceed via the standard argument of applying a union bound over a $\delta$-net. Our aim is to show that for a fixed $u$ in the $\delta$-net with $\|u\| = 1$, inequality~\eqref{eq:concentration} fails with probability $\exp(-\tilde{\Omega}(d))$. For such a fixed $u$, we choose the set $G_u$ to roughly consist of all samples $X_i$ for which $\|X_i\| \lesssim O(d \varepsilon^{-2})$, and $\langle X_i, u\rangle^2 \leq O(1/\varepsilon)$. With this choice of $G_u$ we first show that with high probability over the selection of $X$, for all $u$ simultaneously, it holds that $|G_u| \geq (1-\varepsilon^2) n$ (\cref{lem:SmallNorm,lem:UnivariateProj}).

In order to prove the spectral norm bound \eqref{eq:concentration}, we again apply a net argument coupled with a union bound, i.e., for every fixed $u, v$ (with unit norm), we are required to show that \begin{equation}
    \frac{1}{|G_u|}\sum_{i \in G_u} \Iprod{X_i, u}^2 \Iprod{X_i, v}^2 \lesssim 1,
    \label{eq:concentration2}
\end{equation}
fails with probability at most $\exp(-\tilde{\Omega}(n))$. When $n \gtrsim d$, this allows us to take a union bound over all $\exp(\tilde{O}(d))$ points in the appropriate net over $v$'s.

Since for every $u$, the set $|G_u|$ is large with high probability, Eq.~\eqref{eq:concentration2} can be upper bounded as
\begin{equation}
    \frac{1}{|G_u|}\sum_{i \in G_u} \Iprod{X_i, u}^2 \Iprod{X_i, v}^2  \lesssim \frac{1}{n}\sum_{i \in [n]} \mathbb{1}[i \in G_u] \Iprod{X_i, u}^2 \Iprod{X_i, v}^2.
    \label{eq:concentration3}
\end{equation}
We formally prove how to upper bound the right-hand side of the above inequality in the general covariance case in~\cref{lemma:TruncatedMain}, which captures the core of our proof. Here, we proceed by decomposing $v = \alpha u + v'$, where $v' \perp u$. In general, we note that $v'$ may depend on $v$, which adds to the complexity. By such a decomposition, we have essentially reduced the task of upper-bounding the right-hand side of \eqref{eq:concentration3}  to the following two cases: either $v=u$ or $v \perp u$.

If we denote $Y_i := \mathbb{1}[i \in G_u] \langle X_i, u\rangle$, and $Z_i := \langle X_i, v'\rangle$, we need to show
\begin{equation*}
    \frac{1}{n} \sum Y_i^4 \lesssim 1, \quad \text{and}\quad \frac{1}{n} \sum Y_i^2 Z_i^2 \lesssim 1,
\end{equation*}
with probability $1 - \exp(-\tilde{\Omega}(n))$. In the above inequalities, $Y_i$ and $Z_i$ are independent, $Z_i$ are standard Gaussians, and $Y_i$ are $1$-subgaussian and bounded by $1/\varepsilon$ (by our choice of set $G_u$). Both of these statements now follow from standard Bernstein-type concentration inequalities, thus concluding the result. We refer the reader to~\cref{sec:UpperBound} for a complete proof.

At a high level, our approach differs from that of~\cite{jambulapati2021robust} as follows. For a fixed $u$, they apply a matrix Chernoff bound to obtain Eq.~\eqref{eq:concentration}, but as a result, the failure probability they obtain is only of the order $\exp(-\Omega(n/d))$ and their application of Chernoff is tight, even in the Gaussian case. In order for this probability to be upper bounded by $\exp(-\tilde{\Omega}(d))$ (which is necessary to handle union bound over the net), they need to pick $n \gtrsim d^2$.

\subsection{An Improved Statistical Query Lower Bound }\label{sec:techoverviewSQ}

In this section, we outline the main ideas and techniques used to prove~\cref{thm:MainSQ}. Our lower bound builds on the SQ lower bound for Gaussian robust regression established by~\cite{diakonikolas2019efficient}, who show that any SQ algorithm achieving error $o(\sqrt{\eps})$ must either make exponentially many queries or use tolerances so small that each query requires $n = \Omega(d^{2})$ samples to simulate.

One of the most well-studied hard problems for the SQ model (and a basis of many hardness reductions in this area) is Non-Gaussian Component Analysis (NGCA). Informally, the goal in the NGCA task is to distinguish between the standard Gaussian and a distribution that is a known distribution $A$ in some unknown hidden direction and the standard Gaussian in the orthogonal complement.

\cite{diakonikolas2017statistical} showed that for every distribution $A$ that matches the first $p$ moments of the Gaussian distribution, and does not have too large $\chi^2$ divergence from that standard Gaussian, the related NGCA instance is hard for SQ algorithms (where the hardness depends on the number of moments matched). Often, after carefully selecting the distribution $A$ to not only match the first few moments of a Gaussian, but also satisfy additional problem-specific properties, this NGCA problem can be used as a basis for reduction to show SQ hardness for other problems of interest.

In this work, following~\cite{diakonikolas2019efficient}, we \emph{do not} show the SQ hardness of the robust regression via a black-box reduction from NGCA. Instead, we construct a null distribution over $\R^d \times \R$, as well as a family of alternative distributions over $\R^d \times \R$ (one for each direction $v \in \cS^{d-1})$, such that, if one had access an accurate robust regression algorithm, one could easily distinguish the alternative from the null distribution. The proof strategy for indistinguishability in the SQ model is as follows.
\begin{itemize}
    \item We carefully pick a (univariate) distribution $R$ for the response variable $y$. The null distribution is just $\mathcal{N}(0, I) \times R$.
    \item For a given hidden direction $v$, and a value of response variable $y \in \mathbb{R}$, we prepare a conditional distribution of covariate $X$ in the alternative case, as $\Pr[X | Y = y]$. This distribution resembles the one seen in the NGCA problem: on the hyperplane orthogonal to the hidden direction, the distribution of $X$ is just the standard Gaussian. In the hidden direction $v$, it is a mixture between a Gaussian distribution and another distribution prepared such that the first three moments of the mixture match those of the standard Gaussian. The weights of the mixture depend on the value of $y$. The alternative distribution is then obtained by first sampling $y$ from the same distribution $R$, and then $X$ from the conditional distribution discussed.
    \item Integrating over $y$, we show that the distribution prepared this way is close in the TV-distance to a distribution $(X, \langle \beta, X\rangle + \eta)$ for Gaussian $X$, and $\beta$~--- a rescaling of $v$. As such, the prepared instance is a valid instance of a robust regression problem, and the hidden direction $v$ can be recovered, if we have a good estimate of $\beta$.
    \item Since the conditional distribution of $\mathbb{P}[X | Y = y]$ has the same structure as in NGCA, proving the relevant correlation inequalities needed to establish SQ lower bound for distinguishing the null and alternative distribution can be reduced (by integrating over $y$) to the same inequalities already known for the NGCA instance\footnote{In~\cref{sec:SQLB} we prove hardness for the search version of the problem. NGCA hardness for the testing version is also well established (see e.g.,~\cite[Proposition 8.14]{diakonikolas2023algorithmic}) and has been used in regression settings~\cite{ diakonikolas2025informationcomputation}. We use the testing version here for ease of exposition.}.
\end{itemize}

The main technical difficulty is then crafting the condition distribution $\Pr[X | Y = y]$. Consider first the joint distribution of $(X',Y)$ in the uncorrupted regression case, i.e. when $X \sim \mathcal{N}(0, \Sigma_v)$, and $Y = \langle X, \beta\rangle + \eta$, where $\beta = \alpha \sqrt{\eps} v$, $\Sigma_v = I_d - \gamma vv^T$ and $\sigma^2$ is chosen such that $\sigma_y^2 = 1$. In this case, the conditional distribution is
\[
[X|Y = y] \sim \cN\left(\alpha' \sqrt{\eps} y v, I_d - \gamma' vv^T\right),
\]
where $\alpha'$ and $\gamma'$ depend on $\alpha, \gamma$.
The conditional distribution of $[X | Y=y]$ which we attempt to create, should be a mixture of the Gaussian above with another ``corruption'' distribution (different from a standard Gaussian only in the direction parallel to $v$), such that the first three moments of the mixture matches that of standard Gaussian, and the weight of the corruption distribution in the mixture is small for typical $y$.

In the work of~\cite{diakonikolas2019efficient}, it was enough for them to pick $\alpha$ and $\gamma$ some fixed constants smaller than $1$, and as such the distribution  of $X$ (after projecting onto the non-trivial direction $v$) was simply $\cN(c'\sqrt{\eps}y, 2/3)$, whereas in our setting it is
\[
    \cN(\underbrace{c \sqrt{\epsilon}y}_{=: \mu_s}, \underbrace{1/\kappa - \bar{c}^2 \epsilon}_{=: \sigma^2_s}),
\]
for some constants $c$ and $\bar{c}$.

Since $\kappa$ can be arbitrarily large, the variance $\sigma_s^2$ of the conditional distribution along $v$ can be arbitrarily small, as opposed to just $2/3$ as in the case of instance constructed in~\cite{diakonikolas2019efficient}. 
This qualitative difference necessitates a new moment-matching construction.
We next provide an overview of this construction; full details appear in~\cref{sec:MomentMatching}. To complete the analysis, we additionally show that the resulting corrupted distribution can be realized as a Huber contamination of the joint distribution $(X,Y)$ in~\cref{sec:SQmarginal}, and derive the required $\chi^2$-divergence bounds in~\cref{sec:CorrelationBound}. We omit these details here and focus on the moment-matching construction.

An important observation by~\cite{diakonikolas2019efficient} is that since $y\in \mathbb{R}$, and as a result $\mu_s$ can take any value, the corruption rate $\epsilon_{\mu_s} \in (0,1)$ must depend on $\mu_s$. We split our construction into three regimes as opposed to the four regimes in~\cite{diakonikolas2019efficient}:

\begin{enumerate}[(i)]
\item $|\mu_s| \leq O(\sqrt{\eps})$: We adapt a recent construction of \cite{diakonikolas2025sos} (See~\cref{lem:smallEps}) to our setting. 
Their construction applies to the regime of small $|\mu_s|$, and uses a mixture of four unit-variance Gaussians in the noise component, with component means of magnitude at most $O(1/\sqrt{\eps})$.

\item $\Omega(\sqrt{\eps}) \leq |\mu_s| \leq O(1)$: This is the most delicate regime to handle. In~\cite{diakonikolas2019efficient}, for constant values of $\mu_s$, the corruption rate $\eps_{\mu_s}$ is small. If we adopt a similar relationship between $\mu_s$ and $\eps_{\mu_s}$, then the contribution of the signal component -- whose variance is $\sigma_s^2$ and mixing weight is $(1 - \eps_{\mu_s})$-- to the second moment can become arbitrarily close to zero. At the same time, this component carries the largest mixing weight.
A natural way to compensate is to increase the variance of the noise by introducing Gaussian components with variance $O(1/\eps_{\mu_s})$. However, this approach is incompatible with controlling the $\chi^2$-divergence with respect to the standard Gaussian, which is finite only when the variance remains bounded by a constant (see~\cref{fact:chi2gaussian}). 
As a result, our construction requires introducing additional dependencies between $\mu_s$ and $\sigma^2_s$. 

\item $|\mu_s| \geq \Omega(1)$: 
One can utilize a natural approach of considering a mixture of three Gaussians: one component with mean $0$ and a large constant variance, and two symmetric Gaussian components with equal weights $(1 - \eps_{\mu_s})$ and parameters $(\mu_s, \sigma_s^2)$ and $(-\mu_s, \sigma_s^2)$. While this already suffices for our results, we provide a refined construction that maintains the variances of the noise components a constant instead of letting them be arbitrarily close to zero.

\end{enumerate}
This concludes our construction generalizing~\cite{diakonikolas2019efficient} from the constant-condition-number setting $\kappa = O(1)$ to the general regime $\kappa = \Omega(1)$.

\subsection{A New Information-Computation Trade-Off}\label{sec:TechOverLowDeg}

We finally outline the ideas used in the proof of~\cref{thm:LowDegreeMainInf}. We first present a natural reduction from NGCA to robust linear regression that isolates the source of hardness. This reduction, however, is insufficient for our main result and motivates the stronger construction that follows.

\paragraph{A first reduction}
We begin with the following testing problem.
Let $\cQ \coloneqq \cN(0, I_d)$ be the null hypothesis. To define the alternative $\cP$, draw $v \sim \mathrm{Unif}(\cS^{d-1})$ and define $\cP \coloneqq (1 - \eps) \cN(0, I_d - vv^T) + \eps E_v$, for a specific $E_v$. Given samples from either $\cP$ or $\cQ$, the goal is to efficiently distinguish between them. The distribution $E_v$ can be chosen so that $\cP$ matches the first three moments of $\cQ$, for instance, 
$E_v$ may be taken to have variance $1/\epsilon$ in the direction $v$ and the standard Gaussian in the orthogonal complement, yielding a special instance of the NGCA problem. For this instance, there is evidence against low-degree polynomials that indicates distinguishing requires $\tilde{\Omega}(\eps^2 d^2)$ samples~\cite[Theorem 4.5]{mao2025optimal}. Our first step is to reduce this testing problem to robust linear regression. While related direct reductions from spiked models (these are essentially distributions of the type $\cN(0, I + \theta vv^T)$) to regression are known, they have been intensely studied in the sparse setting \cite{bresler2018sparse, brennan2020reducibility, buhai2024computational, kelner2024lasso}. We reduce the above NGCA instance to robust linear regression as follows. Let $Z \sim \cN(0, I + \theta vv^T)$ for $\theta = -1$, and write $Z = (Y,X)$ where $Y = Z_1$ and $X = Z_{-1}$ (i.e., the vector obtained by dropping the first coordinate). Therefore,
\[
X \sim \cN\left(0, I_{d-1} -  v_{-1} v_{-1}^\top\right),
\]
and 
\[
Y \mid X = x \sim \cN\left(  -\frac{1}{v_1} \cdot \langle x, v_{-1} \rangle, 0\right).
\]
Consequently, the distribution of $Z$ admits two equivalent interpretations:
\begin{enumerate}
    \item a negative spiked model $\cN(0, I_d - vv^T)$,
    \item a linear regression model with $X \sim \cN\left(0, I_{d-1} - (1 - 1/\kappa) ww^\top\right)$, $Y = \langle X, \beta \rangle$,  where ${w = v_{-1}/ \vnorm{v_{-1}}}$, $\kappa = 1/v_1^2$, and $\beta = -\frac{1}{v_1} v_{-1}$.
\end{enumerate}
Routine calculation shows that any algorithm for Gaussian robust regression achieving non-trivial prediction error can solve the above testing problem by computing the norm of its estimator. 
As a result, by using the connection between spiked models and robust regression, we have shown evidence suggesting that unless an efficient algorithm uses $\tilde{\Omega}(\eps^2d^2)$ samples, it must incur a small constant prediction error.

We pause to make several useful observations about the above construction. (i) In contrast to our SQ lower bound, this reduction does not rely on any joint constraint~--- such as $\varepsilon \kappa \lesssim 1$~--- between the corruption rate and the condition number.
(ii) The signal strength $\norm{\beta}_\Sigma$ is inherently limited, as it is coupled to the coordinates of a random vector and the constraints that moment-matching imposes.
(iii) Likewise, the condition number $\kappa$ is tied to the coordinates of a random vector, restricting the construction to $\kappa =\Theta(d)$ in the typical case.
(iv) Finally, the construction falls short of saying much about requiring joint assumptions concerning $\eps$ and $\kappa$.

\paragraph{Our lower bound instance}
Building on the limitations of the previous construction, we make three key changes. First, we modify the null distribution so that the signal strength in the Mahalanobis norm can arbitrarily grow with the dimension. Second, instead of restricting to spiked models, we directly corrupt the joint distribution of $(X, Y)$ in the Huber model by mixing with a $d+1$ dimensional Gaussian.
Most importantly, we design the corruption to act in a two-dimensional subspace. Lifting to two dimensions allows us to explicitly encode the linear relationship between $X$ and $Y$ through the covariance, bringing out the coupling between $\eps$ and $\kappa$ while preserving the $\tilde{\Omega}(\eps^2 d^2)$ lower bound. We now describe the resulting testing problem and hard instance.
    \begin{enumerate}
        \item $\cQ:$ The joint distribution over $(X, Y)$ where $X \sim  \cN(0, I_d)$ and $Y \sim  \cN(0, \sigma_y^2)$ are independent.
        \item $\cP$: $ 
        (1 - \epsilon) \cdot D_v(X, Y) + \epsilon \cdot E_v(X, Y)$, where  
        $v \sim \text{Unif}(\cS^{d-1})$, $D_v(X, Y)$ is the distribution of $(X, Y)$ for $Y = \langle X, \beta_v \rangle + \eta$, $X \sim \cN(0, \Sigma_v)$, and $\eta \sim N(0, \sigma^2)$ independent of $X$. 
        \end{enumerate}
For the hard instance, we set \(\Sigma_v = I_d - (1 - 1/\kappa) vv^\top\), \(\beta_v = \delta v\), and take \(E_v(X,Y)\) to be a mean-zero Gaussian distribution in \(\mathbb{R}^{d+1}\). 
Under this choice, the alternative distribution \(\mathcal P\) is given by,
\begin{align*}
    \cP(X, Y) = (1 - \epsilon) \cdot 
    \cN\left(0, \begin{bmatrix}
        I_d - (1 - 1/\kappa)vv^T & \frac{\delta}{\kappa} v \\ 
    \frac{\delta}{\kappa} v^T & \sigma_y^2
    \end{bmatrix}\right) + \epsilon \cdot \cN(0, \Sigma_E).
\end{align*}
We choose $\Sigma_E$ so that the first three moments of $\cP$ match those of the null distribution $\cQ$. Therefore, $\Sigma_E$ is set to be
\[
   \begin{bmatrix}
         I_d + \frac{(1-\epsilon)}{\epsilon} \cdot (1 - 1/\kappa)vv^T & -\frac{(1-\epsilon)}{\epsilon} \frac{\delta}{\kappa}v  \\ 
    -\frac{(1-\epsilon)}{\epsilon} \frac{\delta}{\kappa}v^T  & \sigma_y^2
\end{bmatrix}.
\]
 Observe that $\Sigma_E$ is the identity matrix on the $(d-1)$-dimensional subspace orthogonal to both $v$ and $Y$; all corruption is therefore confined to the two-dimensional subspace spanned by $v$ and the last coordinate $Y$. 
This matrix is positive semidefinite if and only if $\eps \kappa \geq  1 - \eps$ (\cref{lemma:psdness}).

We now examine the distribution $\cP$ more closely. In the limit when $\kappa \to \infty$, the off-diagonal covariance terms vanish, and one can observe that distinguishing $\cP$ and $\cQ$ is equivalent to distinguishing the $X$ part since the $Y$ part is independent of $X$. Note that this is equivalent to the distinguishing problem in our first reduction as the $X$ part in our alternative has the form $(1 - \eps) \cN(0, I - vv^T) + \eps E_v$. 
In the general case, the off-diagonal terms encode the correlation between $X$ and $Y$. This correlation term plays a central role in our lower bound. After appropriate calculations, the diagonal covariance terms give rise to the $\tilde{\Omega}(\eps^2 d^2)$ term, corresponding to the embedded NGCA hardness. In contrast, the off-diagonal block---responsible for the linear dependence between $X$ and $Y$ through $\beta_v$---yields the additional $\tilde{\Omega}(d \eps^2 \kappa^2)$ term. This latter contribution precisely provides evidence for the necessity of joint assumptions on $\eps$ and $\kappa$ in robust regression when $n = o(\eps^2 d^2)$ for algorithms that use $n = \widetilde{O}(d)$ samples. We note that this construction subsumes the earlier construction in the following sense: setting $\kappa = \Theta(d)$ in our construction yields a sample lower bound of $\tilde{\Omega}(\eps^2 d^2)$.

We remark that a related construction appears in reduction-based lower bounds for robust sparse linear regression \cite{brennan2020reducibility}; however our setting and objective are different. To establish our results formally, we show in~\cref{thm:mainlb} that the advantage of degree-$O(\log n)$ tests is $1 + o(1)$ unless 
\[
n = \tilde{\Omega}\left(\min\{d \eps^2 \kappa^2, \eps^2 d^2 \}\right)
\]
samples are used. Our analysis follows the approach of \cite{mao2025optimal} by expressing the distinguishing advantage in terms of Hermite coefficients via a Fourier decomposition. Along the way, we prove a generalization of their result involving expectations of products of Hermite polynomials which may be of independent interest (see~\cref{lemma:generaladvantagetohermite}). The remaining arguments are combinatorial and we defer the computations to~\cref{sec:LowDegree}.

\section{Preliminaries}
\subsection{Organization}
In~\cref{sec:UpperBound} we present the proof of our improved analysis of \cite{jambulapati2021robust} resulting in~\cref{thm:GaussSampleInf}. In~\cref{sec:SQLB} we present the formal statement for our SQ lower bound and provide complete proofs. In~\cref{sec:LowDegree} we present our low-degree lower bounds formally and provide all accompanying proofs, along with our reduction. In~\cref{sec:privateReg} we discuss about the consequences of our low-degree lower bound in the context of differentially private regression. ~\cref{sec:covariancetoreg} discusses how preconditioning can get rid of the joint conditions on $\eps$ and $\kappa$ and~\cref{sec:meanestimation} discusses a fast algorithm for covariance-aware mean estimation. Finally~\cref{sec:code} provides a link to a repository where we provide code for verifying computations that arise as part of our SQ lower bound.

\subsection{Notation}
We use the following convention: $\N$ is the set of natural numbers and $\R$ is the set of real numbers. $\R^d$ is the set of real vectors in $d$ dimensions.  For a positive integer $n$, $[n]$ is the set $\{1, 2, \dots, n\}$. Unless explicitly stated, the base of the logarithm is $e$. Unless otherwise specified, all vector norms $(\Vert . \Vert$ or $\Vert . \Vert_2)$ are the euclidean norm, and all matrix norms are the spectral norm. We use the notation $O(.), \Theta(.), \Omega(.), \lesssim, \gtrsim $ to hide absolute constants. We use $\mathbb{1}[.]$ for the indicator variable. We use $\widetilde{O}(.), \tilde{\Omega}$ to hide logarithmic factors. We denote the identity matrix in $d$ dimensions by $I_d$. Let $A, B \in \R^{d \times d}$. Then $A \preccurlyeq B$ or $A \preceq B$ is the ordering of $A$ and $B$ in \Lowner order. We use RHS and LHS to refer to the right and the left of an inequality respectively. We use $\operatorname{diag} \left(I_d, a \right)$ for some scalar $a$ to denote the the diagonal matrix in $\R^{d+1}$ that has $1$ in the first $d$ coordinates and $a$ in the last coordinate. $f(n) \ll g(n)$ indicates that there exists a constant $C$ such that $f(n) \leq g(n) / (\log n)^C$.

\section{Improved Sample Complexity for Gaussian Distributions}\label{sec:UpperBound}

 As detailed in the technical overview, the core of our improvement is~\cref{lemma:TruncatedMain}. Before we state and prove our improvement, we will require a few important facts that we will utilize in the rest of this section. We will first state them.

\subsection{Concentration Inequalities}
\begin{fact}[Scalar Bernstein]\label{fact:scalarbernstein}
   
   Suppose $X_1, X_2, \dots, X_n$ are drawn iid from a distribution with mean 0 on the real line and each $|X_i| \leqslant M$. Then we have for all positive $t$
\begin{align*}
    \Pr\left[\sn X_i > t\right] \leqslant \exp\left(- \frac{n^2t^2}{2 \cdot \left( \sum_i \E[X_i^2] + \frac{1}{3} Mn t\right)}\right)
\end{align*} 

\end{fact}
\begin{fact}[Subexponential Tail Bounds \cite{vershynin2018high}]\label{subexponential}
Let
$    S = \sum_{i=1}^n a_i X_i,
   $
where each $X_i$ is an independent chi-squared random variable with one degree of freedom, i.e., $X_i \sim \chi^2_1$, and $\{a_i\}_{i=1}^n$ are fixed weights. The key observation is that if we define the centered variable
\[
Y_i = X_i - \mathbb{E}[X_i] = X_i - 1,
\]
then \(Y_i\) is a mean-zero subexponential random variable. Consequently, we can apply Bernstein's inequality to the centered sum:
\[
S_c = S - \mathbb{E}[S] = \sum_{i=1}^n a_i (X_i - 1).
\]
In particular, for any $t \geq 0$,
\[
\Pr\left( \left| \sum_{i=1}^n a_i (X_i - 1) \right| > t \right) \le 2 \exp\left( -c \cdot \min\left( \frac{t^2}{\sum_{i=1}^n a_i^2 V}, \frac{t}{\max_i |a_i| K} \right) \right)
\]
where $c$ is a universal positive constant, $V = \text{Var}(X_i)$ is the variance of a $\chi^2_1$ random variable and $K$ is a constant related to the sub-exponential norm of a centered $\chi^2_1$ random variable. 
\end{fact}

\begin{fact}(VC Inequality)\label{vcinequality}(\cite{shalev2014understanding})
Let $X_1, \dots, X_n \stackrel{\text{i.i.d.}}{\sim} \mathcal{D}$ be samples in $\R^d$. 
For each unit vector $u \in \R^d$, define
\[
S_{n,u} = \frac{1}{n}\sum_{i=1}^n \mathbb{1}\!\left[\langle X_i, u\rangle^2 \ge \frac{20}{\epsilon}\right],
\qquad 
p_u = \E[S_{n,u}] = \Pr\!\left[\langle X, u\rangle^2 \ge \frac{20}{\epsilon}\right].
\]
Let 
\[
\mathcal{H} = \bigl\{\, H_u = \{x \in \R^d : \langle x, u\rangle^2 \ge 20/\epsilon\} \;:\; \|u\|_\Sigma = 1 \bigr\}
\]
be the corresponding concept class, and let $S_{\mathcal{H}}(n)$ denote its growth function.
Then, for every $\alpha > 0$,
\[
\Pr\!\left[
\sup_{u} \big| S_{n,u} - p_u \big| > \alpha
\right]
\;\le\;
8\, S_{\mathcal{H}}(n)\, e^{-n\alpha^2/32}.
\]
\end{fact}

 We now state and prove our improved version of the result from~\cite{jambulapati2021robust}.

 \begin{theorem}\label{thm:MainUB}
    Let $\epsilon > 0$ be sufficiently small. Let $X_1, X_2, \dots X_n$ be $n$ samples drawn from a $d$-dimensional Gaussian, $D_X = \cN(0, \Sigma)$ where $\mu \cdot I_d \preceq \Sigma \preceq L \cdot I_d$. Then there exist universal constants $c, C_{est} > 0$ such that if 
    \begin{align*}
        n \geqslant c \cdot \left( \frac{d \log d}{\epsilon^4} + \frac{d \log \frac{ L d}{\epsilon^7}}{\epsilon^{2}}\right)
    \end{align*}
    we have with probability at least $0.9$ that for every $u \in \R^d$, there exists a $S_u \subseteq [n]$ satisfying $|S_u| \geqslant (1 - \epsilon^2) n$, and 
    \begin{align*}
        \left\Vert \frac{1}{|S_u|}\sum_{i \in S_u}\langle X_i, u \rangle^2 X_i X_i^T \right\Vert_{\text{op}} \leqslant C_{est}L\Vert u \Vert_{\Sigma}^2
    \end{align*}
\end{theorem}

 We first prove the above result assuming~\cref{lemma:TruncatedMain}, which we prove in~\cref{sec:ProofTruncatedMain}. The proof of~\cref{thm:GaussSampleInf} then follows similarly to Theorem 5 in ~\cite{jambulapati2021robust}, and we have included it in~\cref{sec:ThmUBProof}. 

\subsection{Proof of Theorem~\ref{thm:MainUB}}
To prove our result, we first note that it suffices to prove for $\Vert u \Vert_{\Sigma} = 1$. We first prove that most samples have a small norm.
\begin{lemma}\label{lem:SmallNorm}
    Let $X_1,\cdots,X_n$ be $n$ samples drawn from the distribution $D_X$. If $n\geq \Omega(1/\epsilon^2)$, then 
    \[
    \Pr\left[\left|\left\{ i : \Vert X_i \Vert_2^2 \geqslant \frac{20 L d}{\epsilon^2} \right\} \right| \leqslant \frac{n \epsilon^2}{10}\right] \geq 0.99.
    \]
\end{lemma}
\begin{proof}
      By Markov's inequality,
    \begin{align*}
        \Pr_{X \sim D_X} \left[ \Vert X \Vert_2^2 \geqslant \frac{20 L d}{\epsilon^2} \right] \leqslant \frac{\E_{X \sim D_X}\left[ \Vert X \Vert^2 \right]}{\frac{20 L d}{\epsilon^2}} \leqslant \frac{\epsilon^2}{20},
    \end{align*}
    where in the last inequality, we used that $\E_{X \sim D_X}[\Vert X \Vert^2] \leqslant Ld$.
    Now, by an application of Scalar Bernstein (~\cref{fact:scalarbernstein}) on the indicator random variable $\mathbb{1}\left[ \Vert X_i \Vert_2^2 \geqslant \frac{20 L d}{\epsilon^2} \right]$, we have with probability at least $0.999$, 
    \begin{align*}
        \left|\left\{ i : \Vert X_i \Vert_2^2 \geqslant \frac{20 L d}{\epsilon^2} \right\} \right| \leqslant \frac{n \epsilon^2}{10}
    \end{align*}
    whenever $n = \Omega(1/\epsilon^4)$, which concludes the proof.
    
\end{proof}
We next show that along any arbitrary direction $u$, the one-dimensional projections of the samples cannot be large. 
\begin{lemma}\label{lem:UnivariateProj}
    Let $X_1,\cdots,X_n$ be $n$ samples drawn from the distribution $D_X$, and  $u\in \mathbb{R}^d$ such that $\Vert u \Vert_{\Sigma} = 1$. Let $H_u$ be defined as,
    \[
    H_u := \left\{ X \in \R^d : \langle X,u \rangle^2 \geqslant \frac{20}{\epsilon}  \right\}.
    \]
    Then, if $n \geq \Omega(d\log d/\epsilon^4)$, 
    \[
    \Pr\left[\sup_{H_u} \frac{|i : X_i \in H_u|}{n} \leqslant \frac{\epsilon^2}{50}\right]\geq 0.99.
    \]
\end{lemma}
\begin{proof}
For any vector $u$ with $\Vert u \Vert_{\Sigma} = 1$, define $H_u$ to be the following set: 
\begin{align*}
    H_u := \left\{ X \in \R^d : \langle X,u \rangle^2 \geqslant \frac{20}{\epsilon}  \right\}
\end{align*}
Now for a fixed $u$, we have using Markov's inequality,
\begin{align*}
    \Pr_{X \sim D_X}\left[ \langle X, u \rangle^2 \geqslant \frac{20}{\epsilon}\right] &=  \Pr_{X \sim D_X}\left[ \langle X, u \rangle^4 \geqslant \frac{400}{\epsilon^2}\right] \leqslant \frac{ \epsilon^2}{400} \cdot \E_{X \sim D_X}[\langle X, u \rangle^4] \\
\end{align*}
Note that,
\begin{align*}
   \E_{X \sim D_X}[ \langle X, u \rangle^4] = \E_{Z \sim N(0, I_d)}[\langle Z, \Sigma^{1/2}u \rangle^4]  = \E_{Y \sim N(0, 1)}[Y^4] = 3
\end{align*}
Therefore, 
\begin{align*}
    \Pr_{X \sim D_X}\left[ \langle X, u \rangle^4 \geqslant \frac{20}{\epsilon}\right] \leqslant \frac{3 \epsilon^2}{400} \leqslant \frac{\epsilon^2}{100}
\end{align*}
Next, we look at the fraction of samples that have a large univariate projection along the fixed direction $u$. Define,
\begin{align*}
    S_{n, u} = \frac{1}{n}\sum_{i \in [n]} \mathbb{1}\left[\langle X_i, u \rangle^2 \geqslant \frac{20}{\epsilon}\right]
\end{align*}
We know that $\E[S_{n, u}] =  \Pr\left[ \langle X_i, u \rangle^2 \geqslant \frac{20}{\epsilon} \right] \leqslant \frac{ \epsilon^2}{100}$. 
We will now be interested in
\begin{align*}
    \Pr\left[ \sup_{u} \left|S_{n, u} - \E[S_{n, u}]  \right| > \frac{\epsilon^2}{100}    \right] 
\end{align*}
which by an application of the VC inequality (~\cref{vcinequality}) can be bounded as 
\begin{align*}
     \Pr\left[ \sup_{u} \left|S_{n, u} - \E[S_{n, u}]  \right| > \frac{\epsilon^2}{100}    \right] \leqslant 8 \cdot 
     \cS(H_u, n) \cdot e^{-\Omega(n\epsilon^4)}
\end{align*}
where $\cS$ is the shatter coefficient or the growth function. From Sauer's Lemma (\cite{sauer1972density}), we know that 
\begin{align*}
    \cS(H_u, n) \leqslant (n+1)^{VC(H_u)}
\end{align*}
where VC$(H_u)$ is the VC-dimension of the concept class $H_u$. In our specific case, the set $H_u$ is the union of two halfspaces. It is a well-known fact that union of finitely many halfspaces in $\R^d$ has VC dimension $O(d)$. Therefore we have that as long as 
\begin{align*}
    n \geqslant \Omega\left(\frac{d \log n}{\epsilon^4}\right) = \Omega\left(\frac{d \log d}{\epsilon^4}\right)
\end{align*}
that with probability at least $0.999$ that
\begin{align*}
    \sup_u S_{n, u} = \sup_{H_u} \frac{|i : X_i \in H_u|}{n} \leqslant \frac{\epsilon^2}{50},
\end{align*}
concluding the proof.
\end{proof}
The crux of our proof is the following lemma which we prove in~\cref{sec:ProofTruncatedMain}. 

\begin{restatable}{lemma}{mainlemmasample}\label{lemma:TruncatedMain}
    With probability at least $1 - \exp\left(-\Omega\left(d/\epsilon \cdot \log \left(d/\epsilon^7 \right) \right)\right)$ over the choice of $\{X_i\}_{i=1}^n$, we have that for a fixed $u$ with $\Vert u \Vert_{\Sigma} = 1$ that 
    \begin{align*}
        \sup_{v : \Vert v \Vert_2 = 1} \frac{1}{n}\sum_{i \in [n]} \langle X_i, u \rangle^2 \langle X_i, v \rangle^2 \mathbb{1}\left[\langle X_i, u \rangle^2 \leqslant \frac{20}{\epsilon}\right] \mathbb{1}\left[\Vert X_i \Vert^2 \leqslant \frac{20  L  d}{\epsilon^2}\right] \leqslant L \cdot C^{''}
    \end{align*}
    for an absolute constant $C^{''} > 0$.
\end{restatable}

We now prove~\cref{thm:MainUB}.

\begin{proof}[Proof of Theorem~\ref{thm:MainUB}]

We will use a $\gamma$-net argument to prove this theorem. Let $\cN_{\gamma}$ denote a $\gamma$-net on the ellipsoid $\Vert u \Vert_{\Sigma} = 1$. We first claim that, 
\[
| \cN_\gamma | \leqslant \left( \frac{3}{\gamma}\right)^d.
\]
This follows since $\Phi: \cE \to \cS^{d-1} $ defined as $\Phi(u) = \Sigma^{1/2} u $ where $\cE = \{u \in \R^d: \Vert u \Vert_\Sigma = 1\}$ is a bijection between the ellipsoid and the unit sphere. This implies that whenever $\Vert u - u' \Vert_{\Sigma} \leqslant \delta$, then correspondingly $\Vert \Phi(u) - \Phi(u') \Vert_2 \leqslant \delta$ as well since $\Phi$ is also an isometry between the metric spaces $(\R^d, \Vert . \Vert_{\Sigma})$ and $(\R^d, \Vert . \Vert_2)$.  Therefore, the number of elements in $\cE$ is the same as the number of elements in a $\gamma$-net over a sphere.

We first claim that it is sufficient to prove the result only for $u\in \cN_{\gamma}$ when $\gamma = \Theta(\epsilon^6/Ld^2)$. To prove this we assume the result holds for all $u\in \cN_{\gamma}$ and show that it then holds for all $u$ satisfying $\|u\|_{\Sigma} = 1$. To this end, we first define the set for all $u\in \mathbb{R}^d$, $\|u\|_{\Sigma} = 1$,
\begin{align*}
    G_u := \left\{ i : X_i \notin H_u \text{ and } \Vert X_i \Vert_2^2 \leqslant \frac{20Ld}{\epsilon^2} \right\}.
\end{align*}
Now, assume that,
for every $v \in \cN_{\gamma}$, $S_{v} =G_v\subseteq [n]$ satisfying $|S_{v}| \geqslant (1 - \epsilon^2) n$, and 
    \begin{align*}
        \left\Vert \frac{1}{|S_{v}|}\sum_{i \in S_v}\langle X_i, v \rangle^2 X_i X_i^T \right\Vert_{\text{op}} \leqslant C_{est}L.
    \end{align*}
Let $u\notin \cN_{\gamma}$, $\|u\|_{\Sigma} = 1$, and let $u'\in \cN_{\gamma}$ denote the point on the net closest to $u$, and define,
\[
    S_u := G_u \cap G_{u'}.
\]
Also, define the unnormalized sum, 
\[
F(u) := \sum_{i \in S_u}\langle X_i, u \rangle^2 X_i X_i^T,
\]
and note that we have assumed $\|F(u')\|_{\text{op}}\leq C_{est}L |S_{u'}| \leq C_{est}Ln$.
Now, from triangle inequality

\begin{align*}
    \vnorm{F(u)}_{\text{op}} = \vnorm{F(u) - F(u') + F(u')}_{\text{op}} \leq  \vnorm{F(u) - F(u')}_{\text{op}} + \vnorm{F(u')}_{\text{op}}.
\end{align*}
Therefore as long as 
\[
    \vnorm{F(u) - F(u')}_{\text{op}} \leq O(nL)
\]
we get that $\vnorm{F(u)}_{\text{op}} \leq O(nL)$ and are done. For the remainder of the proof we will use $\|\cdot\|$ to represent the operator norm. Now,
\begin{align*}
    \vnorm{F(u) - F(u')} &= \vnorm{\sum_{i \in S_u}\langle X_i, u \rangle^2 X_i X_i^T - \sum_{i \in S_{u'}}\langle X_i, u' \rangle^2 X_i X_i^T} \\
&\substack{(i)\\=} \vnorm{\sum_{i \in S_u}\left(\langle X_i, u \rangle^2 - \langle X_i, u' \rangle^2 \right)X_i X_i^T + \sum_{i \in S_{u'}\backslash S_u}\langle X_i, u' \rangle^2 X_i X_i^T} \\
& \substack{(ii)\\ \leq} \vnorm{\sum_{i \in S_u}\left(\langle X_i, u \rangle^2 - \langle X_i, u' \rangle^2 \right)X_i X_i^T} + \vnorm{\sum_{i \in S_{u'} \backslash S_u}\langle X_i, u' \rangle^2 X_i X_i^T}\\
&\substack{(iii)\\ \leq} \vnorm{\sum_{i \in S_u}\left(\langle X_i, u \rangle^2 - \langle X_i, u' \rangle^2 \right)X_i X_i^T} + \vnorm{\sum_{i \in S_{u'}}\langle X_i, u' \rangle^2 X_i X_i^T} \\ 
&\substack{(iv)\\ \leq} \vnorm{\sum_{i \in S_u}\left(\langle X_i, u \rangle^2 - \langle X_i, u' \rangle^2 \right)X_i X_i^T}  + O(nL).
\end{align*}
In the above $(i)$ follows from the definition of $S_u$ and $S_{u'}$, $(ii)$ is triangle inequality, $(iii)$ is because we are adding more PSD matrices, and $(iv)$ follows from the bound over the elements of the net.
Now, again from triangle inequality,
\[
\vnorm{\sum_{i \in S_u}\left(\langle X_i, u \rangle^2 - \langle X_i, u' \rangle^2 \right)X_i X_i^T} \leq \sum_{i \in S_u} \vnorm{\left(\langle X_i, u \rangle^2 - \langle X_i, u' \rangle^2 \right)X_i X_i^T}.
\]
Consider each term in the sum,
\begin{align*}
   \vnorm{\left(\langle X_i, u \rangle^2 - \langle X_i, u' \rangle^2 \right)X_i X_i^T}  &= \vnorm{\langle X_i, u - u' \rangle \langle X_i, u + u' \rangle X_i X_i^T} \leqslant \vnorm{X_i}_2^4 \cdot \vnorm{u - u'}_2 \cdot \vnorm{ u + u'}_2 \\
    & \leq  \Vert X_i \Vert_2^4 \cdot \Vert \Sigma^{-1/2}  \Vert_2^2 \cdot \Vert \Sigma^{1/2} (u - u') \Vert_2 \cdot \Vert \Sigma^{1/2} \left(u + u' \right)\Vert_2\\
    &\leq \frac{800 L^2 d^2}{\epsilon^4} \cdot \Vert \Sigma^{-1/2}\Vert_2^2 \cdot \gamma.
\end{align*}
Since $\gamma = \Theta(\epsilon^6/Ld^2)$ is small enough, and $\epsilon\kappa < 1$, we can bound the above by $O(\eps)$, and we have shown
\[
     \vnorm{F(u) - F(u')} \leq O(n(L+ \epsilon)).
\]
Since $F$ was unnormalized (instead of average we took the sum), we get after normalization, 
\[
    \frac{1}{n}\vnorm{F(u) - F(u')} \leq O(\eps+L) = O(L).
\]
Therefore for any vector $u$ with $\|u\|_{\Sigma} = 1$, we get that 
\begin{align*}
    \vnorm{\frac{1}{|S_u|}\sum_{i \in S_u }\langle X_i, u \rangle^2 X_i X_i^T} & = \frac{n}{|S_u|}\vnorm{\frac{1}{n}\sum_{i \in S_u }\langle X_i, u \rangle^2 X_i X_i^T} \\
    &= \frac{n}{|S_u|}\vnorm{\frac{F(u)}{n}} \leq O \left(\frac{n L}{|S_u|} \right) \leq O \left(\frac{n L}{n(1 - O(\eps^2))} \right) \leq O(L)
\end{align*}
which concludes that it is sufficient to prove our theorem for all the points on the net.

It remains to prove that $|G_u| \geq (1-O(\epsilon^2))n$ and that the result holds for all $u\in \cN_{\gamma}$. From~\cref{lem:SmallNorm} and~\cref{lem:UnivariateProj} we can conclude that for all $u\in \cN_{\gamma}$, $|G_u| \geq (1-O(\epsilon^2))n$ with probability at least $0.95$. Furthermore, for any $u\in \cN_{\gamma}$ from~\cref{lemma:TruncatedMain}, with probability at least $1 - \exp\left(-\Omega\left(d/\epsilon \cdot \log \left(d/\epsilon^7 \right) \right)\right)$, $\vnorm{\frac{1}{|S_u|} F(u)}\leq O(L)$. With a union bound over all the $O(\left(Ld^2/\epsilon^6\right)^d)$ vectors in the net, we get that with probability at least $0.95$ for $n$ as defined in the theorem, the result holds for all $u\in \cN_{\gamma}$. Therefore, a final union bound over these two events concludes the proof.

\end{proof}

\subsection{Proof of Lemma \ref{lemma:TruncatedMain}}\label{sec:ProofTruncatedMain}
We restate the lemma and give its proof.
\mainlemmasample*
\begin{proof}
Since $X_i \sim \cN(0, \Sigma)$, we have that $\langle X_i, u \rangle = \langle Z_i, \Sigma^{1/2} u \rangle$ where $Z_i \sim \cN(0, I_d)$. As a consequence, for $u' := \Sigma^{1/2} u, v' := \Sigma^{1/2}v$
\begin{align*}
    \langle X_i, u \rangle^2 \langle X_i, v \rangle^2 = \langle Z_i, u' \rangle^2 \langle Z_i, v' \rangle^2
\end{align*}
and $\|u'\|_2 = 1$. The vector $v'$ however norm $\Vert \Sigma^{1/2} v \Vert$ and we can define the unit vector $v'' = \frac{v'}{\Vert v' \Vert}$ which we can decompose along $u'$ and $u'_\perp$ where $u'_\perp$ is the unit vector perpendicular to $u$ for this choice of $v$\footnote{Note that $u_\perp'$ indeed depends on $v$. Formally $u_\perp' = (I - \frac{uu^T}{\|u\|^2})v$}. Therefore, we can write $v'' = \beta \cdot u' + \sqrt{1 - \beta^2} \cdot u_{\perp}$ for some $\beta \in [-1, 1]$ and we now have,
\begin{align*}
    \langle Z_i, v'' \rangle^2 = \langle Z_i, \beta \cdot u' + \sqrt{1 - \beta^2} \cdot u_{\perp} \rangle^2 &\leqslant 2 \cdot \beta^2 \cdot \langle Z_i, u' \rangle^2 + 2 \cdot (1 - \beta^2) \cdot \langle Z_i, u_\perp \rangle^2\\
    &\leqslant 2 \cdot \left( \langle Z_i, u' \rangle^2 + \langle Z_i, u_\perp \rangle^2 \right)
\end{align*}
where we used that $(a + b)^2 \leqslant 2a^2 + 2b^2 \ \forall a, b \in \R$. Therefore we have that
\begin{align*}
    \langle Z_i, u' \rangle^2 \langle Z_i, v' \rangle^2 &= \langle Z_i, u' \rangle^2 \langle Z_i, v'' \rangle^2 \cdot \Vert v' \Vert_2^2 \leqslant L \cdot \langle Z_i, u' \rangle^2 \langle Z_i, v'' \rangle^2   \\
    &\leq 2L \cdot \left(   \langle Z_i, u' \rangle^4 + \langle Z_i, u' \rangle^2 \langle Z_i, u'_\perp \rangle^2 \right)
\end{align*}
Putting it together, we want to bound,
\begin{align*}
    &\frac{1}{n}\sum_{i \in [n]} \langle X_i, u \rangle^2 \langle X_i, v \rangle^2 \mathbb{1}\left[\langle X_i, u \rangle^2 \leqslant \frac{20}{\epsilon}\right] \mathbb{1}\left[\Vert X_i \Vert^2 \leqslant \frac{20  L  d}{\epsilon^2}\right]  \\
    & \leqslant  2L \cdot  \frac{1}{n}\sum_{i = 1}^n \langle Z_i, u' \rangle^4 \mathbb{1}\left[\langle Z_i, u' \rangle^2 \leqslant \frac{20}{\epsilon}\right] \mathbb{1}\left[\vnorm{Z_i}_{\Sigma}^2 \leqslant \frac{20  L  d}{\epsilon^2}\right]  \\&+ 2L \cdot \frac{1}{n}\sum_{i = 1}^n \langle Z_i, u' \rangle^2  \langle Z_i, u'_\perp \rangle^2   \cdot\mathbb{1}\left[\langle Z_i, u' \rangle^2 \leqslant \frac{20}{\epsilon}\right] \mathbb{1}\left[\vnorm{Z_i}_{\Sigma}^2 \leqslant \frac{20  L  d}{\epsilon^2}\right]  
\end{align*}

\subsection*{Bounding the First Term}
We now consider the first term up to the scaling factor.

\begin{align*}
  \frac{1}{n}\sum_{i = 1}^n \langle Z_i, u' \rangle^4 \cdot\mathbb{1}\left[\langle Z_i, u' \rangle^2 \leqslant \frac{20}{\epsilon}\right] \mathbb{1}\left[\vnorm{Z_i}_{\Sigma}^2 \leqslant \frac{20  L  d}{\epsilon^2}\right] 
\end{align*} 
By applying~\cref{fact:scalarbernstein}, we can bound this quantity, as it represents an average of independent, bounded random variables with bounded variance. To see this, we first define, 
\begin{align*}
    Y_i := \langle Z_i, u' \rangle^4 \cdot\mathbb{1}\left[\langle Z_i, u' \rangle^2 \leqslant \frac{20}{\epsilon}\right] \mathbb{1}\left[\vnorm{Z_i}_{\Sigma}^2 \leqslant \frac{20  L  d}{\epsilon^2}\right] 
\end{align*}
Now we have that
\begin{align*}
    \left| Y_i \right| &= \left| \langle Z_i, u' \rangle^4 \cdot\mathbb{1}\left[\langle Z_i, u' \rangle^2 \leqslant \frac{20}{\epsilon}\right] \mathbb{1}\left[\vnorm{Z_i}_{\Sigma}^2 \leqslant \frac{20  L  d}{\epsilon^2}\right] \right| \\ 
    &\leqslant \left| \langle Z_i, u' \rangle^4 \cdot\mathbb{1}\left[\langle Z_i, u' \rangle^2 \leqslant \frac{20}{\epsilon}\right] \right| \leqslant \frac{400}{\epsilon^2}
\end{align*}
Furthermore, 
\begin{align*}
\E[Y_i] &= \E\left[ \langle Z_i, u' \rangle^4 \cdot\mathbb{1}\left[\langle Z_i, u' \rangle^2 \leqslant \frac{20}{\epsilon}\right] \mathbb{1}\left[\vnorm{Z_i}_{\Sigma}^2 \leqslant \frac{20  L  d}{\epsilon^2}\right] \right] \leqslant \E\left[ \langle Z_i, u' \rangle^4 \right] \leqslant 3 \\ 
    \E[Y_i^2] &= \E\left[ \langle Z_i, u' \rangle^8 \cdot\mathbb{1}\left[\langle Z_i, u' \rangle^2 \leqslant \frac{20}{\epsilon}\right] \mathbb{1}\left[\vnorm{Z_i}_{\Sigma}^2 \leqslant \frac{20  L  d}{\epsilon^2}\right] \right] \leqslant \E\left[ \langle Z_i, u' \rangle^8 \right] \leqslant 105
\end{align*}
Therefore, we have by~\cref{fact:scalarbernstein} we have that 
\begin{align*}
    \Pr\left[\sn \left(Y_i - \E[Y_i] \right) > t \right] \leqslant \exp \left( - \frac{n^2t^2}{2 \left(105 n + \frac{1}{3} \frac{400}{\epsilon^2} n t \right)}\right)
\end{align*}
By taking $t$ to be a sufficiently large constant, we have that for an absolute constant $C' > 0$ that
\begin{align*}
    \Pr\left[  \frac{1}{n}\sum_{i = 1}^n \langle Z_i, u' \rangle^4 \cdot\mathbb{1}\left[\langle Z_i, u' \rangle^2 \leqslant \frac{20}{\epsilon}\right] \mathbb{1}\left[\vnorm{Z_i}_{\Sigma}^2 \leqslant \frac{20  L  d}{\epsilon^2}\right] > C'\right] \leqslant \exp\left(- \Omega(n\epsilon^2) \right)
\end{align*}
Therefore, with probability at least $1 - \exp(-\Omega(n\epsilon^2))$  
\begin{equation}\label{eq:EvenE3}
     \frac{1}{n}\sum_{i = 1}^n \langle Z_i, u' \rangle^4 \cdot\mathbb{1}\left[\langle Z_i, u' \rangle^2 \leqslant \frac{20}{\epsilon} \right]\mathbb{1}\left[\vnorm{Z_i}_{\Sigma}^2 \leqslant \frac{20  L  d}{\epsilon^2}\right] \leqslant C'.
\end{equation}

\subsection*{Bounding the Second Term}
We now consider the second term up to the scaling factor,
\begin{align*}
    \frac{1}{n}\sum_{i = 1}^n \langle Z_i, u' \rangle^2  \langle Z_i, u'_\perp \rangle^2   \cdot\mathbb{1}\left[\langle Z_i, u' \rangle^2 \leqslant \frac{20}{\epsilon}\right] \mathbb{1}\left[\vnorm{Z_i}_{\Sigma}^2 \leqslant \frac{20  L  d}{\epsilon^2}\right]
\end{align*}
and observe that it suffices to bound,
\begin{align*}
    \frac{1}{n}\sum_{i = 1}^n \langle Z_i, u' \rangle^2  \langle Z_i, u'_\perp \rangle^2   \cdot\mathbb{1}\left[\langle Z_i, u' \rangle^2 \leqslant \frac{20}{\epsilon}\right].
\end{align*}
 We now observe that each summand is a product of independent random variables. In particular, for every $i \in [n]$, $\langle Z_i, u' \rangle $ is independent of $\langle Z_i, u'_\perp \rangle$ as $Z_i$ is a standard Gaussian variable. We also observe that for any $i \neq j$, since $X_i$ and $X_j$ are independent, so are $Z_i$ and $Z_j$ and their respective functions. 

To simplify, we define: $\alpha_i(u')^2 = \ \langle Z_i, u' \rangle^2 \cdot \mathbb{1}\left[\langle Z_i, u' \rangle^2  \leqslant \frac{20}{\epsilon}\right] $. Therefore, we now need to bound,
\begin{align*}
    \sum_{i = 1}^n \alpha_i(u')^2 \langle Z_i, u'_\perp \rangle^2,
\end{align*}
which is the sum of the product of independent random variables. We observe that the second random variable in each summand is sub-exponential, since it is the square of a standard Gaussian random variable. We will exploit this crucial property by conditioning on an event that depends on the $\alpha_i(u')$ across all indices $i \in [n]$. 
We begin by defining a good set $\cG$ depending on $u'$ as follows:
\begin{align*}
    \cG_{u'} = \left\{ \alpha := \left[ \alpha_1(u'), \alpha_2(u') \cdots \alpha_n(u') \right] \in \R^{n}:   \sum_i\alpha_i(u')^2 = O(n) \wedge \sum_i \alpha_i(u')^4 = O(n) \right\}
\end{align*}
For a fixed $u_\perp'$, 
\begin{align*}
    \Pr\left[\sn \alpha_i(u')^2 \langle Z_i, u'_\perp \rangle^2 > t\right] &= \Pr\left[ \
    \sn \alpha_i(u')^2 \langle Z_i, u'_\perp \rangle^2 > t \ | \ 
    \alpha \in \cG_{u'} \right] \cdot \Pr\left[\alpha \in \cG_{u'}\right]\\
    &+ \Pr\left[ \
    \sn \alpha_i(u')^2 \langle Z_i, u'_\perp \rangle^2 > t \ | \ 
    \alpha \notin \cG_{u'} \right] \cdot \Pr\left[\alpha \notin \cG_{u'}\right]\\
    &\leqslant \Pr\left[ \
    \sn \alpha_i(u')^2 \langle Z_i, u'_\perp \rangle^2 > t \ | \ 
    \alpha \in \cG_{u'} \right] + \Pr\left[\alpha \notin \cG_{u'}\right]
\end{align*}
where we used $\Pr\left[\alpha \in \cG_{u'}\right]\leq 1$ and $\Pr\left[ \
    \sn \alpha_i(u')^2 \langle Z_i, u'_\perp \rangle^2 > t \ | \ 
    \alpha \notin \cG_{u'} \right]\leq 1$. We abuse notation and will use $\cG$ for $\cG_{u'}$. 
 Defining $\textsc{Avg}(\alpha, u', u'_\perp) := \sn \alpha_i(u')^2 \langle Z_i, u'_\perp \rangle^2$, we have that 
\begin{align*}
    \Pr\left[\textsc{Avg}(\alpha, u', u'_\perp) > t \ | \ 
    \alpha \in \cG \right] &= \frac{1}{\Pr[\alpha \in \cG]} \cdot \Pr[\textsc{Avg}(\alpha, u', u'_\perp) > t \cap \alpha \in \cG] \\
    &= \frac{1}{\Pr[\alpha \in \cG]} \cdot  \int_{a \in \cG}\Pr[ \textsc{Avg}(\alpha, u', u'_\perp) > t | \alpha = a]\cdot p(a)\cdot da
\end{align*}
where $p(a)$ is the probability density function at the vector $a$. We now observe that we can control  $\Pr[ \textsc{Avg}(\alpha, u', u'_\perp) > t | \alpha = a]$ whenever $a \in \cG$ for fixed $a$ as it is a weighted sum of sub-exponential random variables. Denoting this tail probability by $B$ for some $B \in [0, 1]$, we have that 
\begin{align*}
    \Pr\left[ \textsc{Avg}(\alpha, u', u'_\perp) > t \ | \ 
    \alpha \in \cG \right] &= \frac{1}{\Pr[\alpha \in \cG]} \cdot  \int_{a \in \cG}\Pr[ \textsc{Avg}(\alpha, u', u'_\perp) > t | \alpha = a]\cdot p(a)\cdot da \\
    & \leqslant B \cdot  \frac{1}{\Pr[\alpha \in \cG]} \int_{a \in \cG} p(a) \cdot da = B
\end{align*}
Putting it together, we have that 
\begin{align*}
    \Pr\left[\sn \alpha_i(u')^2 \langle Z_i, u'_\perp \rangle^2 > t\right] \leqslant B + \Pr[\alpha \notin \cG]
\end{align*}
In the remainder of the proof, we will first show that $\alpha \in \cG$ with probability at least $1 - 2\exp(-\Omega( n \epsilon^2))$. Then we will apply the tail bound to complete the bound for the second term. Finally, we will put all the pieces together.
 
\subsubsection*{$\alpha \in \cG$ with high probability}
We will make essential use of~\cref{fact:scalarbernstein} once again to show that $\alpha \in \cG$ with high probability. We recall that 
\begin{align*}
    \alpha_i(u')^2 = \ \langle Z_i, u' \rangle^2 \cdot \mathbb{1}\left[\langle Z_i, u' \rangle^2 \leqslant \frac{20}{\epsilon}\right]
\end{align*}
    Once again, observe that this is a bounded random variable with bounded variance since:
   \begin{align*}
        &\E[\alpha_i(u')^2] \leqslant  \E\left[\langle Z_i, u' \rangle^2 \right] = 1 \\ 
        &|\alpha_i(u')^2| \leqslant \frac{20}{\epsilon} \\ 
        & \E[\alpha_i(u')^4] \leqslant \E\left[ \langle Z_i, u' \rangle^4 \right] = 3
    \end{align*} 
    Therefore via an application of~\cref{fact:scalarbernstein}, we have that 
    \begin{align*}
         \Pr\left[\sn \left(\alpha_i(u')^2 - \E[\alpha_i(u')^2] \right) > O(1) \right] \leqslant \exp \left( - \Omega(n \epsilon)\right)
    \end{align*}
    which implies that with probability at least $1 - \exp(-\Omega(n \epsilon))$ that 
    \begin{align*}
        \sum_{i=1}^n \alpha_i(u')^2 \leqslant a \cdot n
    \end{align*}
    for some absolute constant $a > 0$. Similarly observe that 
    \begin{align*}
        &\E[\alpha_i(u')^4] \leqslant  \E\left[\langle Z_i, u' \rangle^4 \right] = 3 \\ 
        &|\alpha_i(u')^4| \leqslant \frac{400}{\epsilon^2} \\ 
        & \E[\alpha_i(u')^8] \leqslant \E\left[ \langle Z_i, u' \rangle^8 \right] = 105
    \end{align*} 
     Therefore via an application of~\cref{fact:scalarbernstein}, we have that 
    \begin{align*}
         \Pr\left[\sn \left(\alpha_i(u')^4 - \E[\alpha_i(u')^4] \right) > O(1) \right] \leqslant \exp \left( - \Omega(n \epsilon^2)\right)
    \end{align*}
    and as a consequence have that 
    \begin{align*}
        \sum_{i=1}^n \alpha_i(u')^4 \leqslant b \cdot n
    \end{align*}
    for some absolute constant $b > 0$.
Therefore, via a union bound, we have that with probability at least $1 - 2 \exp(-\Omega(n \epsilon^2))$
\begin{align*}
    &\sum_i \alpha_i(u')^2 \leqslant a \cdot n \text{ and } \sum_i \alpha_i(u')^4 \leqslant b \cdot n
\end{align*}
from which we deduce,
\begin{equation}\label{eq:EventAlpha}
    \Pr\left[\alpha\in\cG\right] \geq 1 - 2 \exp(-\Omega(n \epsilon^2)).
\end{equation} 

\subsubsection*{Sub-exponential tail bounds for improved norm bound.}
Recall that it suffices for us to bound 
\begin{align*}
    \Pr[ \textsc{Avg}(\alpha, u', u'_\perp) > t | \alpha = a]
\end{align*}
for a fixed $a \in \cG$, 
where $\textsc{Avg}(\alpha, u', u'_\perp) := \sn \alpha_i(u')^2 \langle Z_i, u'_\perp \rangle^2$. Indeed, we will apply a sub-exponential weighted sum tail bound (\cref{subexponential}) to complete the proof. We have that $\langle Z_i, u'_\perp \rangle^2$ is a $\chi^2$-random variable with a single degree of freedom, as it is the square of the standard Gaussian random variable. We have that $\max_i a_i = O(1/\epsilon)$ and given that $a \in \cG$ further  that  $\sum_i a_i^2 = O(n)$. Taking $t = O(n)$ we have that 
\begin{align*}
    \Pr\left( \left| \sum_{i=1}^n \alpha_i(u')^2 \left(\langle Z_i, u'_\perp \rangle^2 - 1 \right) \right| > O(n) \right) &\le 2 \exp\left( - \Omega\left( \min\left( \frac{n^2}{\sum_{i=1}^n 2 a_i^2 }, \frac{n}{\max_i |a_i| K} \right) \right) \right) \\
    & = 2 \exp\left( - \Omega\left( \min\left( n, \frac{n\epsilon}{K} \right) \right) \right)\\
    & = 2 \exp\left( - \Omega\left(n\epsilon \right)\right)
\end{align*}
since $K = O(1)$ for our specific random variable. Putting the pieces together we have that 
\begin{align*}
    \Pr[ \textsc{Avg}(\alpha, u', u'_\perp) > O(1) | \alpha = a] \leqslant \exp\left( - \Omega(n \epsilon)\right)
\end{align*}
Therefore, to complete the entire second term bound, we have that 
\begin{align*}
    &\Pr\left[\frac{1}{n}\sum_{i = 1}^n \langle Z_i, u' \rangle^2  \langle Z_i, u'_\perp \rangle^2   \cdot\mathbb{1}\left[\langle Z_i, u' \rangle^2 \leqslant \frac{20}{\epsilon}\right] > O(1) \right] \\  &= \Pr \left[ \sn \alpha_i(u')^2 \langle Z_i, u'_\perp \rangle^2 > O(1) \right] \\
    &\leqslant \Pr\left[ \
    \sn \alpha_i(u')^2 \langle Z_i, u'_\perp \rangle^2 > O(1) \ | \ 
    \alpha \in \cG \right] + \Pr\left[\alpha \notin \cG\right] \\ 
    &\leqslant 2 \exp(-\Omega(n \epsilon)) + 2 \exp(- \Omega(n \epsilon)) = 4\exp(-\Omega(n \epsilon))
\end{align*}
Therefore, putting everything together, we have that for a fixed $u_\perp'$ that 
\begin{align*}
    \frac{1}{n}\sum_{i = 1}^n \langle Z_i, u' \rangle^2  \langle Z_i, u'_\perp \rangle^2   \cdot\mathbb{1}\left[\langle Z_i, u' \rangle^2 \leqslant \frac{20}{\epsilon}\right]\mathbb{1}\left[\vnorm{Z_i}_{\Sigma}^2 \leqslant \frac{20  L  d}{\epsilon^2}\right] \leqslant C_2.
\end{align*}
with probability at least $1 - 4 \exp(-\Omega(n \epsilon))$ for some absolute constant $C_2 > 0$.
\subsection*{Union Bound}
To complete the bound, we now need a union bound over all choices of $u_\perp'$ since $u_\perp$ depends on $v$. We want to bound, 
\begin{align*}
    \Pr\left[\sup_{u_\perp'} \frac{1}{n}\sum_{i = 1}^n \langle Z_i, u' \rangle^2  \langle Z_i, u'_\perp \rangle^2   \cdot\mathbb{1}\left[\langle Z_i, u' \rangle^2 \leqslant \frac{20}{\epsilon}\right] \mathbb{1}\left[\vnorm{Z_i}_{\Sigma}^2 \leqslant \frac{20  L  d}{\epsilon^2}\right] > C_2 \right]
\end{align*}

Now we will do a net argument over a $\gamma$ net of the unit sphere. The Lipschitzness of the following function will be precisely what we will prove next. 
\begin{align*}
    f(z) := \frac{1}{n}\sum_{i = 1}^n \langle Z_i, u' \rangle^2  \langle Z_i, z \rangle^2   \cdot\mathbb{1}\left[\langle Z_i, u' \rangle^2 \leqslant \frac{20}{\epsilon}\right] \mathbb{1}\left[\vnorm{Z_i}_{\Sigma}^2\leqslant \frac{20  L  d}{\epsilon^2}\right].
\end{align*}
Since $u'$ is fixed, we have that for any $z \in \cN_\gamma$ with $\Vert z \Vert_2 = 1$ such that $\Vert z - u_\perp' \Vert_2 = \gamma$, 
\begin{align*}
   f(u_\perp') - f(z) = &\frac{1}{n}\sum_{i = 1}^n \langle Z_i, u' \rangle^2  \langle Z_i, u'_\perp \rangle^2   \cdot\mathbb{1}\left[\langle Z_i, u' \rangle^2 \leqslant \frac{20}{\epsilon}\right] \mathbb{1}\left[\vnorm{Z_i}_{\Sigma}^2\leqslant \frac{20  L  d}{\epsilon^2}\right] \\
   &- \frac{1}{n}\sum_{i = 1}^n \langle Z_i, u' \rangle^2  \langle Z_i, z \rangle^2   \cdot\mathbb{1}\left[\langle Z_i, u' \rangle^2 \leqslant \frac{20}{\epsilon}\right] \mathbb{1}\left[\vnorm{Z_i}_{\Sigma}^2\leqslant \frac{20  L  d}{\epsilon^2}\right]\\
   & \leqslant \frac{20}{\epsilon} \frac{1}{n} \sum_i \langle Z_i, z - u_\perp' \rangle \langle Z_i, z + u_\perp' \rangle\mathbb{1}\left[\vnorm{Z_i}_{\Sigma}^2\leqslant \frac{20  L  d}{\epsilon^2}\right] \\ 
   & \leqslant \frac{40 \gamma}{\epsilon n}\sum_i \Vert Z_i \Vert^2 \mathbb{1}\left[\vnorm{Z_i}_{\Sigma}^2\leqslant \frac{20  L  d}{\epsilon^2}\right] \\
   &\leq \frac{40 \gamma}{\epsilon n} \sum_i \Vert \Sigma^{-1/2} \Vert_2^2 \Vert Z_i \Vert_{\Sigma}^2 \mathbb{1}\left[\vnorm{Z_i}_{\Sigma}^2\leqslant \frac{20  L  d}{\epsilon^2}\right] \\
   & \leqslant \frac{40}{\epsilon} \frac{1}{\mu} \frac{400 L^2d^2}{\epsilon^4} = \gamma \cdot \frac{16000 L d^2}{\epsilon^6}
\end{align*}
where in the last step we used $\epsilon \kappa < 1$. Taking $\gamma = O\left(\frac{\epsilon^7}{Ld^2}\right)$ suffices to establish the Lipschitzness of the function. 
We know that the $\gamma$-net has 
\begin{align*}
    | \cN_\gamma | \leqslant \left( \frac{3}{\gamma}\right)^d
\end{align*}
many elements. In particular our net has 
\begin{align*}
     | \cN_\gamma | \leqslant \left( \frac{24000 L d^2 }{\epsilon^7}\right)^d
\end{align*}
Now we do a union bound over this net and get that 
\begin{align*}
  &\Pr\left[\sup_{u_\perp'} \frac{1}{n}\sum_{i = 1}^n \langle Z_i, u' \rangle^2  \langle Z_i, u'_\perp \rangle^2   \cdot\mathbb{1}\left[\langle Z_i, u' \rangle^2 \leqslant \frac{20}{\epsilon}\right] \mathbb{1}\left[\vnorm{Z_i}_{\Sigma}^2 \leqslant \frac{20  L  d}{\epsilon^2}\right] > C_2 \right] \\
  &\leqslant  \left( \frac{24000 L d^2 }{\epsilon^7}\right)^d \cdot \exp(-\Omega(n \epsilon))\\
  & =  \exp\left(d \log \frac{24000 L d^2}{\epsilon^7} - \Omega(n \epsilon)\right).
\end{align*}
By taking $n = 100 d/\epsilon^2 \log \frac{24000 L d^2}{\epsilon^7}$, we can make the failure probability $\exp(-99 d/\epsilon \log \frac{24000 L d^2}{\epsilon^7})$. 
Therefore we have that with probability at least $1 - \exp(-99 d/\epsilon \log \frac{24000 L d^2}{\epsilon^7})$ that 
\begin{equation}\label{eq:EventE4}
        \sup_{u_\perp'} \frac{1}{n}\sum_{i = 1}^n \langle Z_i, u' \rangle^2  \langle Z_i, u'_\perp \rangle^2   \cdot\mathbb{1}\left[\langle Z_i, u' \rangle^2 \leqslant \frac{20}{\epsilon}\right] \mathbb{1}\left[\vnorm{Z_i}_{\Sigma}^2 \leqslant \frac{20  L  d}{\epsilon^2}\right] \leqslant C_2
\end{equation}
for some absolute constant $C_2$ whenever $n = \Omega\left( \frac{d}{\epsilon^2} \cdot \log \frac{Ld^2}{\epsilon^7} \right)$.
 
\subsection*{Completing the proof}
From Eqs.~\eqref{eq:EvenE3},~\eqref{eq:EventE4}, and~\eqref{eq:EventAlpha}, we have with probability at least $1 - 2\exp(-\Omega(n \epsilon)) - \exp(-\Omega(n \epsilon^2)) - \exp\left(-\Omega(d \log \frac{Ld^2}{\epsilon^7}) \right) = 1  - \exp\left(-\Omega(d \log \frac{Ld^2}{\epsilon^7}) \right)$ that for a fixed $u$ with $\Vert u \Vert_{\Sigma} = 1$,
\begin{align*}
    \sup_{v : \Vert v \Vert_2 = 1} \frac{1}{n}\sum_{i \in [n]} \langle X_i, u \rangle^2 \langle X_i, v \rangle^2 \mathbb{1}\left[\langle X_i, u \rangle^2 \leqslant \frac{20}{\epsilon}\right] \mathbb{1}\left[\Vert X_i \Vert^2 \leqslant \frac{20  L  d}{\epsilon^2}\right] \leqslant L \cdot C^{''}
\end{align*}
for some absolute constant $C^{''} > 0$. As a consequence,
\begin{align*}
    \vnorm{ \frac{1}{n} \sum_i \langle X_i, u \rangle^2 X_i X_i^T \mathbb{1}\left[\langle X_i, u \rangle^2 \leqslant \frac{20}{\epsilon}\right] \mathbb{1}\left[\Vert X_i \Vert^2 \leqslant \frac{20  L  d}{\epsilon^2}\right]}_{\text{op}} \leqslant L \cdot C'',
\end{align*}
concluding the proof of~\cref{lemma:TruncatedMain}.
\end{proof}

\subsection*{Proof of~\cref{thm:GaussSampleInf}}\label{sec:ThmUBProof}

The proof of Theorem~\ref{thm:GaussSampleInf} now follows similarly to that of Theorem~5 of~\cite{jambulapati2021robust} by replacing their Lemma 19 (\cref{lemma:JLSTnew}) with our~\cref{thm:MainUB}.

\section{ Statistical Query Lower Bounds for Robust Linear Regression}\label{sec:SQLB}

In this work, we rely on the SQ hardness for NGCA derived in~\cite{diakonikolas2017statistical}. More recently, these SQ hardness results have been refined in~\cite{diakonikolas2023sq}, which proves hardness under moment matching alone, without requiring bounded $\chi^2$-divergence. However these refinements appear to only provide quantitatively weaker lower bounds for fine-grained problems such as ours, and seem better suited to coarser separations such as between polynomial time and super-polynomial time SQ algorithms. 
We first give some background on the SQ Lower Bound framework. We begin by defining an SQ algorithm and the standard oracles.

\begin{definition}
    We define SQ Algorithms as algorithms that do not see samples from an underlying distribution $D$ but instead have access to the following oracles.
    \begin{enumerate}
        \item \textup{STAT}$(\tau)$: For a tolerance parameter $\tau > 0$, and \emph{any} bounded function $f : \R^d \to [-1, 1]$, \textup{STAT}$(\tau)$ returns a value $v$ such that $| v - \E_{x \sim D}[f(x)]| \leqslant \tau$.
        \item \textup{VSTAT}$(t)$: For a sample size parameter $t > 0$ and \emph{any} bounded function $f : \R^d \to [0,1]$, \textup{VSTAT}$(t)$ returns a value $v$ such that $| v - \E_{x \sim D}[f(x)]| \leqslant \tau$ where $\tau = \max \left\{\frac{1}{t}, \sqrt{\frac{\text{Var}_{x \sim D}[f(x)]}{t}} \right\}$ and $\text{Var}_{x \sim D}[f(x)] = \E_{x \sim D}[f(x)](1 - \E_{x \sim D}[f(x)])$.
    \end{enumerate}
\end{definition}

\subsection{Search Problems and their SQ Hardness}\label{sec:BasicsSQ}
\begin{definition}
    (Search Problem over Distributions) Let $\cD$ be a set of distributions on $\R^{d}$,  let $\cF$ be a set of solutions, and $\cZ : \cD \to 2^{\cF}$ be a map from the set of distributions to subsets of solutions. The search problem $\cZ$ over $\cD$ and $\cF$ is to find \emph{a} solution $f \in \cZ(D)$ given oracle access to an unknown $D \in \cD$. 
\end{definition}
The hardness of search problems is captured by the SQ Dimension. As a first step, we define the pairwise correlation between distributions.
\begin{definition}[Pairwise Correlation]
    The pairwise correlation of two distributions with probability density functions $D_1, D_2 : \R^{d} \to \R_+$ with respect to a reference distribution with density $S: \R^d \to \R_+$, where $\supp(D_1)\cup\supp(D_2)\subseteq\supp(S)$ is defined as $\chi_S(D_1, D_2) := \int_{\R^d} D_1(x) D_2(x)/S(x) dx - 1$. When $D_1 = D_2$, this is the $\chi^2$-divergence between $D_1$ and $S$. 
\end{definition}
\begin{definition}
    For $\gamma, \beta > 0$, the set of distributions $\{D_1, D_2, \dots, D_m\}$ is called $(\gamma, \beta)$-correlated relative to the distribution $S$, if $|\chi_S(D_i, D_j)| \leqslant \gamma$ whenever $i \neq j$ and $|\chi_S(D_i, D_i)| \leqslant \beta$. 
\end{definition}
The Statistical Query (SQ) dimension of a search problem is the largest set of $(\gamma, \beta)$-correlated distributions assigned to each solution.

\begin{definition}
    [SQ Dimension Defn. 2.11 in \cite{diakonikolas2017statistical}] 
    For $\gamma, \beta > 0$, a search problem $\cZ$ over a set of solutions $\cF$ and a class of distribution $\cD$ over $X$, we define the statistical dimension of $\cZ$ denoted by $SD(\cZ, \gamma, \beta)$ to be the largest integer $m$ such that there exists a reference distribution $S$ over $X$ and a finite set of distributions $\cD_S \subset \cD$ such that for any solution $f \in \cF$, the set $\D_f = \cD_S \backslash \cZ^{-1}(f)$ is  $(\gamma, \beta)$ correlated with respect to $S$ and $|\D_f| \geqslant m$.
\end{definition}

\begin{lemma}
    [Corollary 3.12 in \cite{feldman2017statistical}]\label{lem:SQHardness} Let $\cZ$ be a search problem over a set of solutions $\cF$ and a class of distributions $\cD$ over $\R^d$. For $\gamma, \beta > 0$, let $s = SD(\cZ, \gamma, \beta)$ be the statistical dimension of the problem. For any $\gamma' > 0$, any SQ algorithm for $\cZ$ requires at least  $s\gamma'/(\beta - \gamma)$ queries to the $\textup{STAT}(\sqrt{\gamma + \gamma'})$ or $\textup{VSTAT}(1/\left(3(\gamma+\gamma')\right))$ oracles to solve $\cZ$.
\end{lemma}
We next state a few standard techniques from \cite{diakonikolas2017statistical} that we will utilize. 

\begin{definition}(Hidden Direction Distribution)
    For a unit vector $v \in \R^d$ and a distribution $A$ on the real line with density function $A(x)$, we define $P_{A, v}$ to be the distribution that is $A$ in the direction $v$ and standard Gaussian in the $d-1$ dimensional complement of $v$. Formally, $P_{A,v}(x)=A(\langle v,x\rangle)\,\phi_{v^\perp}(x)$ where $\phi_{v \perp}(x) = \exp(- \Vert x - \langle v, x \rangle v \Vert^2 / 2)/(2\pi)^{(d-1)/2}$
\end{definition}
The distributions $P_{A, v}$ can be shown to be nearly uncorrelated as long as the directions where $A$ is embedded is pairwise almost orthogonal. We remark here that the problem of finding the hidden direction is also referred to as Non-Gaussian component analysis (NGCA) in contemporary literature.
\begin{lemma}[Lemma 3.4 in \cite{diakonikolas2017statistical}]\label{lem:Correlation2Moment}
    Let $m \in \Z_+$. Let $A$ be a distribution over $\R$ that matches the first $m$ moments of $\cN(0,1)$. For any $v$, let $P_{A, v}$ be the distribution defined above. For all $v',v \in \R^d$, we have that $\chi_{\cN(0, I_d)}(P_{A, v}, P_{A, v'}) \leqslant | \langle v, v' \rangle |^{m+1} \cdot \chi^2(A, \cN(0,1))$.
\end{lemma}
\begin{lemma}[Lemma 3.7 in \cite{diakonikolas2017statistical}]\label{lem:SizeSQSet}
    For any $0 < c < \frac{1}{2}$, there is a set $S$ of at least $2^{\Omega(d^c)}$ unit vectors in $\R^d$ such that for each pair of distinct $v, v'$ it is the case that $| \langle v, v' \rangle | \leqslant O(d^{c-1/2})$.
\end{lemma}

\subsection*{Proof of~\cref{thm:MainSQ}}
Utilizing the machinery described above, we now prove our main result of this section,~\cref{thm:MainSQ}. 

\begin{theorem}[Full Version of~\cref{thm:MainSQ}]\label{thm:SQFull}
Let $c > 0$ be an arbitrary sufficiently small constant, and $\kappa$, $\eps$ be given, satisfying $\eps \kappa \lesssim 1$, $\kappa$ sufficiently large and $\varepsilon$ sufficiently small. For every vector $v \in \mathbb{R}^d$ with $\|v\| = 1$, there exists a covariance matrix $\Sigma_v$ (with condition number $\kappa$), a direction $\beta_v$, and a corruption distribution $E_v$, such that the following holds.

Let $Q_v$ be a distribution over $(X, y)$ where $X \sim \cN(0, \Sigma_v)$ and $y \sim \langle X, \beta_v \rangle + \eta$ where  $\eta \sim \cN(0, \zeta^2)$, $\zeta^2 < 1$ unknown and $\eta$ independent of $X$.

Let $Q'_v = (1-\eps) Q_v + \eps E_v$. Then, any SQ algorithm, which is given access to $Q'_v$ (without knowing $v$) that outputs a vector
$\widehat{\beta_v}$ such that $\Vert \widehat{\beta}_v - \beta_v \Vert_{\Sigma_v} = o(\sqrt{\epsilon \kappa})$ either makes at least $2^{\Omega(d^c)}$ queries, or at least one query to the STAT oracle with tolerance 
\[
\tau \leq  O\left(\sqrt{d^{4c-2} \cdot ( \sqrt{\kappa} e^{O(1/\epsilon)})}\right).
\]
\end{theorem}

\begin{proof}[Proof of~\cref{thm:SQFull}]
Let $T \subset \cS^{d-1}$ be the set of nearly-orthogonal vectors from~\cref{lem:SizeSQSet} with $|T| = 2^{\Omega(d^c)}$. For each $v \in T$ we define:
\begin{align*}
    Q_v'(x, y) = (1 - \epsilon) \cdot Q_v(x, y) + \epsilon \cdot E_v(x, y),%
\end{align*}
where $Q_v$ and $Q'_v$ are as defined in~\cref{sec:LBInstance} and~\cref{sec:SQmarginal} respectively. Our goal will be ensure the preconditions for applying~\cref{lem:SQHardness}. To that end we let the reference distribution be
\[
S  = \cN(0, I_d) \times  R(y),
\]
where $R(y)$ is defined in~\cref{sec:SQmarginal}, and define $\cD_S := \{Q_v'\}_{v \in T}$. 
Let $\beta_v = c_1 \sqrt{\epsilon} \kappa v$ denote the regression vector corresponding to $Q_v'$ for a constant $c_1 > 0$ that will be picked later in~\cref{sec:SQmarginal}. The set of solutions $\cF$ is defined as, $\cF: = \{u\in \R^d: \|u\|_2\leq c_1\sqrt{\epsilon}\kappa\}$, and the set of solutions $\cZ(Q_v')$ corresponding to $Q_v'$ is 
\begin{align*}
    \cZ(Q_v') := \{u \in \cF : \Vert u - \beta_v \Vert_{\Sigma_v} \leqslant r \},
\end{align*}
for $r = o(\sqrt{\epsilon \kappa})$. We claim that for any $u \in \cF$, there is at most one element $Q_v' \in \cD_S$ that satisfies $\Vert u - \beta_v \Vert_2 = o(\sqrt{\epsilon} \kappa)$. To see why, observe that for some $u_v \in \cZ(Q_v')$,
\begin{align*}
    \Vert u_v - \beta_v \Vert_2 &= \Vert \Sigma_v^{-1/2} \Sigma_v^{1/2} (u_v - \beta_v) \Vert_2 \leqslant \Vert \Sigma_v^{-1/2} \Vert_2 \Vert \Sigma_v^{1/2} (u_v - \beta_v) \Vert_2 \\
    &=  \Vert \Sigma_v^{-1/2} \Vert_2 \Vert  (u_v - \beta_v) \Vert_{\Sigma_v} \leqslant \sqrt{\kappa} \cdot o(\sqrt{\epsilon \kappa}) = o(\sqrt{\epsilon} \kappa)
\end{align*}
However, for 
$v \neq v'$, $\Vert \beta_v - \beta_{v'}  \Vert_2 = \Omega(\sqrt{\epsilon} \kappa)$ since
\begin{align*}
    \Vert \beta_v - \beta_{v'}  \Vert_2 = \sqrt{2 c_1^2 \epsilon \kappa^2 (1 - \langle v, v' \rangle)} = \Omega(\sqrt{\epsilon} \kappa).
\end{align*}
We now claim that $u_v \notin \cZ(Q_{v'}')$ for any $v'\neq v$. Note that
\begin{align*}
    \Omega(\sqrt{\epsilon} \kappa) &=  \Vert \beta_v - \beta_{v'}  \Vert_2  \\
    &\leq \Vert u_v -\beta_v\|_2 +\|u_v- \beta_{v'}  \Vert_2 \leq o(\sqrt{\epsilon}\kappa) + \|u_v- \beta_{v'}  \Vert_2.
\end{align*}
Now,
\[
\Omega(\sqrt{\epsilon} \kappa)\leq \|u_v- \beta_{v'}\|_2 \leq \sqrt{\kappa} \cdot \|u_v- \beta_{v'}\|_{\Sigma_{v'}}, 
\]
concluding our claim. Therefore, $|\cZ^{-1}(u)|\leq 1,$ 
and  $|\cD_S \backslash \cZ^{-1}(u)| \geqslant |\cD_S| - 1\geq 2^{\Omega(d^c)}$, where we use~\cref{lem:SizeSQSet}. We will prove in~\cref{SQ:gammabeta}, that the pairwise correlations $\gamma, \beta$ with respect to $S$ are 
\begin{align*}
    &\gamma \leqslant O(| \langle v, v' \rangle|^4 \cdot ( \sqrt{\kappa} e^{O(1/\epsilon)})) = O(d^{4c-2} \sqrt{\kappa}  e^{O(1/\epsilon)}), \quad\text{and,} \quad \beta \leqslant O(\sqrt{\kappa}  e^{O(1/\epsilon)}).
\end{align*}
As a result, our search problem $\cZ$ has SQ dimension at least $2^{\Omega(d^c)}$ and is $(\gamma,\beta)$ correlated.  
We finally apply~\cref{lem:SQHardness} to conclude that any SQ algorithm that can output a vector $u$ such that $\Vert u - \beta_v \Vert_{\Sigma_v} \leqslant o(\sqrt{\epsilon \kappa})$ makes at least $2^{\Omega(d^c)}$ queries, or at least one query to STAT of tolerance 
\[
\tau \leq  O\left(\sqrt{d^{4c-2} \cdot ( \sqrt{\kappa} e^{O(1/\epsilon)})}\right).
\]
\end{proof}

In the following sections, we will construct distributions $Q_v$ and $Q'_v$ and prove bounds on their pairwise correlations. In order to prove such correlation bounds we utilize~\cref{lem:Correlation2Moment}. Therefore, we first define our distribution $Q_v$ in~\cref{sec:LBInstance}, then construct a moment matching distribution $Q'_v$ of the required form in~\cref{sec:MomentMatching}. Further in~\cref{sec:SQmarginal} we show that our constructed $Q'_v$ is indeed a corruption of $Q_v$. Finally in~\cref{sec:CorrelationBound} we prove the required pairwise correlation bounds.

\subsection{Lower Bound Instance}\label{sec:LBInstance}

 We describe the uncorrupted distribution $Q_v$. The main idea is similar to prior works and focuses on aligning the signal $\beta$ with the direction of low variance of $\Sigma$, which is the same as the unknown direction $v$. More precisely, for every $v$, we define the parameters of the uncorrupted distribution $Q_v$ as follows.
\begin{definition}[Hard Instance for SQ Lower Bound]\label{def:sq hard instance}
\begin{align*}
\Sigma &= I_d - (1 - 1/\kappa)vv^T\\
\beta &= c_1\sqrt{\epsilon} \kappa v, \ c_1 > 0. 
\end{align*}
\end{definition}
Similar to \cite{diakonikolas2019efficient}, we set $y$ to have unit variance. That is, $\sigma_y^2 = c_1^2 \epsilon \kappa + \zeta^2 = 1$. The main observation in the lower bound construction of~\cite{diakonikolas2019efficient} is that the conditional distribution $X | Y = y$ is an instance of the hidden direction detection problem. They then match the first three moments of the conditional distribution $X|Y=y$ with $\cN(0,1)$. This conditional distribution takes the following form
\begin{align*}
 [X|Y=y] \sim \cN\left(\frac{\Sigma \beta  y}{\sigma_y^2}, \Sigma - \frac{ (\Sigma \beta)(\Sigma \beta)^T}{\sigma_y^2}\right).
\end{align*}
and is not Gaussian in the direction $v$. The one-dimensional conditional distribution along the direction $v$ is 
\[
    \cN\left(\underbrace{c_1 \sqrt{\epsilon} y}_{\coloneqq \mu_s.}, \underbrace{\frac{1}{\kappa} - c_1^2 \epsilon}_{\coloneqq \sigma_s^2.} \right).
\] 

We pause here to remark that in \cite{diakonikolas2019efficient}, they pick $\kappa$ to be a constant (at most $2$). This choice made the resulting conditional distribution have constant variance in the direction of $v$. In that setting, the above Gaussian with large constant variance is corrupted to match moments with $\cN(0, 1)$. This sufficed for their goal of showing the hardness of achieving error $o(\sqrt{\eps})$ in the sub-quadratic sample regime.

In contrast, in our setting, the variance of the one-dimensional conditional distribution can be arbitrarily small. Since our goal is to establish hardness results for achieving an error $o(\sqrt{\eps \kappa})$, we cannot use the same moment matching construction as \cite{diakonikolas2019efficient}. Instead, we require a new moment matching construction, and the ill-conditioning introduces additional challenges that we address in the next section.

We also note that the conditional distribution considered above depends on $y$. Therefore, we must corrupt the distribution so that it matches the moments of $\cN(0, 1)$ for every possible $y \in \R$. \cite{diakonikolas2019efficient} handled this by choosing a different corruption rate for each $y$ while still ensuring a global corruption rate of $\epsilon$ on the joint distribution. We adopt a similar strategy.

\subsection{Moment Matching}\label{sec:MomentMatching}
In this section, we will prove the following. 
\begin{lemma}\label{SQ:MainMomentMatching}
  For $\epsilon > 0$ sufficiently small, $\mu_s \in \mathbb{R}$ , $\sigma_s^2 \in (0, 0.1]$, there exists a distribution $A_{\mu_s}$ such that $A_{\mu_s}$ agrees with the first three moments of $\cN(0, 1)$ and has the form $A_{\mu_s} = (1-\epsilon_{\mu_s}) \cdot \cN (\mu_s, \sigma_s^2)+\epsilon_{\mu_s}\cdot B_{\mu_s}$ for some distribution $ B_{\mu_s}$ and $\epsilon_{\mu_s}$ such that:
\begin{itemize}
    \item If $|\mu_s| \geqslant \frac{\sqrt{\epsilon}}{10000}$, then $\epsilon_{\mu_s} / (1 - \epsilon_{\mu_s}) \leqslant O(\mu_s^2)$, and,   \[
    \chi^2(A_{\mu_s}, \cN(0,1)) = O(\sqrt{\kappa}e^{O(\max\{\mu_s^2,1/\mu_s^2\})})
    \]
    \item If $|\mu_s| < \frac{\sqrt{\epsilon}}{10000}$, then $\epsilon_{\mu_s} = \epsilon$, and $\chi^2(A_{\mu_s}, \cN(0,1)) = \exp(O(\frac{1}{\epsilon}))$.
\end{itemize}  
\end{lemma}
Our result generalizes the moment matching construction of \cite[Lemma E.2]{diakonikolas2019efficient} for $\kappa = \Omega(1)$. 

We now present our construction.
\begin{definition}\label{def:Mixture}
Define the distribution,
    \begin{align}\label{eqn:A}
    A_{\mu_s} \coloneqq \left\{
    \begin{array}{l}
        P_{1, \epsilon_{\mu_s}} \text{ whenever } \ |\mu_s| \leqslant \sqrt{\epsilon}/{10000},     \\
        P_{2, \epsilon_{\mu_s}} \text{ whenever } \ \sqrt{\epsilon}/10000 \leqslant |\mu_s| < 0.65,         \\
        P_{3, \epsilon_{\mu_s}} \text{ whenever } \ 0.65 \leq |\mu_s|.
        \end{array}
    \right\}
\end{align}
\end{definition}

\begin{enumerate}
    \item For $|\mu_s| < \frac{\sqrt{\epsilon}}{10000}$, we utilize the following construction from~\cite{diakonikolas2025sos}. We have modified their statement to include additional properties of their construction for our use case.
    \begin{lemma}[Lemma 7.10,~\cite{diakonikolas2025sos}]\label{lem:smallEps}
        There exists a positive constant $\eta_0$ such that for all $\eta \in (0,\eta_0)$, and for every $\xi \in (0,1/2)$, there exist univariate distributions $A$ and $Q$ satisfying $A := (1-\eta)\cN (\delta, \xi^2)+\eta Q$ such that,
        \begin{enumerate}[(i)]
            \item $\delta = 0.001\sqrt{\eta}$, and
            \item $A$ matches the first three moments with $\cN(0,1)$.
        \end{enumerate}
        Furthermore, $Q = \sum_{i = 1}^4 w_i \cN(\mu_i, 1)$ where $|\mu_i| \leqslant \frac{10}{\sqrt{\epsilon}}$, and $w_i \geqslant 0$ for $ \ i \in \{1, 2, 3, 4\}$ and $\sum_{i=1}^4 w_i = 1$.
    \end{lemma}  
    We now apply~\cref{lem:smallEps} to obtain distribution $ P_{1, \epsilon_{\mu_s}}$. Given $\epsilon$ and $\sigma_s^2\leq 0.1$, let $A, Q$ be as described in~\cref{lem:smallEps} for $\eta \leq \epsilon/100$ and $\delta = 0.001\sqrt{\eta} < 0.0001\sqrt{\epsilon}$. We get,
    \begin{align*}
        P_{1, \epsilon_{\mu_s}} = A = (1 - \epsilon) \cdot \cN(\delta, \sigma_s^2) + \epsilon \cdot Q,
    \end{align*}
    for $\epsilon_{\mu_s} = \epsilon$. Now  $P_{1,\epsilon_{\mu_s}}$ matches the first three moments with $\cN(0,1)$. 
    
\item For $\frac{\sqrt{\epsilon}}{10000} \leqslant | \mu_s | < 0.67$, we consider the following construction. Let $\epsilon_{\mu_s}$ be a function of $\mu_s$, to be specified later. We define a mixture of Gaussians whose first three moments match those of $\mathcal N(0,1)$:
    \begin{align*}
        P_{2, \eps_{\mu_{s}}} = (1 - \epsilon_{\mu_s}) \cdot \cN(\mu_s, \sigma_s^2) + \epsilon_{\mu_s}\left(\frac{1}{9} \cdot \cN(-2\mu_N, \sigma_N^2) + \frac{8 }{9} \cdot \cN(\mu_N, \tau^2)\right)
    \end{align*}
    where $\tau^2 > 0$ is a fixed constant. This construction is inspired by \cite{diakonikolas2019efficient}, from which we adopt the mixing weights and the relationship between the noise means. The key difference is that our setting is ill-conditioned, as $\sigma_s^2$ may be arbitrarily close to zero. To address this, we allow $\mu_s$ to depend on $\sigma_s^2$ in a precise manner that ensures all component variances remain uniformly bounded (specifically in $(0,2)$), which is necessary to obtain finite $\chi^2$ divergence with respect to $\mathcal N(0,1)$. We provide more intuition for this later. 

    We next specify $\mu_N$, $\sigma_N^2$, and $\tau^2$, determine the relationship between $\epsilon_{\mu_s}$ and $\mu_s$, and show how to express $\mu_s$, $\mu_N$, and $\sigma_N^2$ in terms of $\sigma_s^2$, $\tau^2$, and $\epsilon_{\mu_s}$. Finally, we prove that for every fixed $\sigma_s^2$, there is a bijection between values of $\mu_s$ and $\epsilon_{\mu_s}$.

Using that the first moment must be $0$ yields,
    \begin{align*}
        \mu_s = - \frac{2 \epsilon_{\mu_s}}{3 (1 - \epsilon_{\mu_s})} \mu_N.
    \end{align*}

Substituting this into the second- and third-moment equations, we note that the third-moment constraint is linear in $\sigma_N^2$, while the second-moment constraint is linear in $\mu_N^2$. We therefore solve first for $\sigma_N^2$, substitute into the variance equation to obtain $\mu_N^2$ as a function of $\epsilon_{\mu_s}$, $\tau^2$, and $\sigma_s^2$, and then recover $\sigma_N^2$. These symbolic computations were carried out using SymPy \cite{sympy}\footnote{Our SymPy code is available in~\cref{sec:code}.}. With this, we get the following values,

    \begin{align*}
    \mu_N^2 = 
\frac{27(1-\epsilon_{\mu_s})^2\left(4\epsilon_{\mu_s}\sigma_s^2-4\epsilon_{\mu_s}\tau^2-3\sigma_s^2+3\right)}
{4\epsilon_{\mu_s}\left(17\epsilon_{\mu_s}^2-45\epsilon_{\mu_s}+27\right)}.
    \end{align*}
    \begin{align*}
        \sigma_N^2 = \frac{-63 \epsilon_{\mu_s}^2 \sigma_s^2  + 80 \epsilon_{\mu_s}^2 \tau^2 + 144 \epsilon_{\mu_s} \sigma_s^2 - 180 \epsilon_{\mu_s} \tau^2 - 9\epsilon_{\mu_s} - 81 \sigma_s^2 + 108\tau^2}{17 \epsilon_{\mu_s}^2 - 45 \epsilon_{\mu_s} + 27}.
    \end{align*}
   Using $\tau^2 = 0.2$ and taking the positive square root gives
    \begin{equation}\label{eqn:mumain}
        \mu_s = \sqrt{\frac{3\epsilon_{\mu_s}\left( 20 \epsilon_{\mu_s} \sigma_s^2 - 4\epsilon_{\mu_s} - 15 \sigma_s^2 + 15\right)}{5 \left( 17 \epsilon_{\mu_s}^2 - 45 \epsilon_{\mu_s} + 27 \right)}},
    \end{equation} and, 
    \begin{align*}
        \sigma_N^2 = \frac{-315 \epsilon_{\mu_s}^2 \sigma_s^2  + 80 \epsilon_{\mu_s}^2 + 720 \epsilon_{\mu_s} \sigma_s^2 - 225\epsilon_{\mu_s} - 405 \sigma_s^2 + 108}{5 \left( 17 \epsilon_{\mu_s}^2 - 45 \epsilon_{\mu_s} + 27 \right)}.
    \end{align*}

    As $\epsilon_{\mu_s} \to 0$, the parameter $\mu_s$ scales as $
\mu_s = O\!\left(\sqrt{(1 - \sigma_s^2)\,\epsilon_{\mu_s}}\right)$.
While one might expect $\mu_s$ to depend only on $\epsilon_{\mu_s}$, this additional dependence on $\sigma_s^2$ is necessary in our ill-conditioned setting. This dependency can be seen by taking noise component variances to be a constant and expressing $\mu_s$ in terms of the other parameters in the second moment equations.

We claim that for every $\sigma_s^2 \in (0,0.1)$, by varying the parameter $\eps_{\mu_s}$ between $(0, 0.51]$ the parameter $\mu_s$ can be made to attain every value in $[0, 0.65)$. Moreover, for $\epsilon_{\mu_s} \in (0,0.51]$ we have for all $\sigma_s^2 \in (0,0.1)$ that $\sigma_N^2 \in (0.01, 0.8)$.

These properties are verified symbolically using Mathematica~\cite{mathematica}; see~\cref{sec:code}.
Moreover, for each fixed $\sigma_s^2 \in (0,0.1)$, the map
$\epsilon_{\mu_s} \mapsto \mu_s(\epsilon_{\mu_s})$
is a bijection over $\epsilon_{\mu_s} \in (0,0.5]$; see~\cref{sec:code}.

\item  For $0.61 \le |\mu_s| < \infty$, we use the following construction.

Let
\[
\mu_s = \frac{1}{3\sqrt{1 - \epsilon_{\mu_s}}},
\]
which implies $0.7 \le \epsilon_{\mu_s} < 1$. Fix $k = \tfrac{4}{5}$ and define
\begin{align*}
P_{3,\epsilon_{\mu_s}}
&=
(1 - \epsilon_{\mu_s}) \cdot \mathcal N(\mu_s, \sigma_s^2)
\\
&\quad
+ \epsilon_{\mu_s}
\left(
\frac{(1/\epsilon_{\mu_s} - 1)}{k^3}\, \mathcal N(-k\mu_s, v_2)
+
\left(1 - \frac{1 - \epsilon_{\mu_s}}{k^3\epsilon_{\mu_s}}\right)
\mathcal N(\mu_3, v_3)
\right).
\end{align*}
where 
\[
\mu_3 = \frac{36(1-\epsilon_{\mu_s})\mu_s}{189\epsilon_{\mu_s}-125}.
\]

We claim that for every $\sigma_s^2 \in (0,0.1]$ and $\epsilon_{\mu_s} \in [0.7,1)$,
\[
v_2 \in (0.2,1)
\qquad\text{and}\qquad
v_3 \in (0.7,1.9),
\]
so all the noise variances remain bounded by absolute constants. This is verified symbolically using Mathematica; see~\cref{sec:code}.

 In this construction we re-utilize symmetries from the second construction, such as scaling the means and the mixing weights together, and search over the appropriate scaling $k$, allowing simpler moment calculations. We then set up a linear system involving $v_2$ and $v_3$, expressed in terms of the weights and the signal mean parameterized by $\eps_{\mu_s}$ using the second and third moments. Solving this linear system in SymPy gives us 
 \[
v_2(\eps_{\mu_s},\sigma_s^2)= \frac{
42525\,\eps_{\mu_s}^{2}\sigma_s^{2}
-58554\,\eps_{\mu_s} \sigma_s^{2}
+20125\,\sigma_s^{2}
+5157\,\eps_{\mu_s}
-3429
}{
25\,(9\eps_{\mu_s}-5)\,(189\eps_{\mu_s}-125)
}
\]
\[
\begin{aligned}
v_3(\eps_{\mu_s},\sigma_s^2)
&=\frac{3}{(9\eps_{\mu_s}-5)\,(189\eps_{\mu_s}-125)^{2}}
\Bigl(
107163\,\eps_{\mu_s}^{3}\sigma_s^{2}
-248913\,\eps_{\mu_s}^{2}\sigma_s^{2}
+188625\,\eps_{\mu_s}\,\sigma_s^{2}\\
&\qquad\qquad\qquad\qquad
-46875\,\sigma_s^{2}
+36243\,\eps_{\mu_s}^{2}
-48102\,\eps_{\mu_s}
+15955
\Bigr).
\end{aligned}
\]
which can be verified to be within the specified range using Mathematica, as mentioned earlier.

\end{enumerate} 
We remark that in the above constructions, to deal with negative means, we can simply reflect the constructions.
\begin{lemma}\label{lem:MomMat}
    $A_{\mu_s}$ is well-defined for every $\mu_s \in \R$ and has first three moments $0, 1, 0$.
    Furthermore, whenever $|\mu_s| > \sqrt{\epsilon}/10000$, we have 
    \[
        \frac{\epsilon_{\mu_s}}{1 - \epsilon_{\mu_s} } \leqslant 9 \mu_s^2.
    \]
\end{lemma}
\begin{proof}
 Based on our construction, we have two cases. In the first we have that 
\begin{align*}
    \mu_s^2 = \frac{3\epsilon_{\mu_s}\left( 20 \epsilon_{\mu_s} \sigma_s^2 - 4\epsilon_{\mu_s} - 15 \sigma_s^2 + 15\right)}{5 \left( 17 \epsilon_{\mu_s}^2 - 45 \epsilon_{\mu_s} + 27 \right)},
    \end{align*}
which simplifies to
\begin{align*}
    \frac{\epsilon_{\mu_s}}{1 - \epsilon_{\mu_s}} \leqslant 5 \mu_s^2,
\end{align*}
for all choices of $\sigma_s^2\leq0.1$ and $\epsilon_{\mu_s} \in [0,0.5]$. We provide Mathematica code for verification in~\cref{sec:code}.  Now, for the last construction, we have that 
\begin{align*}
    \mu_s = \frac{1}{3(\sqrt{1 - \epsilon_{\mu_s}})}.
\end{align*}
We conclude since
\begin{align*}
    \frac{\epsilon_{\mu_s}}{1 - \epsilon_{\mu_s}}  \leqslant \frac{1}{1 - \epsilon_{\mu_s}} \leqslant 9 \mu_s^2.
\end{align*}
\end{proof}

\subsection{The Marginal and Corruption } \label{sec:SQmarginal}

In this section we formally define the marginal distribution over $y$. We then prove that the construction in~\cref{sec:MomentMatching} can indeed be viewed as an $\epsilon$-corruption of $Q_v$ for every $v \in T$.

Following \cite{diakonikolas2019efficient}, we define
\begin{align*}
    R(y) \propto \frac{G(y)}{1 - \epsilon_{\mu_s}},
\end{align*}
where $G(y)$ is the standard 1-dimensional Gaussian probability density function.
To see that this is an actual distribution, it suffices to show that the normalization factor is bounded. Indeed we have,
\begin{align*}
    C &= \int_\R \frac{G(y)}{1 - \epsilon_{\mu_s}}dy = \int_\R G(y)dy + \int_\R \frac{\epsilon_{\mu_s}}{1 - \epsilon_{\mu_s}}\cdot G(y) dy \\
    &= 1 + \int_{-{1/( 10000c_1)}}^{{ 1/(10000 c_1)}} \frac{\eps}{1 - \eps} \cdot G(y) dy  + \int_{-\infty}^{-{1/(10000c_1)} }  \frac{\epsilon_{\mu_s}}{1 - \epsilon_{\mu_s}} G(y) dy + \int_{{1/(10000c_1)}}^{\infty} \frac{\epsilon_{\mu_s}}{1 - \epsilon_{\mu_s}} \cdot G(y) dy\\
    &\leqslant 1 + \frac{\eps}{1 - \eps} \cdot \Pr\left(|y| \leq \frac{1}{10000c_1}\right) + 9 c_1^2 \eps \left(\int_{-\infty}^{-{1/(10000c_1)} } y^2 G(y) dy + \int_{{1/(10000c_1)}}^{\infty} y^2 G(y) dy\right) \\
    &\leq 1 +  \frac{\eps}{1 - \eps} \cdot \Pr\left(|y| \leq \frac{1}{10000c_1}\right) + 9 c_1^2 \eps \left( \int_\R y^2 G(y) dy \right) \leq \frac{1}{1 - \eps}.
\end{align*}
The first inequality uses~\cref{lem:MomMat} and the last inequality follows by taking $c_1$ appropriately small. We further note that $C \geq 1$ since 
\[
 G(y) \leq  G(y)/(1 - \eps_{\mu_s}), 
\]
which in particular implies
\[
R(y) \leq G(y)/(1 - \eps_{\mu_s}).
\]
Therefore we have 
\begin{align*}
    R(y) = \frac{G(y)}{(1 - \epsilon_{\mu_s}) \cdot C} \geqslant \frac{(1- \epsilon) \cdot G(y)}{(1 - \epsilon_{\mu_s})} .
\end{align*}
We note that $c_1$ was picked so that $C \leq \frac{1}{1 - \eps}$. This bound on $C$ was chosen in a way to allow the corrupted conditional distribution to still produce a joint distribution that is an $\eps$ corruption of $Q_v$.
\begin{remark}\label{SQ:assumptionepskappa}
    A careful reader might have noticed that the condition $\eps \kappa < 1$ was crucial in the above proof. Indeed, if $\eps \kappa = \omega(1)$, then $\sigma_y^2 = \omega(1)$ in our construction and the corresponding choice of $R(y) \propto G(\sigma_y y) / (1 - \eps_{\mu_s})$ would no longer produce a Huber contamination.
    \end{remark}
\begin{definition}
    For every $v \in T$, define the corrupted distribution as
    \[
         Q_v' \coloneqq \frac{1}{(2 \pi)^{(d-1)/2}}A_{\mu_s}(\langle v, x \rangle) \exp(- \Vert x - \langle v, x \rangle v \Vert^2 / 2) \cdot R(y).
    \]
\end{definition}
Denoting the $d$-dimensional conditional distribution $Q_v'(X | Y = y)$ as $P_{\mu_s, v}$, we will now show that $Q_v'$ is indeed an $\eps$ Huber contamination of $Q_v$ (from~\cref{sec:LBInstance}).

\begin{lemma}\label{SQ:condhuber} 
There exists $E_v$ such that $Q_v' = (1 - \epsilon) Q_v + \epsilon E_v$.
\end{lemma}
\begin{proof}
By construction,
\[
    A_{\mu_s} \geq (1 - \eps_{\mu_s}) \cN(\mu_s, \sigma_s^2).
\]
This implies that 
\[
    P_{\mu_s, v} \geq \left(1 - \eps_{\mu_s}) \cN(\mu_s v, I_d - vv^T + \sigma_s^2 vv^T\right).
\]
Since $\mu_s = c_1 \sqrt{\epsilon} y$ and $\sigma_s^2 = \frac{1}{\kappa} - c_1^2 \epsilon$, we have
\[
    Q_v' \coloneqq P_{\mu_s, v} \cdot R(y)  \geq (1 - \eps_{\mu_s})  \cN\left(c_1 \sqrt{\epsilon} y v, I_d - vv^T + \left(\frac{1}{\kappa} - c_1^2 \eps\right) vv^T\right) \cdot R(y).
\]
Recall that by construction, we have that
\[
    R(y) \geqslant \frac{(1- \epsilon) \cdot G(y)}{(1 - \epsilon_{\mu_s})}. 
\]
Using this, we have
\[
Q_v' \geq (1 - \eps)  \cdot \cN\left(c_1 \sqrt{\epsilon} y v, I_d - vv^T + \left(\frac{1}{\kappa} - c_1^2 \eps\right) vv^T\right) \cdot G(y).
\]
Now the RHS is precisely $(1 - \eps) Q_v(X, y)$ expressed as $(1 - \eps) Q_v(X | Y = y) Q_v(y)$. Therefore we conclude that the following inequality holds point wise. 
\[
    Q_v'(X, y) \geq (1 - \eps) Q_v(X, y).
\]
Now define
\[
E_v(X, y) \coloneqq \frac{1}{\eps} \left(Q_v'(X, y) - (1 - \eps) Q_v(X, y)\right).
\]
Observe that $E_v(X, y) \geq 0$ for all $(X, y)$ and that $E_v(X, y)$ integrates to $1$. Therefore $E_v(X, y)$ is a probability distribution, concluding the proof.
\end{proof}

\subsection{Bounds on the Pairwise Correlation}\label{sec:CorrelationBound}

 Recall that we need to bound $\chi^2_S(Q_v', Q_{v'}')$ where $S = R(y)\cdot G(x)$. To that end, we have the following Lemma. 
\begin{lemma}\label{SQ:gammabeta}
\[
 \gamma = \chi^2_S(Q_v', Q_{v'}') \leq O(| \langle v, v' \rangle |^{4} \sqrt{\kappa}e^{O(1/\eps)}),
\]
\[
    \beta = \chi^2_S(Q_v', Q_{v}') \leq O( \sqrt{\kappa}e^{O(1/\eps)}).
\]
\end{lemma}
\begin{proof}

We begin by expanding out the expression for the $\chi^2$ divergence.
\begin{align*}
    \chi_S(Q_v', Q_{v'}') &= \int_{x, y} \frac{Q_v'(x, y)}{S(x,y)}\frac{Q_{v'}'(x, y)}{S(x, y)}S(x, y) dxdy - 1 = \int_{x, y} \frac{Q_v'(x, y)}{S(x,y)}Q_{v'}'(x, y) dxdy - 1\\
    &= \int_{x, y}\left( \frac{P_{\mu_s, v}(x)  R(y)}{G(x)  R(y)} \right)\cdot \left(P_{\mu_s, v'}(x)  R(y) \right)dx dy  - 1\\ 
    &= \int_\R \chi^2_{\cN(0,I_d)}\left(P_{\mu_s, v}(x) P_{\mu_s, v'}(x)\right) R(y) dy \\
    &\leq  | \langle v, v' \rangle |^{4} \int_\R \left(\chi^2_{\cN(0, 1)}\left(A_{\mu_s}, A_{\mu_s}\right)\right) R(y) dy.
\end{align*}
where the last inequality follows from applying~\cref{lem:Correlation2Moment}. From~\cref{SQ:MainMomentMatching} we have that
\begin{align*}
    \int_\R \left(\chi^2_{\cN(0, 1)}\left(A_{\mu_s}, A_{\mu_s}\right)\right) R(y) dy &= \int_{|\mu_s| \leq \sqrt{\eps}/{10000}}\left(\chi^2_{\cN(0, 1)}\left(A_{\mu_s}, A_{\mu_s}\right)\right) R(y) dy \\ 
    &+ \int_{|\mu_s| > \sqrt{\eps}/{10000}}\left(\chi^2_{\cN(0, 1)}\left(A_{\mu_s}, A_{\mu_s}\right)\right) R(y) dy.
\end{align*}
In the first case, 
\[
\chi^2_{\cN(0, 1)}\left(A_{\mu_s}, A_{\mu_s}\right) \leq e^{O(1/\eps)}.
\]
And in the second case, we have 
\[
\chi^2_{\cN(0, 1)}\left(A_{\mu_s}, A_{\mu_s}\right) \leq O(\sqrt{\kappa}e^{O(\max\{\mu_s^2, 1/\mu_s^2\})}).
\]
Using this, we have 
\[
\int_{|\mu_s| \leq \sqrt{\eps}/{10000}}\left(\chi^2_{\cN(0, 1)}\left(A_{\mu_s}, A_{\mu_s}\right)\right) R(y) dy \leq O(\sqrt{\kappa}) e^{O(1/\eps)},
\]
and, 
\begin{align*}
&\int_{|\mu_s| > \sqrt{\eps}/{10000}}\left(\chi^2_{\cN(0, 1)}\left(A_{\mu_s}, A_{\mu_s}\right)\right) R(y) dy \leq O(\sqrt{\kappa}) \int_{|\mu_s| > \sqrt{\eps}/{10000}} e^{O(\max\{1/\eps, \mu_s^2\})} R(y) dy \\
&\leq O(\sqrt{\kappa} e^{O(1/\eps)}) + \int_{|\mu_s| > \sqrt{\eps}/{10000}} e^{ O(\mu_s^2)} R(y) dy.
\end{align*}
Now,
\begin{align*}
    &\int_{|\mu_s| > \sqrt{\eps}/{10000}} e^{ O(\mu_s^2)} R(y) dy \leq \int_{|\mu_s| > \sqrt{\eps}/{10000}} e^{ O(c_1^2 \eps y^2)} \frac{G(y)}{(1 - \eps_{\mu_s})} dy \\ 
    & \leq \int_{|\mu_s| > \sqrt{\eps}/{10000}}  \frac{e^{ -\Omega(y^2)} \cdot ( 1 - \eps_{\mu_s} + \eps_{\mu_s})}{(1 - \eps_{\mu_s})} dy \\
    &\leq \int_{|\mu_s| > \sqrt{\eps}/{10000}} e^{-\Omega(y^2)} dy + \int_{|\mu_s| > \sqrt{\eps}/{10000}} 9 y^2 e^{-\Omega(y^2)}dy \leq O(1),
\end{align*}
where we used the definition of $R(y)$, and that for $c_1, \eps$ sufficiently small we can indeed take $O(c_1^2 \eps^2) < \frac{1}{2}$. Therefore, putting it all together, 
\[
\chi^2_S(Q_v', Q_{v'}') \leq O(| \langle v, v' \rangle |^{4} \sqrt{\kappa}e^{O(1/\eps)}).
\]
Taking $v = v'$ gives the result for $\beta$.
\end{proof}

To prove the bounds on the divergences, we will require the following facts.
\begin{fact}\label{fact:chimixture}
For distributions $\{X_i\}_{i=1}^n$ absolutely continuous with respect to $D$ and for convex weights $\{w_i\}_{i=1}^n$ we have that 
\[
\chi^2\left(\sum_{i=1}^n w_i X_i, D\right) = \sum_{i = 1}^n \sum_{j = 1}^n w_i w_j \chi_D(X_i, X_j) \leq \sum_{i=1}^n \sum_{j = 1}^n \chi_D(X_i, X_j).
\]
\end{fact}
\begin{fact}\label{fact:chi2gaussian}
\[
\chi^2(\cN(\mu_1,\sigma_1^2),\cN(\mu_2,\sigma_2^2)) =  \frac{\sigma_2^2}{\sigma_1\sqrt{2\sigma_2^2 - \sigma_1^2}}
    e^{\frac{(\mu_1-\mu_2)^2}{2\sigma_2^2 - \sigma_1^2}} - 1.
\]
In particular with respect to the standard Gaussian we obtain
\[
    \chi^2(\cN(\mu_1,\sigma_1^2),\cN(0, 1)) =  \frac{1}{\sigma_1\sqrt{2 - \sigma_1^2}}
    e^{\frac{\mu_1^2}{2 - \sigma_1^2}} - 1.
\]
\end{fact}
\noindent Note that this places a bound on the variance $\sigma_1^2$.
\begin{fact}\label{fact:chicorrelation}
\[
\chi^2_{\mathcal{N}(0,1)}(\cN
(\mu_1, \sigma_1^2),\cN(\mu_2, \sigma_2^2))=\frac{\exp \left(-\frac{\mu_1^2 \left(\sigma_2^2-1\right)+2 \mu_1 \mu_2+\mu_2^2 \left(\sigma_1^2-1\right)}{2 \sigma_1^2 \left(\sigma_2^2-1\right)-2 \sigma_2^2}\right)}{ \sqrt{\sigma_1^2+\sigma_2^2-\sigma_1^2\sigma_2^2}}-1.
\]

\end{fact}
 We are now ready to prove the following lemma.
Since our constructions for $A_{\mu_s}$ are mixtures of Gaussians it suffices to control the pairwise correlations between the mixture components with respect to the standard Gaussian using~\cref{fact:chimixture}. We next prove the following bound on the pairwise correlation for each part of our construction.
\begin{lemma}\label{lemma:mainchisquarecalc}
For $P_{1, \eps_{\mu_s}}$, $P_{2, \eps_{\mu_s}}$ and $P_{3, \eps_{\mu_s}}$ as defined in~\cref{def:Mixture},
    \[ 
        \chi^2(P_{1, \eps_{\mu_s}}, \cN(0, 1)) \leq O(\exp(O(1/\eps))),
    \]
    \[ 
        \chi^2(P_{2, \eps_{\mu_s}}, \cN(0, 1)) \leq O\left(\sqrt{\kappa}e^{O\left(\max\left\{\mu_s^2, 1/\mu_s^2\right\}\right)}\right).
    \]
    \[
        \chi^2(P_{3, \eps_{\mu_s}}, \cN(0, 1)) \leq O\left(\sqrt{\kappa}e^{O\left(\max\left\{\mu_s^2, 1/\mu_s^2\right\}\right)}\right).
    \]

\end{lemma}
\begin{proof}
    From~\cref{fact:chi2gaussian} first note that the signal component in our case has $\chi^2$ divergence
    \[  
        \chi^2\left( \cN(\mu_s, \sigma_s^2), \cN(0, 1) \right) = \frac{1}{\sigma_s\sqrt{2 - \sigma_s^2}}e^{\frac{\mu_s^2}{2 - \sigma_s^2} } - 1 = \frac{1}{\sigma_s\sqrt{2 - \sigma_s^2}}e^{O(\mu_s^2)} - 1 = O\left(\sqrt{\kappa}e^{O\left(\mu_s^2\right)}\right) - 1,
    \]
    since $\sigma_s^2 \leq 0.1$ and $\frac{1}{\sigma_s^2} = \frac{1}{1/\kappa - c_1^2 \eps} \leq \frac{\kappa}{1 - c_1^2 \eps \kappa} \leq O(\kappa)$. We first consider $P_{1, \eps_{\mu_s}}$. The noise components all have variance $1$. So the $\chi^2$ divergence with respect to the standard Gaussian is at most $e^{O\left(\mu_i^2\right)}$ for each of these noise components. There are two types of pairwise correlations: one for those with the signal component and the noise component, and between the noise components. From~\cref{fact:chicorrelation}, the correlation between the signal component and the noise components is at most $O(e^{O(\mu_i^2)})$ and within the noise components it is $O(e^{O(\mu_i \mu_j)})$. Since we know that $\max_i|\mu_i| = O(\frac{1}{\sqrt{\eps}})$, the total $\chi^2$ divergence for this component is, 
    \[  
        O\left(\sqrt{\kappa}e^{O\left(\mu_s^2\right)}\right) + O(e^{O(1/\eps)}) = O(\sqrt{\kappa}) + O(e^{O(1/\eps)}).
    \]
    However since $\eps \kappa < 1$ we have that 
    \[
        O(\sqrt{\kappa}) < O(\sqrt{1/\eps}) < O(e^{O(1/\eps)}),
    \]
    and hence the $\chi^2$ divergence for $P_{1, \eps_{\mu_s}}$ is bounded by $ e^{O(1/\eps)}$.
    
    In $P_{2, \eps_{\mu_s}}$ the noise component means scale as $O(1/\mu_s)$ since $\eps_{\mu_s} \approx \mu_s^2$ in this construction. In terms of variances, we have $\tau^2 = 0.2$ and $\sigma_N^2 > 0$ is always a constant bounded away from $0$. This implies that the individual $\chi^2$ divergence contributions for the noise components scales with the signal mean $\mu_s$ as $O(e^{O(1/\mu_s^2)})$. The pairwise correlation between the noise components scales similarly as $O(e^{O(1/\mu_s^2)})$ and between the signal and the noise component scales as 
    $O(e^{O\left(\max\left\{\mu_s^2, 1/\mu_s^2\right\}\right)}$. Putting it together, we have that the $\chi^2$ divergence for this component scales as 
    \[
        O\left(\sqrt{\kappa}e^{O\left(\mu_s^2\right)}\right) + O(e^{O(1/\mu_s^2)}) + O\left(e^{O\left(\max\left\{\mu_s^2, 1/\mu_s^2\right\}\right)}
        \right) = O\left(\sqrt{\kappa}e^{O\left(\max\left\{\mu_s^2, 1/\mu_s^2\right\}\right)}\right).
    \]
    In $P_{3, \eps_{\mu_s}}$ the two noise components have means one of which scales as $\mu_s$ and the other is much smaller than $\mu_s$ and closer to the origin. Furthermore, by construction the noise components have constant variances. This implies that the pairwise correlation among the noise components and the individual $\chi^2$ divergences for the noise terms scales as as $O(e^{O(\mu_s^2)})$. The signal component is the dominant mean term here as $P_{3, \eps_{\mu_s}}$ deals with large $y$. The pairwise correlation between the signal and the noise components also scales as $O(e^{O(\mu_s^2)})$ since the dominant mean term is that of the signal. Therefore putting all pieces together we have that the $\chi^2$ divergence for this component scaling as 
    \[
        O\left(\sqrt{\kappa}e^{O\left(\mu_s^2\right)}\right) + O(e^{O(\mu_s^2)}) = O\left(\sqrt{\kappa}e^{O\left(\mu_s^2\right)}\right).
    \]
    \end{proof}

\newcommand{\Adv}{\mathrm{Adv}}
\newcommand{\He}{\mathrm{He}}

\section{Low-Degree Lower Bound}\label{sec:LowDegree}

We begin with some background on the low-degree polynomial tests which are very closely related to Statistical Query algorithms \cite{brennan2020statistical}.
\begin{definition}
Consider a testing problem of distinguishing $H_0: y \sim \cQ$ and $H_1: y \sim \cP$ for input $y \in \R^{n \times d}$. The degree-$D$ advantage of the testing problem is defined as 
    \begin{align*}
    \Adv^{\leqslant D} =  \max_{f \in \R[z]_{\leq D}} \frac{\E_{\cP}[f(z)]}{\sqrt{\E_{\cQ}[f(z)^2]}}  = \sqrt{1 + \left(\max_{f \in \R[z]_{\leq D}} \frac{\E_{\cP}[f(z)] - \E_{\cQ}[f(z)]}{\sqrt{\Var_{\cQ}[f(z)]}}\right)^2}
\end{align*}
where the maximum is taken over all degree-D polynomials $f$ over the $n \times d$ dimensional input $y$.
\end{definition}
This method for proving lower bounds against low-degree polynomials is supported by the low-degree conjecture of Hopkins~\cite{hopkins2018statistical}. The conjecture posits that if $\Adv^{\leq D} = O(1)$ for $D = \Omega(\poly \log n)$, then no polynomial-time algorithm exists for the above distinguishing task. We remark that a recent work of \cite{buhai2025quasi} finds a counterexample for this conjecture. Despite the caution offered by this work, the low-degree method remains widely useful for obtaining initial evidence of computational hardness. We refer the reader to \cite{kunisky2019notes, wein2025computational} for more details about this method and the wide applicability of the method in providing evidence for hardness in statistical problems. In the rest of this section, we shall show that the degree-$D$ advantage is bounded for
our testing problem. We start with some facts about Hermite polynomials that we will make extensive use of throughout the rest of this section.

\paragraph{Hermite polynomials}
Hermite polynomials are orthogonal polynomials that form a complete orthogonal basis of the vector space $\cL^2(\R, \cN(0, 1))$ for all functions $f : \R \to \R$ such that $\E_{X \sim \cN(0, 1)}[f^2(X)] < \infty$ (See for example \cite{Szego1939}). We will make use of the normalized version of the probabilist's Hermite polynomials. We denote the $k^{th}$ normalized probabilist's Hermite polynomials as $h_k, k \in \N$. These normalized polynomials $h_k$ are defined using the probabilist's Hermite polynomials $\He_k$ as follows
\[
h_k(x) := \frac{\He_k(x)}{\sqrt{k!}}.
\]
Following \cite{mao2025optimal} we will also make use of multivariate versions of the above polynomials: $H_{\alpha} : \R^{n \times d} \to \R$, indexed by the multi-index $\alpha \in \N^{n \times d}$ and defined as 
\[
H_{\alpha}(z) = \prod_{i=1}^n \prod_{j=1}^d h_{\alpha_{i, j}}(z_{i, j})
\]
Taking $|\alpha| = \sum_{i=1}^n \sum_{j=1}^d \alpha_{i, j}$ we have that $\{H_\alpha\}_{|\alpha| \leq D}$ is a basis for the subspace of polynomials $\R^{n \times d} \to \R$ of degree at most $D$. If $z$ has $\mathrm{iid} \  \cN(0,1)$ entries we further have that these polynomials are orthonormal in the following sense
\[
\E [H_\alpha(z) H_\beta(z)] = \mathbb{1}[\alpha = \beta]
\]
We are now ready to state our lower-bound instance and prove the result.

\subsection{Main result}

\begin{problem}[Robust Linear Testing in Mahalanobis Norm]\label{prob:mainlineartesting}
Given corruption rate $\epsilon \in (0, 1/2)$, $\kappa\geq 1$, signal strength $\alpha \in \R_+$, sample size $n \in \N$, noise variance $\sigma^2 > 0$, dimension $d \in \N$, define $\sigma_y^2 \coloneqq \alpha^2 + \sigma^2$ and consider the following hypothesis testing problem with input samples $\{z_i\}_{i=1}^n \in \R^{n \times (d+1)}$. 
\begin{enumerate}
    \item $H_0$: Null $\cQ:$ Let $X \sim  \cN(0, I_d)$ and $Y \sim  \cN(0, \sigma_y^2)$ be independent.
        Define $Z := (X,Y) \in \mathbb R^{d+1}$.
        Then
        \[
        z_1, \dots, z_n \overset{\mathrm{iid}}{\sim}
        \cQ\defeq\cN \left(0,\operatorname{diag} \left(I_d,\sigma_y^2 \right)\right).
        \]
        
    \item $H_1$: Alternative $\cP$: 
    $v \sim \mathrm{Unif}(\cS^{d-1})$. Conditioned on $v$,
    \[
    z_1, \dots, z_n \overset{\mathrm{iid}}{\sim}  \cP \defeq (1 - \epsilon) \cdot D_v+ \epsilon \cdot E
    \]  
    where $D_v$ is the distribution of $(X, Y)$ for $Y = \langle X, \beta_v \rangle + \eta$ for $X \sim \cN(0, \Sigma_v)$ for $\Sigma_v \succ 0$ with condition number $\kappa$ and $\eta \sim N(0, \sigma^2)$ independent of $X$. Here
    $E$ is an arbitrary adversarial distribution on $\R^{d+1}$ that may depend on $v$. The Mahalanobis norm of $\beta_v$ satisfies
    \[
        \|\Sigma_v^{1/2} \beta_v \| = \alpha.
    \]
    Furthermore, the marginal variance of $Y$ under the inlier distribution $D_v$ matches that under the null:
    \[
        \mathrm{Var}_{D_v}(Y) = \alpha^2 + \sigma^2 = \sigma_y^2    
    \]
\end{enumerate}
\end{problem}

We now define a hard instance for the above problem.

\begin{definition}[Hard Instance for Robust Testing]\label{def:HardInstLowDeg}
    We define a hard instance for the Robust Testing in Mahalanobis Norm Problem. For $\epsilon>0,\kappa,\alpha,n,d,\sigma$, the distribution in $Q$ is uniquely defined. We define the distribution $P$ as, for $v\sim \textup{Unif}(S^{d-1})$ as, $\beta_v = \delta v$, $\Sigma_v = I_d - (1-1/\kappa)vv^T$. Let $E(X,y) = \cN(0_{d+1},\Sigma)$ where,
    \[
    \Sigma =  \begin{bmatrix}
         I_d + \frac{(1-\epsilon)}{\epsilon} \cdot (1 - 1/\kappa)vv^T & -\frac{(1-\epsilon)}{\epsilon} \frac{\delta}{\kappa}v  \\ 
    -\frac{(1-\epsilon)}{\epsilon} \frac{\delta}{\kappa}v^T  & \frac{\delta^2}{\kappa} + \sigma^2
    \end{bmatrix}.
    \]
\end{definition}
\begin{remark}
    In~\cref{lemma:psdness} we show that $\Sigma$ defined above is positive semi-definite. We further note that $\Sigma$ is chosen so that the distribution $P$ has the same first three moments as that of $Q$. We show this in~\cref{lem:MathcMomLowDeg}.
\end{remark}
We will now prove the main result for this section.
\begin{theorem}\label{thm:mainlb}
    Let $\kappa \geq 1$ and $\epsilon > 0$ be such that $\epsilon \kappa \geq 1 - \epsilon$.
    For any $\alpha > 0$, there exists a choice of $\beta_v, \Sigma_v, E(X, y)$ such that the degree-$D$ advantage of~\cref{prob:mainlineartesting} is $1 + o(1)$ for
    \[
        n \ll \frac{1}{\textup{poly}(D)}\min \left(d \epsilon^2 \kappa^2, \epsilon^2 d^2  \right).
    \]
    In particular, the choice for $\beta_v, \Sigma_v, E$ is the same as the one in \cref{def:HardInstLowDeg}.
\end{theorem}
\begin{remark}
Our joint assumption on $\eps, \kappa$ is a direct result of requiring PSD-ness of the covariance matrix of the corruption distribution in our hard instance.
\end{remark}
We claim now that the proof of~\cref{thm:mainlb} follows from the following lemma.

\begin{restatable}{lemma}{mainlowdegreelemma}\label{lemma:mainlowdegree}
The squared degree $D$ advantage of~\cref{prob:mainlineartesting} is at most the following for the instance in~\cref{def:HardInstLowDeg}.
\begin{align*}
    \left(\Adv^{\leq D}\right)^2 - 1 & \lesssim \sum_{\substack{p = 4\\  p \mathrm{         even}}}^D\sum_{m=1}^{p/4} n^m p^{p/2} \left(\frac{4p}{d}\right)^{p/2} \eps^{2m - p} p^{p/2} \sum_{L=0}^{p/2} \left(\frac{\sqrt{d}}{\kappa}\right)^L.    
\end{align*}
\end{restatable}

We first prove~\cref{thm:mainlb}.
\begin{proof}[Proof of~\cref{thm:mainlb}]
    We consider two cases depending on which of $d\eps^2\kappa^2$ and $\eps^2 d^2$ is smaller. 
    \begin{itemize}
        \item Case 1: $\eps^2 d^2 > d\eps^2\kappa^2$, or equivalently $\kappa < \sqrt{d}$. In this case, the innermost geometric sum is diverging and is dominated by the largest term. Therefore the geometric sum can be bounded by $O\left( \left(\frac{\sqrt{d}}{\kappa}\right)^{p/2}\right)$. Using this, the sum reduces to
        \begin{align*}
        \left(\Adv^{\leq D}\right)^2 - 1 & \lesssim \sum_{\substack{p = 4\\  p \text{ even}}}^D \sum_{m=1}^{p/4} n^m p^{p/2} \left(\frac{4p}{d}\right)^{p/2} \eps^{2m - p} p^{p/2}\left(\frac{\sqrt{d}}{\kappa}\right)^{p/2} \\ 
        &= \sum_{\substack{p = 4\\  p \text{ even}}}^D p^{p} \left(\frac{4p}{d}\right)^{p/2} \eps^{-p} \left(\frac{\sqrt{d}}{\kappa}\right)^{p/2} \sum_{m=1}^{p/4} \left(n\epsilon^2\right)^m \\ 
        &\lesssim \sum_{\substack{p = 4\\  p \text{ even}}}^D p^{p} \left(\frac{4p}{d}\right)^{p/2} \eps^{-p} \left(\frac{\sqrt{d}}{\kappa}\right)^{p/2} \left(n\epsilon^2\right)^{p/4} \\
        &=  \sum_{\substack{p = 4\\  p \text{ even}}}^D \left(\frac{16p^6 d n \eps^2}{d^2 \eps^4 \kappa^2}  \right)^{p/4} = \sum_{\substack{p = 4\\  p \text{ even}}}^D \left(\frac{16p^6  n }{d \eps^2 \kappa^2}  \right)^{p/4}.
        \end{align*}
    When $n \ll \frac{1}{\textup{poly}(D)} \left(d \eps^2 \kappa^2\right)$, the right hand side above is $o(1)$, thus enabling us to bound the advantage by $1 + o(1)$.
    \item  Case 2: $\eps^2 d^2 \le d\eps^2\kappa^2$ , or equivalently $\kappa \ge \sqrt{d}$. In this case, the innermost geometric series is bounded by at most a constant. This implies that we have
\begin{align*}
    \left(\Adv^{\leq D}\right)^2 - 1  &\lesssim \sum_{\substack{p = 4\\  p \text{ even}}}^D \sum_{m=1}^{p/4} n^m p^{p/2} \left(\frac{4p}{d}\right)^{p/2} \eps^{2m - p} p^{p/2} \\ 
    & \lesssim \sum_{\substack{p = 4\\  p \text{ even}}}^D p^{p} \left(\frac{4p}{d}\right)^{p/2} \eps^{-p}  \left(n\epsilon^2\right)^{p/4} \\ 
    &=  \sum_{\substack{p = 4\\  p \text{ even}}}^D \left(\frac{16p^6 n \eps^2}{d^2 \eps^4}  \right)^{p/4} = \sum_{\substack{p = 4\\  p \text{ even}}}^D \left(\frac{16p^6  n }{d^2 \eps^2 }  \right)^{p/4}.
\end{align*}
which is $o(1)$ whenever $n \ll \frac{1}{\textup{poly}(D)} \left(d^2 \eps^2 \right)$. Taking the minimum concludes the proof.
    \end{itemize}

\end{proof}

\subsection*{Proof of~\cref{lemma:mainlowdegree}}\label{proof:lem:mainLD}
\mainlowdegreelemma*
The proof of the above lemma is obtained by combining multiple intermediate lemmas. We first outline the proof and then present the mathematical details. As a first step we will express the squared advantage in terms of expectations over Hermite polynomials. To do so, we introduce a new, more general result, stated as~\cref{lemma:generaladvantagetohermite}. 
This result can be seen as a generalization of \cite[Lemma 6.4]{mao2025optimal}. Then, we will compute a closed form expression for these expectations using exponential generating series for our case. The rest of the argument will be primarily combinatorial and will concern the split of the degrees among the samples. We are now ready to prove~\cref{lemma:mainlowdegree}. We split our proof internally into different sections to make it more amendable for the reader to parse the proof.
\begin{proof}
\textbf{Relating the advantage to Hermite coefficients.}
We first apply~\cref{lemma:generaladvantagetohermite} to obtain that the squared advantage is
\begin{align*}
\left(\Adv^{\leq D}\right)^2 = \sum_{p = 0}^D \sum_{\substack{|\alpha| = p \\ \alpha \text{ supp. on } \{1, d+1\}}} \E_{u, u'}[\langle u, u' \rangle^{\sum_{i=1}^n \alpha_{i, 1}}] \cdot \prod_{i=1}^n \left(\E_{(x, y) \sim \cP_{u}'}\left[h_{\alpha_{i,1}}(x) \cdot h_{\alpha_{i, d+1}}\left(\frac{y}{\sqrt{s}}\right)\right]\right)^2. 
\end{align*}    
where $u, u'$ are random vectors drawn from $\cS^{d-1}$ and $\cP'_u$ is the two-dimensional distribution along $u$ and $y$. In our case, $u$ is the hidden direction (which we represent using $v$ in the hard instance) and $s = \delta^2/\kappa + \sigma^2$. Using~\cref{fact:ipunit} about correlation between random vectors we bound the advantage as 
\begin{align*}
   \left(\Adv^{\leq D}\right)^2 \leqslant \sum_{p = 0}^D \sum_{\substack{|\alpha| = p \\ \alpha \text{ supp. on } \{1, d+1\}}} \left( \frac{p}{d} \right)^{\left( \sum_{i=1}^n \alpha_{i, 1}\right)/2} \cdot \prod_{i=1}^n \left(\E_{(x, y) \sim \cP_{u}'}\left[h_{\alpha_{i,1}}(x) \cdot h_{\alpha_{i, d+1}}\left(\frac{y}{\sqrt{s}}\right)\right]\right)^2. 
\end{align*}
whenever  $\sum_{i=1}^n \alpha_{i, 1}$ is even. We will use this implicitly in the remainder of the proof.

By our construction the first three moments of the alternate distribution agree with those of the null distribution. 
Therefore the expectation vanishes for any polynomial up to degree 3. In particular, this implies that the squared advantage is non-zero only when each sample that receives a non-zero degree has an even degree of at least 4. Furthermore, the alternate distribution is symmetric by construction, ensuring that all odd moments also match those of the null distribution and equal zero. Therefore, to bound the advantage, it is sufficient to consider even degrees $p \geq 4$. Thus we have 
\begin{align*}
     &{\left(\Adv^{\leqslant D}\right)}^2 \\
     &\leqslant 1 + \sum_{\substack{p = 4 \\  p \text{ even}}}^D \sum_{\substack{|\alpha| = p \\ \alpha \text{ supp. on } \{1, d+1\}}} \left( \frac{p}{d} \right)^{\left( \sum_{i=1}^n \alpha_{i, 1}\right)/2} \cdot \prod_{i=1}^n \left(\E_{(x, y) \sim \cP_{u}'}\left[h_{\alpha_{i,1}}(x) \cdot h_{\alpha_{i, d+1}}\left(\frac{y}{\sqrt{s}}\right)\right]\right)^2. 
\end{align*}
Henceforth we will use the following notation for convenience: We will replace $\alpha_{i, 1}$ by $k_i$ and $\alpha_{i, d+1}$ by $\ell_i$ and the total degree of sample $i$ by $d_i = k_i + \ell_i$ and define $L = \sum_{i=1}^n \ell_i$ and $K = \sum_{i=1}^n k_i$. We will now make use of the following lemma that enables us to obtain a closed form expression for the expectation of Hermite polynomials.
\paragraph{A closed form for the Hermite coefficients.}
\begin{restatable}{lemma}{hermitecoeff}\label{lemma:hermitecoeff}
    Let $(x, y) \sim D$ for $D = \cN\left(0, \begin{bmatrix}
        \sigma_x^2 & \sigma_{x y} \\ 
        \sigma_{x y} & \sigma_y^2
    \end{bmatrix}\right)$. 
    Then, we have
    \[
        \E_{(x, y) \sim D}\left[h_{k}(x) \cdot h_{\ell}\left(\frac{y}{\sigma_y}\right)\right] = \sqrt{\frac{k!}{\ell!}} \cdot \frac{1}{\left((k - \ell)/2\right)! \cdot 2^{(k -\ell)/2}} \cdot (\sigma_x^2 - 1)^{(k -\ell)/2} \cdot \left(\frac{\sigma_{x y}}{\sigma_y}\right)^{\ell}.
    \]
    for $k \geq \ell$ and $k - \ell$ even and $0$ otherwise.
\end{restatable}
\cref{lemma:hermitecoeff} enables us to compute a closed form equation for the expectation of the Hermite polynomials for our two-dimensional mixture. We provide a proof for this lemma in~\cref{sec:lowdegreefactsproofs} where we utilize exponential generating functions for Hermite polynomials.
We first focus on the terms in the lemma that depend only on the variances and then later deal with the combinatorial coefficients. From~\cref{def:HardInstLowDeg} we recall that we have for the component with weight $(1 - \epsilon)$ that 
\[
    \sigma_x^2 = \frac{1}{\kappa}, \ \ \sigma_{xy} = \frac{\delta}{\kappa}, \ \ \sigma_y^2 = \frac{\delta^2}{\kappa} + \sigma^2,
\]
which implies 
\[
    \frac{\sigma_{xy}}{\sigma_y} = \frac{\delta}{\kappa \sqrt{\delta^2/\kappa + \sigma^2}} \leq \frac{\delta}{\kappa \sqrt{\delta^2/\kappa}} = \frac{1}{\sqrt{\kappa}}.
\]
Now for the component with weight $\epsilon$ we have
\[
    \sigma_x^2 = \frac{1}{\eps} \cdot \left( 1 - \frac{1}{\kappa}(1 - \eps) \right) , \ \ \sigma_{xy} = - \frac{1 - \eps}{\eps} \frac{\delta}{\kappa}, \ \ \sigma_y^2 = \frac{\delta^2}{\kappa} + \sigma^2
\]
which implies that 
\[
    \left|\frac{\sigma_{xy}}{\sigma_y}\right| = \frac{(1 - \eps) \cdot \delta}{\eps \cdot \kappa \cdot \sqrt{\delta^2/\kappa + \sigma^2}} \leq \frac{\delta}{\epsilon \kappa \sqrt{\delta^2/\kappa}} = \frac{1}{\eps \sqrt{\kappa}}.
\]
We further note that $\sigma_x^2 \leqslant \frac{1}{\eps}$ for the second component and we will use this bound going forward.
Putting it together we have the contribution from only the variance terms being
\begin{align*}
&= (1 - \eps) \cdot \left(\frac{1}{\kappa} - 1\right)^{(k - \ell)/2} \cdot \left(\frac{1}{\sqrt{\kappa}}\right)^\ell + \epsilon \cdot \left( \frac{1}{\epsilon}\right)^{(k - \ell)/2} \cdot \left(\frac{1}{\eps \sqrt{\kappa}}\right)^{\ell} \\ 
& \lesssim \epsilon \cdot \left( \frac{1}{\epsilon}\right)^{(k - \ell)/2} \cdot \left(\frac{1}{\eps \sqrt{\kappa}}\right)^{\ell}
\end{align*}
assuming $\kappa \geq 1$ and $\epsilon$ sufficiently small and by using that the absolute value of $1/\kappa - 1$ is at most a constant. Combining this with the combinatorial coefficients, we have
\begin{align*}
    \left|\E_{(x, y) \sim \cP_{u}'}\left[h_k(x) \cdot h_\ell\left(\frac{y}{\sqrt{s}}\right) \right] \right|  \lesssim \underbrace{\sqrt{\frac{k!}{\ell!}} \cdot \frac{1}{\left((k - \ell)/2\right)! \cdot 2^{(k -\ell)/2}}}_{\coloneqq C_{k, \ell}} \cdot \epsilon \cdot \left( \frac{1}{\epsilon}\right)^{(k - \ell)/2} \cdot \left(\frac{1}{\eps \sqrt{\kappa}}\right)^{\ell}
\end{align*}
Squaring and substituting into the advantage equation, we obtain
\begin{align*}
    &\left(\Adv^{\leq D}\right)^2 - 1\\
    &\lesssim \sum_{\substack{p = 4 \\  p \text{ even}}}^D \sum_{\substack{\sum_{i=1}^n d_i = p \\ \forall i \in [n]: \ d_i = k_i + \ell_i \\ \forall i \in [n]: \ k_i \geq \ell_i  \\ \forall i \in [n]: \ k_i - \ell_i \text{ even}}} \left( \frac{p}{d} \right)^{\left( \sum_{i=1}^n k_i\right)/2} \cdot \left( \prod_{i=1}^n \left(C_{k_i, \ell_i}\right)^2 \cdot \left( \eps^2 \cdot \left(\frac{1}{\eps}\right)^{k_i - \ell_i} \cdot \left(\frac{1}{\eps^2 \kappa}\right)^{\ell_i}\right)\right)
\end{align*}
\paragraph{Degree split, done globally.}
We will now concern ourselves with splitting the degrees globally among the samples. We now count the number of ways of dividing a given degree $p$ among the samples so that each sample with a non-zero degree receives degree at least $4$. We observe that this is the same as the cardinality of the following set.
\begin{lemma}[Lemma 6.7 in \cite{mao2025optimal}]\label{lemma:basiccounting} Consider the following set.
    \[
    S(p, m) = \left\{ \beta \in \N^n : \sum_{i = 1}^n \beta_i = p, \Vert \beta \Vert_0 = m, \beta_i \in \{0, 4, 6, 8 \dots \} \forall \ i \in [n] \right\}
\]
Then $|S(p, m)| \leq n^m \cdot p^{p/2}$
\end{lemma}
$|S(p, m)|$ counts the number of ways of distributing degree $p$ among $m$ samples from a total of $n$ samples. For the next parts of the computations, we will fix a specific subset of samples $S$ that has cardinality $m$. The subset $S$ satisfies that $\forall i \in S : d_i > 0$ and $\forall i \notin S: d_i = 0$. We will finally take a summation over all choices of $m$ to finish our computations. 
We will now bound the advantage for a fixed choice of $p, m$ and $\{d_i\}_{i \in S}$. 
Therefore for a fixed $p, m, \{d_i\}_{i \in S}$ we have the contribution to the advantage being at most 
\begin{align*}
    & \sum_{\substack{\forall i \in S: k_i + \ell_i = d_i \\ k_i \geq \ell_i \\ k_i - \ell_i \text{ even} } }  \left( \frac{p}{d} \right)^{\left( \sum_{i=1}^n k_i\right)/2} \cdot \left( \prod_{i \in S} \left(C_{k_i, \ell_i}\right)^2 \cdot \left( \eps^2 \cdot \left(\frac{1}{\eps}\right)^{k_i - \ell_i} \cdot \left(\frac{1}{\eps^2 \kappa}\right)^{\ell_i}\right)\right) \\ 
    & = \sum_{\substack{\forall i \in S: k_i + \ell_i = d_i \\ k_i \geq \ell_i \\ k_i - \ell_i \text{ even} } } \left( \frac{p}{d} \right)^{\left( p - L\right)/2} \cdot \left( \eps^{2m} \cdot \left(\frac{1}{\eps}\right)^{p - 2L} \cdot \left(\frac{1}{\eps^2 \kappa}\right)^{L}\right) \cdot \left( \prod_{i \in S} \left(C_{k_i, \ell_i}\right)^2 \right) \\
    &=  \left(\frac{p}{d}\right)^{p/2} \eps^{2m - p}\sum_{\substack{\forall i \in S: k_i + \ell_i = d_i \\ k_i \geq \ell_i \\ k_i - \ell_i \text{ even} } }  \left( \frac{\sqrt{d}}{ \sqrt{p}\kappa }\right)^{L} \cdot \left( \prod_{i \in S} \left(C_{k_i, \ell_i}\right)^2 \right)\\
    & \leq  \left(\frac{p}{d}\right)^{p/2} \eps^{2m - p}\sum_{\substack{\forall i \in S: k_i + \ell_i = d_i \\ k_i \geq \ell_i \\ k_i - \ell_i \text{ even} } }  \left( \frac{\sqrt{d}}{ \kappa }\right)^{L} \cdot \left( \prod_{i \in S} \left(C_{k_i, \ell_i}\right)^2 \right).
\end{align*}
where the last inequality follows since $p \geq 4$.
\paragraph{Bounding the Combinatorial Coefficients.}
We will now compute a closed form for the combinatorial coefficients that appear in our computations.
\begin{restatable}{lemma}{lemmacombinatorial}\label{lemma:combinatorialcoefficients}
    For a fixed choice of $p, m, \{d_i\}_{i \in S}$, we have
    \[
        \prod_{i \in S} \left(C_{k_i, \ell_i}\right)^2 = \prod_{i \in S} \left(\sqrt{\frac{k_i!}{\ell_i!}} \cdot \frac{1}{\left((k_i - \ell_i)/2\right)! \cdot 2^{(k_i -\ell_i)/2}}\right)^2 \leq 2^p
    \]
\end{restatable}
We provide a proof of~\cref{lemma:combinatorialcoefficients} in~\cref{sec:lowdegreefactsproofs}.
The proof follows by applying Stirling's approximation and univariate calculus.
By an application of~\cref{lemma:combinatorialcoefficients} the contribution to the advantage for a fixed choice of $p, m, \{d_i\}_{i \in S}$ is at most 
\begin{align*}
    \left(\frac{p}{d}\right)^{p/2} \eps^{2m - p}\sum_{\substack{\forall i \in S: k_i + \ell_i = d_i \\ k_i \geq \ell_i \\ k_i - \ell_i \text{ even} } }  \left( \frac{\sqrt{d}}{ \kappa }\right)^{L} \cdot 2^{p} = \left(\frac{4p}{d}\right)^{p/2} \eps^{2m - p}\sum_{\substack{\forall i \in S: k_i + \ell_i = d_i \\ k_i \geq \ell_i \\ k_i - \ell_i \text{ even} } }  \left( \frac{\sqrt{d}}{ \kappa }\right)^{L}
\end{align*}
\paragraph{Degree split, within the samples.}
We now observe that the terms inside the summation only depends on $L$ and not the individual $\{\ell_i\}_{i=1}^n$ themselves. Recall that since $k_i \geq \ell_i$, each $\ell_i \leq d_i/2$ and thus $L \leq p/2$. Using this we rewrite the summation in terms of $L$ as follows:
\begin{align*}
    \sum_{\substack{\forall i \in S: k_i + \ell_i = d_i \\ k_i \geq \ell_i \\ k_i - \ell_i \text{ even} } } \left( \frac{\sqrt{d}}{ \kappa }\right)^{L} = \sum_{L = 0}^{p/2}\sum_{\substack{\sum_{i \in S} \ell_i = L \\ \ k_i + \ell_i = d_i \\ \ k_i \geq \ell_i \\  \ k_i - \ell_i \text{ even}}} \left( \frac{\sqrt{d}}{ \kappa }\right)^{L} = \sum_{L = 0}^{p/2}  \left(\left( \frac{\sqrt{d}}{ \kappa }\right)^{L} \cdot \left( \sum_{\substack{\sum_{i \in S} \ell_i = L \\ \ k_i + \ell_i = d_i \\ \ k_i \geq \ell_i \\ \ k_i - \ell_i \text{ even}}} 1 \right)\right)
\end{align*}
The contribution of the innermost summation is at most the number of ways of splitting a fixed $\{d_i\}_{i \in S}$ into $k_i$ and $\ell_i$ subject to the constraints on $k_i$ and $\ell_i$. Since $\ell_i$ is determined once $k_i$ is fixed and $k_i \geq d_i/2$, the number of choices for each chosen sample $i$ is at most $\left(d_i/2 + 1\right)$ and we have totally $m$ samples. Therefore, across all $m$ chosen samples, the number of choices is at most 
\[
\prod_{i \in S}\left(\frac{d_i}{2} + 1 \right) \leq \prod_{i \in S} p \leq p^m \leq p^{p/4} 
\]
where we used $p \geq 4m$. Therefore for a fixed choice of $p, m, \{d_i\}_{i \in S}$ we have that the above quantity is bounded by
\begin{align*}
    \left(\frac{4p}{d}\right)^{p/2} \eps^{2m - p} p^{p/2} \sum_{L=0}^{p/2} \left(\frac{ \sqrt{d}}{\kappa}\right)^L.
\end{align*}
\paragraph{Putting it together.}
Combining the above steps, we have that the total squared advantage over all choices of $p, m$ and the $\{d_i\}_{i \in S}$ as follows. 
\begin{align*}
\left(\Adv^{\leq D}\right)^2 - 1 \lesssim \sum_{\substack{p = 4\\  p \text{even}}}^D \sum_{m=1}^{p/4} n^m p^{p/2} \left(\frac{4p}{d}\right)^{p/2} \eps^{2m - p} p^{p/2} \sum_{L=0}^{p/2} \left(\frac{\sqrt{d}}{\kappa}\right)^L.     
\end{align*}
This concludes the proof of~\cref{lemma:mainlowdegree}.

\end{proof}

\subsection{A general testing problem.}
We begin with a more general testing problem, of which our original problem is a special case.

\begin{problem}\label{prob:lowdegtwodimtesting}
Let $\nu$ be a distribution over $\R^2$. Let $s > 0$. We define the following null and planted distributions.  
\begin{itemize}
    \item Under $\cQ$, we observe i.i.d.\ samples $z_1,\dots,z_n \in \R^d \times \R$ with
    \[
        z_i = (x_i, y_i), \qquad
        x_i \sim \cN(0, I_d), \quad y_i \sim \cN(0, s),
    \]
    where $x_i$ and $y_i$ are independent.
    \item Under \(\cP\), we first draw \(v \sim \mathrm{Unif}(\cS^{d-1})\). Conditional on \(v\), we then draw i.i.d.\ samples \(z_1, z_2, \dots, z_n\) as follows:
    \begin{align*}
        &(a_i, b_i) \sim \nu \\
        g_i &\sim \cN(0, I_d) \quad \text{independently,}\\
        x_i &= a_i\, v + (I_d - v v^\top) g_i,\\
        y_i &= b_i,
    \end{align*}
    and set \(z_i = (x_i, y_i)\).
\end{itemize}
Equivalently, under \(\cP\), the pair \((\langle x_i, v\rangle, y_i)\) has distribution \(\nu\), while the component of $x_i$ orthogonal to $v$ is standard Gaussian.  
\end{problem}
\begin{lemma}\label{lemma:generaladvantagetohermite}
Consider the distribution $\cP$ in~\cref{prob:lowdegtwodimtesting}. Suppose the first $D$ moments of $\nu$ are finite. For $\alpha \in \N^{n \times (d+1)}$, let $|\alpha| = \sum_{i=1}^n \sum_{j=1}^{d+1} \alpha_{i, j}$.
Then we have
\begin{align*} 
    {\left(\Adv^{\leqslant D}\right)}^2 &= \sum_{p = 0}^D \sum_{\substack{|\alpha| = p \\ \alpha \textup{ supp. on } \{1, d+1\}}} \E_{u, u'}[\langle u, u' \rangle^{\sum_{i=1}^n \alpha_{i, 1}}] \cdot \prod_{i=1}^n \left(\E_{(x, y) \sim \nu}\left[h_{\alpha_{i,1}}(x) \cdot h_{\alpha_{i, d+1}}\left(\frac{y}{\sqrt{s}}\right)\right]\right)^2. 
\end{align*}
for $u, u'$ uniformly chosen random vectors from the unit sphere.
\end{lemma}

\begin{proof}
We begin by recalling that the degree $D$ advantage is defined as 
\begin{align*}
    \Adv^{\leqslant D} = \max_{f \in \R[z]_{\leq D}} \frac{\E_{\cP}[f(z)]}{\sqrt{\E_{\cQ}[f(z)^2]}}
\end{align*}
It is a well-known fact that the orthogonal projection of the likelihood ratio $L$ to the set of all polynomials of degree at most $D$ maximizes the above quantity and that the optimum objective value is $\Vert L^{\leq D} \Vert$. 
Here the norm is with respect to the following inner product space, defined using the null distribution. 
Below, we follow the techniques introduced in \cite{mao2025optimal} for bounding the advantage. In particular, our result can be viewed as a generalization of their result for broader classes of testing problems.

For functions $f, g: \R^{n \times (d+1)} \to \R$, we define the inner product $\langle f, g \rangle \coloneqq \E_{z \sim \cQ}[f(z)g(z)]$ and the associated norm $\Vert f \Vert = \sqrt{\langle f, f \rangle}$. 
We then have that 
\begin{align*}
    \Adv^{\leqslant D} = \max_{f \in \R[z]_{\leqslant D}} \frac{\langle f, L \rangle}{\Vert f \Vert} = \frac{\langle L^{\leqslant D}, L \rangle}{\Vert L^{\leqslant D} \Vert} = \Vert L^{\leqslant D}\Vert,
\end{align*} 
where $f^{\le D}$ is the orthogonal projection (with respect to $\langle \cdot,\cdot \rangle$) of a function $f$ onto $\R[z]_{\le D}$ (the subspace of polynomials $\R^{n \times (d+1)} \to \R$ of degree at most $D$), and $L(z) := \frac{d \cP}{d \cQ}(z)$ is the likelihood ratio. 
Let $\cP_u$ denote the distribution of $z \sim \cP$ conditioned on a particular choice of $u$, and let $L_u(z) := \frac{d \cP_u}{d \cQ}(z)$ so that $L(z) = \E_{u \sim \text{Unif}(\cS^{d-1})} L_u(z)$. We have the squared advantage 
\begin{equation}\label{eq:L-E}
\left(\Adv^{\leq D}\right)^2 = \|L^{\le D}\|^2 = \langle L^{\le D}, L^{\le D} \rangle = \langle \E_u L_u^{\le D}, \E_{u'} L_{u'}^{\le D} \rangle = \E_{u,u'} \langle L_u^{\le D}, L_{u'}^{\le D} \rangle
\end{equation}
where $u,u' \sim \text{Unif}(\cS^{d-1})$ are independent.
We fix $u, u'$ and then focus on computing $\langle L_u^{\leq D}, L_{u'}^{\leqslant D} \rangle$. As in \cite{mao2025optimal}, we will work in the orthonormal basis where $u = e_1$ and $u' = \tau e_1 + \sqrt{1 - \tau^2} e_2$ for $\tau \coloneqq \langle u, u' \rangle$. 
Defining 
\begin{align*}
H_\alpha(z) \coloneqq \prod_{i=1}^n \left(h_{\alpha_{i, d+1}}\left(\frac{z_{i, d+1}}{ \sqrt{s}} \right)\cdot \left( \prod_{j=1}^{d} h_{\alpha_{i, j}}(z_{i, j}) \right)\right),
\end{align*}
we expand the projected likelihood in the Hermite basis as follows.
\begin{align*}
    \langle L_u^{\leq D}, L_{u'}^{\leq D} \rangle = \sum_{|\alpha| \leq D}\sum_{|\beta| \leq D} \langle c_{\alpha, u} H_\alpha(z),  c_{\beta, u'}H_\beta(z)  \rangle = \sum_{|\alpha| \leq D}\sum_{|\beta| \leq D} c_{\alpha, u} c_{\beta, u'} \langle H_\alpha(z), H_\beta(z) \rangle
\end{align*}
for multi-indices $\alpha$ and  $\beta$. We claim that $H_{\alpha}(.)$ is orthonormal with respect to $\cQ$ (See~\cref{lemma:orthonormalbasis} for a proof).
As a consequence, we have 
\begin{align*}
    \langle L_u^{\leq D}, L_{u'}^{\leq D} \rangle = \sum_{|\alpha| \leq D} c_{\alpha, u} c_{\alpha, u'}.
\end{align*}
We observe that for $\alpha$ such that $|\alpha| \leq D$ we have
\[
    c_{\alpha, u} \coloneqq \langle L_u, H_\alpha \rangle = \E_{z \sim \cP_u}[H_\alpha(z)].
\]
Using this, we therefore have
\begin{align*}
    \left(\Adv^{\leqslant D}\right)^2 = \E_{u, u'} \langle L_u^{\leq D}, L_{u'}^{\leq D} \rangle = \E_{u, u'}\left[\sum_{|\alpha| \leq D} c_{\alpha, u} c_{\alpha, u'} \right] = \sum_{|\alpha| \leq D}\E_{u, u'}\left[ c_{\alpha, u} c_{\alpha, u'}\right].
\end{align*}
Recalling the definition of $H_{\alpha}$ we have
\begin{align*}
c_{\alpha, u} = \E_{z \sim \cP_u}\left[\prod_{i=1}^n \left(h_{\alpha_{i, d+1}}\left(\frac{z_{i, d+1}}{ \sqrt{s}} \right)\cdot \left( \prod_{j=1}^{d} h_{\alpha_{i, j}}(z_{i, j}) \right)\right)\right].    
\end{align*}
Since $u = e_1$, observe that the above term is only non-zero if the multi-indices are supported on the first and the last coordinate (respectively the first and last columns of $\alpha$)\footnote{In the alternate distribution, the signal resides in this two-dimensional subspace. The remaining $d-1$ dimensional subspace is Gaussian noise.} since the expectation under the standard Gaussian of $h_k(x)$ vanishes for $k \neq 0$. 
This allows for further simplification and we can express $c_{\alpha, u}$ in this more compact form.
\begin{align*}
   c_{\alpha, u} = \prod_{i=1}^n \left(\E_{(x, y) \sim \nu} \left[h_{\alpha_{i, 1}}(x) \cdot h_{\alpha_{i, d+1}}\left(\frac{y}{\sqrt{s}}\right)\right]    \right),
\end{align*}
where $\cP_{u}'$ is the two-dimensional distribution of $\cP$ along $u$ (and thus $e_1$ in this case) and the last coordinate. That is, this expression for $\alpha$ is supported only on the two relevant columns and for the remaining $\alpha$ the Hermite coefficient is zero.

We notice that as a direct consequence, even for $c_{\alpha, u'}$ we only need to calculate the coefficients where $\alpha$ is supported on the first and last columns. There could be more non-zero coefficients for $u'$: however they do not contribute to the advantage since they are multiplied by zero in the above inner product expression. Thus, expressing $u'$ in terms of its component along $u = e_1$ and $e_2$, we obtain

\begin{align*}
    c_{\alpha, u'} =  \prod_{i=1}^n \left(\mathbb{E}_{\substack{(x, y) \sim \nu \\ z \sim \cN(0, 1)}} \left[ h_{\alpha_{i,1}}\left(\tau x + \sqrt{1 - \tau^2}\cdot z\right) \cdot h_{\alpha_{i, d+1}}\left(\frac{y}{\sqrt{s}}\right) \right] \right).
\end{align*}
Now,
\begin{align*}
    &\mathbb{E}_{\substack{(x, y) \sim \nu \\ z \sim \cN(0, 1)}} \left[ h_{\alpha_{i,1}}\left(\tau x + \sqrt{1 - \tau^2} z\right) \cdot h_{\alpha_{i, d+1}}\left(\frac{y}{\sqrt{s}}\right) \right]\\
    &= \mathbb{E}_{\substack{(x, y) \sim \nu}} \left[ h_{\alpha_{i, d+1}}\left(\frac{y}{\sqrt{s}}\right) \E_{z \sim \cN(0,1)}\left[h_{\alpha_{i,1}}\left(\tau x + \sqrt{1 - \tau^2} z\right) \mid x, y \right] \right]
\end{align*}
by the law of iterated expectations. Now, since $\tau \in [-1, 1]$ and because Hermite polynomials are eigenfunctions of the Gaussian Noise operator we obtain that
\begin{align*}
    &\mathbb{E}_{\substack{(x, y) \sim \nu}} \left[ h_{\alpha_{i, d+1}}\left(\frac{y}{\sqrt{s}}\right) \E_{z \sim \cN(0,1)}\left[h_{\alpha_{i,1}}\left(\tau x + \sqrt{1 - \tau^2} z\right) \mid x, y \right] \right] \\
    &= \mathbb{E}_{\substack{(x, y) \sim \nu}} \left[ h_{\alpha_{i, d+1}}\left(\frac{y}{\sqrt{s}}\right) \tau^{\alpha_{i, 1}}h_{\alpha_{i, 1}}(x) \right]\\
    &= \tau^{\alpha_{i, 1}} \cdot \mathbb{E}_{(x, y) \sim \cP_{u}' }\left[h_{\alpha_{i, 1}}(x) \cdot h_{\alpha_{i, d+1}}\left(\frac{y}{\sqrt{s}}\right)\right]. 
\end{align*}
Therefore we have 
\[
    c_{\alpha, u'} = \prod_{i=1}^n  \left(\tau^{\alpha_{i, 1}} \cdot \E_{(x, y) \sim \nu}  \left[h_{\alpha_{i, 1}}(x) \cdot h_{\alpha_{i, d+1}}\left(\frac{y}{\sqrt{s}}\right)\right]    \right),
\]
and thus 
\[
    c_{\alpha, u} \cdot c_{\alpha, u'} = \left(\langle u, u' \rangle^{\sum_{i=1}^n \alpha_{i, 1}}\right) \cdot  \prod_{i=1}^n  \E_{(x, y) \sim \nu}  \left[h_{\alpha_{i, 1}}(x) \cdot h_{\alpha_{i, d+1}}\left(\frac{y}{\sqrt{s}}\right)\right]^2.
\]
Putting it together, we have that the squared advantage is:
\begin{align*}
     &\sum_{|\alpha| \leq D} \E_{u, u'}\left[ c_{\alpha, u} c_{\alpha, u'} \right] \\
     &= \sum_{|\alpha| \leq D} \E_{u, u'}\left[ \langle u, u' \rangle^{\sum_{i=1}^n \alpha_{i, 1}} \right] \cdot \left(\prod_{i=1}^n  \E_{(x, y) \sim \nu}  \left[h_{\alpha_{i, 1}}(x) \cdot h_{\alpha_{i, d+1}}\left(\frac{y}{\sqrt{s}}\right)\right]^2\right) \\ 
     &= \sum_{p = 0}^D \sum_{\substack{|\alpha| = p \\ \alpha \text{ supp. on } \{1, d+1\}}} \E_{u, u'}[\langle u, u' \rangle^{\sum_{i=1}^n \alpha_{i, 1}}] \cdot \prod_{i=1}^n \left(\E_{(x, y) \sim \nu}\left[h_{\alpha_{i,1}}(x) \cdot h_{\alpha_{i, d+1}}\left(\frac{y}{\sqrt{s}}\right)\right]\right)^2. 
\end{align*}
Observing that $\cP'_u$ is $\nu$ concludes our proof.
\end{proof}

\subsection{Missing facts and proofs for low-degree lower bound lemmas}\label{sec:lowdegreefactsproofs}

\begin{fact}[Almost orthogonality of random vectors (Lemma 6.5 in \cite{mao2025optimal})]\label{fact:ipunit}
    Let $u, u'$ be independent random vectors drawn from $\cS^{d-1}$. For $q \in \N$, if $q$ is odd, then $\E_{u, u'}[\langle u, u' \rangle^q] = 0$ and if $q$ is even
    \[
        \E_{u, u'}[\langle u, u' \rangle^q] \leqslant \left( \frac{q}{d}\right)^{q/2}
    \]
\end{fact}

\noindent \textbf{Ornstein-Uhlenbeck Operator}. 
For some $\rho > 0$, the Ornstein-Uhlenbeck operator $U_\rho$ maps some distribution $F$ on $\R$ to the distribution of the random variable $\rho X + \sqrt{1 - \rho^2} Z$, where $X \sim F$ and $Z \sim \cN(0, 1)$ is independent of $X$. A useful property of Ornstein–Uhlenbeck operator is that it operates diagonally with respect to
Hermite polynomials.
\begin{fact}[See e.g., \cite{o2014analysis}]
\label{fact:ouhermite}
For a Hermite polynomial $h_k$, and for any distribution $F$ on $\R$, $\rho \in (0, 1)$, we have 
\[
    \E_{X \sim U_\rho F}[h_k(x)] = \rho^{i} \E_{X \sim F}[h_k(X)]
\]
\end{fact}

\begin{lemma}[PSD-ness of the corruption covariance]\label{lemma:psdness}
    For $\epsilon \kappa \geq 1 - \epsilon$, we have
    \[
        \begin{bmatrix}
         \frac{1}{\epsilon} \cdot (1 - \frac{1}{\kappa} ( 1- \epsilon))& -\frac{(1-\epsilon)}{\epsilon} \frac{\delta}{\kappa}  \\ 
        -\frac{(1-\epsilon)}{\epsilon} \frac{\delta}{\kappa}  & \frac{\delta^2}{\kappa} + \sigma^2
        \end{bmatrix} \succeq 0.
    \]
    As a result $\Sigma$ defined in~\cref{def:HardInstLowDeg} is a valid covariance matrix.
\end{lemma}
\begin{proof}
For $\Sigma$ to be a valid covariance matrix, it needs to be positive semi-definite. 
Our choice of $\Sigma$ is $I_{d-1}$ in the subspace orthogonal to $v$ and the last coordinate. Therefore, it suffices for the following $2 \times 2$ sub-matrix  to be PSD:
\[
    \begin{bmatrix}
     \frac{1}{\epsilon} \cdot (1 - \frac{1}{\kappa} ( 1- \epsilon))& -\frac{(1-\epsilon)}{\epsilon} \frac{\delta}{\kappa}  \\ 
    -\frac{(1-\epsilon)}{\epsilon} \frac{\delta}{\kappa}  & \frac{\delta^2}{\kappa} + \sigma^2
    \end{bmatrix}
\]

    For any symmetric $2 \times 2$ matrix to be PSD, it suffices that both diagonal elements and the determinant are non-negative.  We note that the condition on $\epsilon$ and $\kappa$ ensures non-negativity of the diagonal elements. We will prove the result for the case $\sigma^2 = 0$ as $\sigma^2 > 0$ only makes it easier to satisfy the non-negativity of the determinant. For this case
    by a direct calculation we have that 
    \begin{align*}
    &\frac{1}{\epsilon}\left(1 - \frac{1-\epsilon}{\kappa}\right) \frac{\delta^2}{\kappa} \geqslant \left(\frac{1-\epsilon}{\epsilon}\right)^2 \frac{\delta^2}{\kappa^2} \iff \\ 
    &\left(1 - \frac{1-\epsilon}{\kappa}\right) \geqslant \frac{\left(1-\epsilon\right)^2}{\epsilon \kappa} \iff \\
    &\left(1 - \frac{1-\epsilon}{\kappa}\right) - \frac{\left(1-\epsilon\right)^2}{\epsilon \kappa} \geqslant 0 \iff \\ 
    &\frac{1}{\epsilon \kappa} \left(  \epsilon (\kappa - (1-\epsilon)) - (1-\epsilon)^2 \right) \geqslant 0 \iff \\ 
    &\epsilon \kappa - \epsilon + \epsilon^2 - 1 - \epsilon^2 + 2\epsilon \geqslant 0 \iff \\ 
    & \epsilon \kappa \geqslant 1 - \epsilon
\end{align*} 
\end{proof}
\begin{lemma}[Moment Matching]\label{lem:MathcMomLowDeg}
    For distributions $\cP$ and $\cQ$ as defined in~\cref{def:HardInstLowDeg}, the first three moments of $\cP$ agree with $\cQ$.
\end{lemma}
\begin{proof}
     For $y = \langle X, \beta \rangle + \eta$ where $X \sim \cN(0, \Sigma), \eta \sim \cN(0, \sigma^2)$ independent of $X$, the joint distribution $(X, y) \sim N\left(0, \Sigma_{Xy}\right)$ for 
    \[
        \Sigma_{Xy} = \begin{bmatrix}
                    \Sigma &\Sigma\beta \\ 
                    \beta^T \Sigma & \beta^T \Sigma \beta + \sigma^2
                    \end{bmatrix}.
    \]
    Therefore, our distribution looks like,
\begin{align*}
    \cP(X, y) = (1 - \epsilon) \cdot \cN\left(0, \begin{bmatrix}
        I_d - (1 - 1/\kappa)vv^T & \frac{\delta}{\kappa} v \\ 
    \frac{\delta}{\kappa} v^T & \frac{\delta^2}{\kappa} + \sigma^2
    \end{bmatrix}\right) + \epsilon \cdot \cN(0, \Sigma_2).
\end{align*}
We pick $\Sigma$ to match the first three moments with the null distribution $\cQ$.
Our first observation is that we require moment matching only in the subspace spanned by the last coordinate $y$ and the direction $v$. 
Furthermore, since $E(X, y)$ is a centered Gaussian, the first and third moments already match with $\cQ$. Therefore it suffices to match only the second moment and pick $\Sigma_2$ such that 
\begin{align*}
    (1 - \epsilon) \cdot \begin{bmatrix}
        I_d - (1 - 1/\kappa)vv^T & \frac{\delta}{\kappa} v \\ 
    \frac{\delta}{\kappa} v^T & \frac{\delta^2}{\kappa} + \sigma^2
    \end{bmatrix} + \epsilon \cdot \Sigma = \begin{bmatrix}
        I_d & 0_d \\ 
    0_d^T & \frac{\delta^2}{\kappa} + \sigma^2
    \end{bmatrix}
\end{align*}
By direct calculation, we have 
\begin{align*}
    \Sigma = \begin{bmatrix}
         I_d + \frac{(1-\epsilon)}{\epsilon} \cdot (1 - 1/\kappa)vv^T & -\frac{(1-\epsilon)}{\epsilon} \frac{\delta}{\kappa}v  \\ 
    -\frac{(1-\epsilon)}{\epsilon} \frac{\delta}{\kappa}v^T  & \frac{\delta^2}{\kappa} + \sigma^2
    \end{bmatrix}
\end{align*}
This is a valid covariance for $\eps \kappa \geq 1 - \eps$ as established in~\cref{lemma:psdness}.

\end{proof}

\begin{lemma}[Orthonormality of the Hermite Basis]\label{lemma:orthonormalbasis}
    $H_\alpha(z)$ is orthonormal with respect to $\cQ$.
\end{lemma}
\begin{proof}
\begin{align*}
    &\langle H_\alpha(z), H_\beta(z) \rangle \\
    &= \E_{z \sim \cQ}\Bigg[ \prod_{i=1}^n \left(h_{\alpha_{i, d+1}}\left(\frac{z_{i, d+1}}{ \sqrt{s}} \right) \cdot \prod_{j=1}^{d} h_{\alpha_{i, j}}(z_{i, j}) \right) \cdot
     \qquad \prod_{i=1}^n \left(h_{\beta_{i, d+1}}\left(\frac{z_{i, d+1}}{ \sqrt{s}} \right)\cdot \prod_{j=1}^{d} h_{\beta_{i, j}}(z_{i, j}) \right) \Bigg] \\ 
    &= \E_{z \sim \cQ}\Bigg[ \prod_{i=1}^n \left(h_{\alpha_{i, d+1}}\left(\frac{z_{i, d+1}}{ \sqrt{s}}  \right) h_{\beta_{i, d+1}}\left(\frac{z_{i, d+1}}{ \sqrt{s}} \right) \cdot \prod_{j=1}^{d} h_{\alpha_{i, j}}(z_{i, j}) h_{\beta_{i, j}}(z_{i, j}) \right)\Bigg] \\
    &= \E_{z_{i,d+1} \sim \cN(0, s)}\left[\prod_{i=1}^n h_{\alpha_{i, d+1}}\left(\frac{z_{i, d+1}}{ \sqrt{s}}  \right) h_{\beta_{i, d+1}}\left(\frac{z_{i, d+1}}{ \sqrt{s}} \right)\right] \cdot
     \qquad \E_{z_{i,j} \sim \cN(0, I_d)}\left[\prod_{i=1}^n \prod_{j=1}^d h_{\alpha_{i, j}}(z_{i, j}) h_{\beta_{i, j}}(z_{i, j})\right]\\
    &= \E_{z'_{i,d+1} \sim \cN(0, 1)}\left[\prod_{i=1}^n h_{\alpha_{i, d+1}}\left(z_{i, d+1}'  \right) h_{\beta_{i, d+1}}\left(z_{i, d+1}' \right)\right] \cdot
     \qquad \E_{z_{i,j} \sim \cN(0, I_d)}\left[\prod_{i=1}^n \prod_{j=1}^d h_{\alpha_{i, j}}(z_{i, j}) h_{\beta_{i, j}}(z_{i, j})\right]\\
    &= \mathbb{1}\left[{\alpha = \beta}\right]
\end{align*}

\end{proof}
We now provide a proof for obtaining a closed form for the Hermite coefficients.
\hermitecoeff*
\begin{proof}\label{sec:hermiteclosedformproof}
We note that a direct computation using Wick’s probability theorem \cite{wick1950evaluation}, also known as Isserlis’ celebrated theorem for Gaussian random variables \cite{isserlis1918formula} is in principle possible, but the resulting calculations appear tedious. To circumvent this we will try to obtain the above closed form expression for the expectations as coefficients of certain terms in an infinite sum. We begin by utilizing the exponential generating function for Hermite polynomials \cite{roman1984umbral} (See also the Wikipedia entry \cite{WikipediaHermitePolynomials}). We have for every $t \in \R$ that
\begin{align*}
    e^{xt - t^2/2} = \sum_{k=0}^\infty h_k(x) \frac{t^k}{\sqrt{k!}}
\end{align*}
For $\tilde{y} \coloneqq y / \sigma_y$, we similarly have for every $u \in \R$ that
\begin{align*}
        e^{\tilde{y}u - u^2/2} = \sum_{\ell=0}^\infty h_{\ell}(\tilde{y}) \frac{u^{\ell}}{\sqrt{\ell!}}
\end{align*}
Multiplying both equations we obtain
\begin{align*}
    \exp\left( xt - t^2/2 + \tilde{y}u - u^2/2\right) &= \sum_{k=0}^\infty \sum_{\ell = 0}^{\infty} \left(h_k(x) \frac{t^k}{\sqrt{k!}} \right) \cdot \left(h_{\ell}(\tilde{y}) \frac{u^{\ell}}{\sqrt{\ell!}}\right) \\
    &= \sum_{k=0}^\infty \sum_{\ell = 0}^{\infty} \left(h_k(x) h_{\ell}(\tilde{y}) \right) \frac{t^k}{\sqrt{k!}} \frac{u^{\ell}}{\sqrt{\ell!}}
\end{align*}
We will take the expectation on both sides. Now computing the expectation on the LHS we have that
\begin{align*}
    \E\left[ \exp\left( xt - t^2/2 + \tilde{y}u - u^2/2\right) \right] &= \E[ \exp(-t^2/2 - u^2/2) \cdot \exp\left( xt + \tilde{y}u \right)] \\
    &= e^{-\frac{1}{2}\left(t^2 + u^2\right)} \cdot \E\left[e^{\left(tx + u \tilde{y} \right)}\right]
\end{align*}
Using the following fact about Gaussian distributions 
\begin{align*}
    \E_{X \sim \cN(0, \Sigma)}\left[e^{\langle t, X \rangle}\right] = \exp\left( \frac{t^T \Sigma t}{2} \right)
\end{align*}
allows us to compute a closed form for the expectation. Therefore by a direct computation we have the exponent being
\begin{align*}
    [t, u]^T \begin{bmatrix}
        \sigma_x^2 & \sigma_{x\tilde{y}} \\ 
        \sigma_{x\tilde{y}} & \sigma_{\tilde{y}}^2 
    \end{bmatrix} 
    \begin{bmatrix}
        t \\ 
        u 
    \end{bmatrix} = \frac{\sigma_x^2 t^2 + 2 t u \sigma_{x\tilde{y}} + u^2 \sigma_{\tilde{y}}^2}{2}
\end{align*}
Therefore we have that the LHS evaluates after taking expectations to the following.
\begin{align*}
     \E\left[ \exp\left( xt - t^2/2 + \tilde{y}u - u^2/2\right) \right] &= \exp\left(t^2 \cdot \frac{\sigma_x^2 - 1}{2} + u^2 \cdot\frac{\sigma_{\tilde{y}}^2 - 1}{2} + tu \sigma_{x \tilde{y}} \right) \\
     &= \exp\left(t^2 \cdot\frac{\sigma_x^2 - 1}{2} + tu \sigma_{x \tilde{y}}\right) \\ 
     &= \exp\left( t^2 \cdot \frac{\sigma_x^2 - 1}{2}\right) \exp\left(t u \sigma_{x \tilde{y}} \right)
\end{align*}
We now observe that we need to read off the coefficient for $t^k u^{\ell}$ by expanding out the above sum to get the expectation of the product of the Hermite polynomials. By using the Taylor series expansion for $e^x$ we have
\begin{align*}
    \sum_{k=0}^\infty \sum_{\ell = 0}^{\infty} \E\left[h_k(x) h_{\ell}(\tilde{y})\right] \cdot  \frac{t^k}{\sqrt{k!}} \frac{u^{\ell}}{\sqrt{\ell!}} = \left(\sum_{i=0}^{\infty} \frac{\left(\frac{t^2}{2}\right)^i \left(\sigma_x^2 - 1\right)^i}{i!}\right) \cdot \left(\sum_{j=0}^\infty \frac{(t u \sigma_{x \tilde{y}})^j}{j!} \right)
\end{align*}
By equating terms on both sides, we require $j = \ell$. Since  $k = 2i + j = 2i + \ell$, we require that $i = (k - \ell)/ 2$. Now since $i \geq 0$ we also need $k \geqslant \ell$ and $k - \ell$ even as well for the coefficients to be non-zero. Finally the result follows by picking the coefficient of $t^k u^{\ell}$ on the LHS.
\end{proof}
We now provide the proof of our bound on the combinatorial coefficients.
\lemmacombinatorial*
\begin{proof}
We begin by defining
\[
     T_i \coloneqq \sqrt{\frac{k_i!}{\ell_i!}} \cdot \frac{1}{\left((k_i - \ell_i)/2\right)! \cdot 2^{(k_i -\ell_i)/2}} 
\]
Taking natural logarithm we have 
\begin{align*}
    \log T_i = \frac{1}{2} \log k_i ! - \frac{1}{2} \log \ell_i ! - \log ((k_i - \ell_i)/2)! - (k_i - \ell_i)/2 \cdot \log 2
\end{align*}
Using Stirling's approximation for $d_i, \ell_i, k_i$
\begin{align*}
    \log T_i & \approx \frac{1}{2} \left(k_i \log k_i - k_i\right) - \frac{1}{2} \left(\ell_i \log \ell_i - \ell_i\right) \\
    & \quad - \left(\frac{1}{2}(k_i - \ell_i) \log (k_i - \ell_i)/2 - (k_i - \ell_i)/2 \right) \\
    & \quad - (k_i - \ell_i)/2 \cdot \log 2 \\ 
    & = \frac{1}{2}(k_i \log k_i - \ell_i \log \ell_i) - \frac{1}{2}(k_i - \ell_i) \log (k_i - \ell_i)
\end{align*}We express the above equation in terms of $k_i$ and set the derivative with respect to $k_i$ to zero to find the maximum. Recall that since $d_i$ is fixed, it suffices to maximize the terms corresponding to the samples individually. Since $\ell_i = d_i - k_i$ its derivative with respect to $k_i$ is $-1$. Using this we have
\begin{align*}
    &\log k_i + 1 - \left(\log \ell_i + 1\right)(-1) - \left(\log (k_i - \ell_i) + 1 \right)(1 - (-1)) = 0 \iff \\ 
    &\log k_i \ell_i  + 2 - 2 \log (k_i - \ell_i) - 2 = 0 \iff \\ 
    & \log \sqrt{ k_i \ell_i} = \log (k_i - \ell_i)
\end{align*}
This gives the following condition.
\begin{align*}
    (k_i \ell_i) = (k_i - \ell_i)^2
\end{align*}
Expanding the RHS and dividing both sides by $\ell_i^2$ we obtain the following quadratic equation
\begin{align*}
    \left(\frac{k_i}{\ell_i}\right)^2 - 3 \frac{k_i}{\ell_i} + 1 = 0
\end{align*}
This has solutions
\begin{align*}
\frac{k_i}{\ell_i} = \frac{3 \pm \sqrt{5}}{2}
\end{align*}
Since $k_i \geqslant \ell_i$, we pick the solution for $\frac{k_i}{\ell_i} = \frac{3 + \sqrt{5}}{2}$. 
To show that this is a global maxima, we also consider the second derivative test by taking the derivative of 
\begin{align*}
    \log k_i \ell_i - 2 \log (k_i - \ell_i) = \log k_i + \log \ell_i - 2 \log (k_i - \ell_i)
\end{align*}
Taking derivative with respect to $k_i$ we have that 
\begin{align*}
    \frac{1}{k_i} - \frac{1}{\ell_i} - \frac{2}{k_i - \ell_i} \cdot 2 = \frac{1}{k_i} - \frac{1}{\ell_i} - \frac{4}{k_i - \ell_i}  
\end{align*}
which is always negative since $k_i \geqslant \ell_i$. Substituting back in the original expression we obtain that
\begin{align*}
    \log T_i &\approx \frac{1}{2}(k_i \log k_i - \ell_i \log \ell_i) - \left(\frac{1}{2}(k_i - \ell_i) \log (k_i - \ell_i) \right) \\
    &= \frac{1}{2}(k_i \log k_i - \ell_i \log \ell_i) - \frac{(k_i - \ell_i)}{2} \log \sqrt{k_i \ell_i} \\
    &= \frac{1}{2}(k_i \log k_i - \ell_i \log \ell_i) - \frac{(k_i - \ell_i)}{4}\log k_i - \frac{(k_i - \ell_i)}{4} \log \ell_i\\
    &= \left(\frac{k_i + \ell_i}{4}\right) \log k_i - \left(\frac{\ell_i}{2} + \frac{k_i - \ell_i}{4}\right) \log \ell_i\\
    &= \frac{d_i}{4} \log \frac{k_i}{\ell_i} = \frac{d_i}{4} \log \frac{3 + \sqrt{5}}{2} = \frac{d_i}{4} \log  \left(\frac{1 + \sqrt{5}}{2}\right)^2 = \frac{d_i}{2} \log  \frac{1 + \sqrt{5}}{2} \leqslant \frac{d_i}{2} \log 2
\end{align*}
Therefore,  after  taking product over the $m$ samples for a fixed set of $d_i$, we have that 
\begin{align*}
    \prod_{i \in S}  T_i^2 &\leqslant \prod_{i \in S} 2^{d_i} =  2^{p} 
\end{align*} 
\end{proof}

\subsection{Reduction}
In this section we will state our reduction. Our reduction is slightly involved as we have an \emph{asymmetric} robust regression instance in the null distribution and the alternative distribution. This is because the label noise in both cases are different, and in the null distribution it could be arbitrarily large.  We now define the following more specific problem which we will utilize for our reduction.  
\begin{problem}[Robust Linear Testing with a Fixed Adversary]\label{prob:fixedlineartesting}
    Given corruption rate $\epsilon \in (0, 1/2)$, $\kappa\geq 1$, signal strength $\alpha \in \R_+$, sample size $n \in \N$, dimension $d \in \N$, define $\sigma_y^2 \coloneqq \alpha^2 + 1, \delta \coloneqq \alpha \sqrt{\kappa}$ and consider the following hypothesis testing problem with input samples $\{z_i\}_{i=1}^n \in \R^{n \times (d+1)}$. 
\begin{enumerate}
    \item $H_0$: Null $\cQ:$ Let $X \sim  \cN(0, I_d)$ and $Y \sim  \cN(0, \sigma_y^2)$ be independent.
        Define $Z := (X,Y) \in \mathbb R^{d+1}$.
        Then
        \[
        z_1, \dots, z_n \overset{\mathrm{iid}}{\sim}
        \cQ\defeq\cN \left(0,\operatorname{diag} \left(I_d,\sigma_y^2 \right)\right).
        \]
        
    \item $H_1$: Alternative $\cP$: 
    $v \sim \mathrm{Unif}(\cS^{d-1})$. Conditioned on $v$,
    \begin{align*}
    z_1, \dots, z_n \overset{\mathrm{iid}}{\sim} \cP \defeq   (1 - \eps) \cdot \   \cN\left(0,  \Sigma_1\right) + \eps \cdot \cN(0, \Sigma_2)
    \end{align*}
    where 
    \[
        \Sigma_1 = \begin{bmatrix}
        I_d - (1 - 1/\kappa)vv^T & \frac{\delta}{\kappa} v \\ 
    \frac{\delta}{\kappa} v^T & \frac{\delta^2}{\kappa} + 1 \end{bmatrix} \\
    \]
    and 
    \[
      \Sigma_2 = \begin{bmatrix}
         I_d + \frac{(1-\epsilon)}{\epsilon} \cdot (1 - 1/\kappa)vv^T & -\frac{(1-\epsilon)}{\epsilon} \frac{\delta}{\kappa}v  \\ 
    -\frac{(1-\epsilon)}{\epsilon} \frac{\delta}{\kappa}v^T  & \frac{\delta^2}{\kappa} + 1
    \end{bmatrix}.
    \]
\end{enumerate}
We observe that the $\cN(0, \Sigma_1)$ above describes an uncorrupted linear regression model with $\beta = \delta v$, $\Sigma = I - (1 - 1/\kappa)vv^T$ and $\sigma^2 = 1$.
\end{problem}
\begin{conjecture}\label{conj:hardnessfixedadversary}
   For any $\alpha > 0$, $\eps \kappa \geq C$ for some constant $C > 0$ sufficiently large,~\cref{prob:fixedlineartesting} is computationally hard for efficient algorithms unless $n  = \tilde{\Omega}( \min\{d \eps^2 \kappa^2, \eps^2 d^2  \})$. 
\end{conjecture}
~\cref{conj:hardnessfixedadversary} is supported by our main low-degree lower bound in~\cref{thm:mainlb}. In the above instance we set the label noise variance to $1$ for simplicity. We now provide a reduction from~\cref{prob:fixedlineartesting} to robust regression as follows. 

\subsubsection{The regression algorithm} The estimation algorithm $\mathcal A$ takes as input (i) $\varepsilon$-corrupted samples from the linear model $y = \langle X, \beta \rangle + \eta$
for $X \sim \mathcal N(0, \Sigma)$, $\eta \sim \mathcal N(0, \sigma^2)$ is independent noise and $\eps \in (0, 1/2)$ and (ii) a parameter $\alpha$ such that whenever $\Vert \Sigma^{1/2} \beta \Vert = \alpha \text{ and } 0 < \sigma^2 \le 1$, the algorithm outputs an estimator $\widehat{\beta}$ satisfying
\[
\vnorm{\Sigma^{1/2}\left(\widehat{\beta} - \beta\right)} \le 0.1\alpha
\]
with probability $1 - o(1)$, using $n$ samples and running time $T$. 

We note here that the choice of $0 < \sigma^2 \leq 1$ is for simplicity and one can indeed let the error of the regression algorithm scale with $\sigma$ more explicitly. We will assume $\alpha = \Omega(1)$ for our reduction. Information-theoretically, one requires $\alpha = \Omega(\sigma\epsilon)$, which in our case reduces to $\Omega(\epsilon)$. In this regard, our reduction requires (slightly) stronger assumption.

As a first step, we will prove that an estimator achieving the above error guarantee has good correlation with the underlying unknown vector in our hard instance.

\begin{lemma}[Small prediction error implies correlation]\label{lemma:reductioncorrelation}
When $\mathcal A$ is run on samples from the alternative distribution $\mathcal P$ in~\cref{prob:fixedlineartesting}, using the first $d$ coordinates as features and the last coordinate as the label, the estimator $\widehat{\beta}$ satisfies
\[
|\langle \widehat{\beta}, v\rangle| \ge 0.9\delta
\qquad\text{and}\qquad
\|\widehat{\beta}\| \le 1.1\delta
\]
with probability $1-o(1)$.
\end{lemma}

\begin{proof}
Under the alternative, the regression vector is $\beta=\delta v$, where $v$ is a unit vector and
\[
\delta=\alpha\sqrt{\kappa}.
\]
Moreover, since
\[
\Sigma = I_d - (1-1/\kappa)vv^\top,
\]
we have
\[
\Sigma v = \frac{1}{\kappa}v
\qquad\text{and hence}\qquad
\Sigma^{1/2}v = \frac{1}{\sqrt{\kappa}}v.
\]
By the guarantee of $\mathcal A$, with probability $1-o(1)$,
\[
\|\Sigma^{1/2}(\widehat{\beta}-\beta)\| \le 0.1\alpha.
\]
On this event,
\begin{align*}
\left|\left\langle \widehat{\beta}-\beta,\,v \right\rangle\right|
&=
\left|\left\langle \Sigma^{1/2}(\widehat{\beta}-\beta),\,\Sigma^{-1/2}v \right\rangle\right| \\
&\le
\|\Sigma^{1/2}(\widehat{\beta}-\beta)\| \cdot \|\Sigma^{-1/2}v\| \\
&=
\|\Sigma^{1/2}(\widehat{\beta}-\beta)\| \cdot \sqrt{\kappa} \\
&\le
0.1\alpha\sqrt{\kappa}
=
0.1\delta.
\end{align*}
Using $\beta=\delta v$, it follows that
\begin{align*}
|\langle \widehat{\beta},v\rangle|
&=
|\langle \beta,v\rangle + \langle \widehat{\beta}-\beta,v\rangle| \\
&\ge
|\langle \beta,v\rangle| - |\langle \widehat{\beta}-\beta,v\rangle| \\
&\ge
\delta - 0.1\delta
=
0.9\delta.
\end{align*}
Now, for the norm bound, we similarly have
\begin{align*}
\|\widehat{\beta}\|
&\le
\|\widehat{\beta}-\beta\| + \|\beta\| \\
&\le
\|\Sigma^{-1/2}\|\,\|\Sigma^{1/2}(\widehat{\beta}-\beta)\| + \delta.
\end{align*}
Since the smallest eigenvalue of $\Sigma$ is $1/\kappa$, we have $\|\Sigma^{-1/2}\|=\sqrt{\kappa}$. Therefore,
\begin{align*}
\|\widehat{\beta}\|
&\le
\sqrt{\kappa}\cdot 0.1\alpha + \delta \\
&=
0.1\delta + \delta \\
&=
1.1\delta.
\end{align*}
This proves the claim.
\end{proof}

Therefore, under the alternative distribution, the algorithm $\mathcal A$ outputs an estimator that is well correlated with the hidden direction $v$. We now use this observation to formalize the reduction. The ideas underlying this reduction are standard and closely related to those used in the literature; see, for example,~\cite{brennan2018reducibility,brennan2020reducibility}.

\begin{lemma}[Testing using an estimation algorithm]\label{lemma:fixedreduction}
There is an algorithm $\mathcal B$ that
\begin{enumerate}
    \item takes $2n$ samples from~\cref{prob:fixedlineartesting},
    \item runs in time $T + \poly(n,d)$, and
    \item distinguishes the null and alternative in~\cref{prob:fixedlineartesting} with probability $1-o(1)$.
\end{enumerate}
\end{lemma}
\noindent Before proving the lemma, we note that the testing problem itself requires $n=\Omega(d)$ samples information-theoretically. In the regime $\epsilon \gg 1/\sqrt d$, the sample size required by the tester constructed below will be much smaller than $d$, so this information-theoretic lower bound does not obstruct the reduction.

\begin{proof}
The algorithm $\mathcal B$ proceeds as follows. Given samples
\[
z_1,z_2,\dots,z_{2n}
\]
from~\cref{prob:fixedlineartesting}, split them into two halves. Using the first half, namely $z_1,\dots,z_n$, run the regression algorithm $\mathcal A$ with the first $d$ coordinates of each sample as the feature vector and the last coordinate as the response. Let $\widehat\beta$ denote the output. Define
\[
\widehat v :=
\begin{cases}
\widehat\beta / \|\widehat\beta\|, & \text{if } \widehat\beta \neq 0,\\
\hat{e}, & \text{if } \widehat\beta = 0,
\end{cases}
\]
where $\hat{e}$ is any fixed unit vector. By construction, $\widehat v$ depends only on the first half of the samples, and is therefore independent of the second half. We then use the remaining samples $z_{n+1},\dots,z_{2n}$ to compute a test statistic, with the choice depending on the value of $\kappa$. Writing
\[
X_{i+n} := z_{i+n,1:d}
\qquad\text{and}\qquad
y_{i+n} := z_{i+n,d+1},
\]
define
\[
f(z) :=
\begin{cases}
\displaystyle
\frac{1}{\alpha^2\sigma_y^2}\sum_{i=1}^n
\langle X_{i+n},\widehat v\rangle^2\bigl(y_{i+n}^2-\sigma_y^2\bigr),
& \text{if } \kappa \le \sqrt d,\\[3mm]
\displaystyle
\sum_{i=1}^n \Bigl(\langle X_{i+n},\widehat v\rangle^4 - 3\Bigr),
& \text{if } \kappa > \sqrt d.
\end{cases}
\]
The sample split ensures that the randomness in $\widehat v$ is independent of the samples used to compute the test statistic. We analyze the two regimes separately.

\subsection*{Small Condition Number}

In this regime we assume $\kappa \le \sqrt d$ and consider the statistic
\[
f(z)
=
\frac{1}{\alpha^2\sigma_y^2}\sum_{i=1}^n
\langle X_{i+n},\widehat v\rangle^2\bigl(y_{i+n}^2-\sigma_y^2\bigr).
\]
Since $\widehat v$ is computed from the first half of the samples, it is independent of the second half. Thus, throughout the argument below, we will condition on $\widehat v$.

\subsubsection*{Null distribution}

Under the null, conditioned on $\widehat v$, we have
\[
X_{i+n}\sim \cN(0,I_d),
\qquad
y_{i+n}\sim \cN(0,\sigma_y^2),
\]
and these are independent. Therefore,
\begin{align*}
\E[f(z)\mid \widehat v]
&=
\frac{1}{\alpha^2\sigma_y^2}\sum_{i=1}^n
\E\!\left[\langle X_{i+n},\widehat v\rangle^2(y_{i+n}^2-\sigma_y^2)\mid \widehat v\right] \\
&=
\frac{1}{\alpha^2\sigma_y^2}\sum_{i=1}^n
\E\!\left[\langle X_{i+n},\widehat v\rangle^2\mid \widehat v\right]
\E\!\left[y_{i+n}^2-\sigma_y^2\right] \\
&= 0.
\end{align*}
For the variance, the summands are independent and centered, so
\[
\Var(f(z)\mid \widehat v)
=
n\Var\!\left(
\frac{\langle X,\widehat v\rangle^2(y^2-\sigma_y^2)}{\alpha^2\sigma_y^2}
\Bigm| \widehat v
\right),
\]
where $X\sim N(0,I_d)$ and $y\sim N(0,\sigma_y^2)$ are independent. Hence
\begin{align*}
\Var(f(z)\mid \widehat v)
&=
\frac{n}{\alpha^4\sigma_y^4}
\E[\langle X,\widehat v\rangle^4\mid \widehat v]\,
\E[(y^2-\sigma_y^2)^2] \\
&=
\frac{n}{\alpha^4\sigma_y^4}\cdot 3 \cdot 2\sigma_y^4 \\
&=
\frac{6n}{\alpha^4}.
\end{align*}
Thus,
\[
\Var(f(z)\mid \widehat v)=O\!\left(\frac{n}{\alpha^4}\right).
\]

\subsubsection*{Alternative distribution}

We first compute $\E[XX^\top y^2]$ under the alternative.

\begin{fact}\label{fact:wickfourthmoment}
If $(X,y)$ is jointly Gaussian with mean zero and block covariance
\[
\begin{pmatrix}
\Sigma & c\\
c^\top & \sigma_y^2
\end{pmatrix},
\]
then
\[
\E[XX^\top y^2] = \sigma_y^2 \Sigma + 2cc^\top.
\]
\end{fact}
\noindent Under the alternative, the first mixture component has
\[
\Sigma_{x,1}=I_d-(1-1/\kappa)vv^\top,
\qquad
c_1=\frac{\delta}{\kappa}v,
\]
while the second has
\[
\Sigma_{x,2}=I_d+\frac{1-\epsilon}{\epsilon}(1-1/\kappa)vv^\top,
\qquad
c_2=-\frac{1-\epsilon}{\epsilon}\frac{\delta}{\kappa}v.
\]
Applying~\cref{fact:wickfourthmoment} to each component and averaging over the mixture gives
\begin{align*}
\E[XX^\top y^2]
&=
(1-\epsilon)\bigl(\sigma_y^2\Sigma_{x,1}+2c_1c_1^\top\bigr)
+
\epsilon\bigl(\sigma_y^2\Sigma_{x,2}+2c_2c_2^\top\bigr).
\end{align*}
Since
\[
(1-\epsilon)\Sigma_{x,1}+\epsilon\Sigma_{x,2}=I_d,
\]
the covariance contribution is $\sigma_y^2 I_d$. Also,
\begin{align*}
(1-\epsilon)c_1c_1^\top+\epsilon c_2c_2^\top
&=
(1-\epsilon)\frac{\delta^2}{\kappa^2}vv^\top
+
\epsilon\left(\frac{1-\epsilon}{\epsilon}\right)^2\frac{\delta^2}{\kappa^2}vv^\top \\
&=
\frac{1-\epsilon}{\epsilon}\frac{\delta^2}{\kappa^2}vv^\top \\
&=
\frac{1-\epsilon}{\epsilon}\frac{\alpha^2}{\kappa}vv^\top.
\end{align*}
Therefore,
\[
\E[XX^\top y^2]
=
\sigma_y^2 I_d
+
2\frac{1-\epsilon}{\epsilon}\frac{\alpha^2}{\kappa}vv^\top.
\]
Using also $\E[XX^\top]=I_d$, we obtain
\begin{align*}
\E[f(z)\mid \widehat v]
&=
\frac{n}{\alpha^2\sigma_y^2}
\E\!\left[\langle X,\widehat v\rangle^2(y^2-\sigma_y^2)\mid \widehat v\right] \\
&=
\frac{n}{\alpha^2\sigma_y^2}
\widehat v^\top\bigl(\E[XX^\top y^2]-\sigma_y^2 I_d\bigr)\widehat v \\
&=
\frac{2n(1-\epsilon)}{\epsilon\kappa\sigma_y^2}\langle v,\widehat v\rangle^2.
\end{align*}
By~\cref{lemma:reductioncorrelation},
\[
|\langle \widehat\beta,v\rangle|\ge 0.9\delta
\qquad\text{and}\qquad
\|\widehat\beta\|\le 1.1\delta,
\]
and hence
\[
\langle v,\widehat v\rangle^2
=
\frac{\langle \widehat\beta,v\rangle^2}{\|\widehat\beta\|^2}
\ge
\left(\frac{0.9}{1.1}\right)^2
>
\frac14.
\]
Since $\epsilon<1/2$, we also have $1-\epsilon\ge 1/2$. Therefore,
\[
\E[f(z)\mid \widehat v]
\ge
\frac{n}{4\epsilon\kappa\sigma_y^2}
=
\frac{n}{4\epsilon\kappa(\alpha^2+1)}.
\]
We now bound the variance. Writing
\[
W:=\frac{\langle X,\widehat v\rangle^2(y^2-\sigma_y^2)}{\alpha^2\sigma_y^2},
\]
we have
\[
\Var(f(z)\mid \widehat v)
=
n\Var(W\mid \widehat v)
\le
n\E[W^2\mid \widehat v].
\]
Thus
\[
\Var(f(z)\mid \widehat v)
\le
\frac{n}{\alpha^4\sigma_y^4}
\E\!\left[\langle X,\widehat v\rangle^4(y^2-\sigma_y^2)^2 \mid \widehat v\right].
\]
By Cauchy--Schwarz,
\[
\E\!\left[\langle X,\widehat v\rangle^4(y^2-\sigma_y^2)^2 \mid \widehat v\right]
\le
\sqrt{\E[\langle X,\widehat v\rangle^8\mid \widehat v]}\,
\sqrt{\E[(y^2-\sigma_y^2)^4]}.
\]
Since $y\sim N(0,\sigma_y^2)$ in each mixture component,
\[
\E[(y^2-\sigma_y^2)^4]=O(\sigma_y^8).
\]
Moreover, under the first mixture component,
\[
\langle X,\widehat v\rangle \sim N(0,s_1^2),
\qquad
s_1^2=1-(1-1/\kappa)\langle \widehat v,v\rangle^2 \le 1,
\]
while under the second,
\[
\langle X,\widehat v\rangle \sim N(0,s_2^2),
\qquad
s_2^2=1+\frac{1-\epsilon}{\epsilon}(1-1/\kappa)\langle \widehat v,v\rangle^2 \le \frac{1}{\epsilon}.
\]
Therefore,
\begin{align*}
\E[\langle X,\widehat v\rangle^8\mid \widehat v]
&=
105\bigl((1-\epsilon)s_1^8+\epsilon s_2^8\bigr) \\
&=
O(\epsilon^{-3}).
\end{align*}
Combining the above bounds yields
\[
\Var(f(z)\mid \widehat v)
\le
O\!\left(\frac{n}{\alpha^4\epsilon^{3/2}}\right).
\]
We now define the test to accept the alternative iff
\[
f(z)\ge T,
\qquad
T:=\frac{n}{8\epsilon\kappa(\alpha^2+1)}.
\]
Under the null, Chebyshev's inequality gives
\[
\Pr_{H_0}\bigl(|f(z)|>T \mid \widehat v\bigr)
\le
\frac{\Var(f(z)\mid \widehat v)}{T^2}
=
O\!\left(\frac{\epsilon^2\kappa^2(\alpha^2+1)^2}{n\alpha^4}\right).
\]
Under the alternative, since $\E[f(z)\mid \widehat v]\ge 2T$, another application of Chebyshev gives
\[
\Pr_{H_1}\bigl(f(z)\le T \mid \widehat v\bigr)
\le
\frac{\Var(f(z)\mid \widehat v)}{T^2}
=
O\!\left(\frac{\epsilon^{1/2}\kappa^2(\alpha^2+1)^2}{n\alpha^4}\right).
\]
Since $\alpha=\Omega(1)$, we have
\[
\frac{(\alpha^2+1)^2}{\alpha^4}=O(1),
\]
and hence
\[
\Pr_{H_0}\bigl(|f(z)|>T \mid \widehat v\bigr)
=
O\!\left(\frac{\epsilon^2\kappa^2}{n}\right),
\qquad
\Pr_{H_1}\bigl(f(z)\le T \mid \widehat v\bigr)
=
O\!\left(\frac{\epsilon^{1/2}\kappa^2}{n}\right).
\]
Therefore, it suffices to take
\[
n \gg \epsilon^{1/2}\kappa^2
\]
for this to yield a valid distinguisher. Since $n \gtrsim d$ samples is required for the regression algorithm even information-theoretically, and $\eps^{1/2} \kappa^2 \lesssim d$, we are done.

\subsection*{Large condition number}

We now consider the regime $\kappa > \sqrt d$, in which the test statistic is
\[
f(z)
=
\sum_{i=1}^n \Bigl(\langle X_{i+n},\widehat v\rangle^4 - 3\Bigr).
\]
As before, we condition on $\widehat v$, which is independent of the second half of the samples.

\subsubsection*{Null distribution}

Under the null, conditioned on $\widehat v$, each
\[
\langle X_{i+n},\widehat v\rangle \sim N(0,1).
\]
Since the fourth moment of a standard Gaussian is $3$, we have
\[
\E[f(z)\mid \widehat v]=0.
\]
Moreover, the summands are independent, so
\begin{align*}
\Var(f(z)\mid \widehat v)
&=
n\,\Var\!\left(\langle X,\widehat v\rangle^4-3 \mid \widehat v\right) \\
&\le
n\,\E\!\left[(\langle X,\widehat v\rangle^4-3)^2 \mid \widehat v\right] \\
&=
n\,(\E[U^8]-9),
\end{align*}
where $U\sim N(0,1)$. Since $\E[U^8]=105$, it follows that
\[
\Var(f(z)\mid \widehat v)=96n=O(n).
\]

\subsubsection*{Alternative distribution}

Let
\[
t^2 := \langle \widehat v,v\rangle^2.
\]
Under the first mixture component,
\[
\langle X,\widehat v\rangle \sim N(0,s_1^2),
\qquad
s_1^2 = 1-a,
\qquad
a := (1-1/\kappa)t^2,
\]
while under the second,
\[
\langle X,\widehat v\rangle \sim N(0,s_2^2),
\qquad
s_2^2 = 1+b,
\qquad
b := \frac{1-\epsilon}{\epsilon}(1-1/\kappa)t^2.
\]
Therefore,
\begin{align*}
\E[\langle X,\widehat v\rangle^4 - 3 \mid \widehat v]
&=
3\bigl((1-\epsilon)(1-a)^2 + \epsilon(1+b)^2 - 1\bigr).
\end{align*}
Expanding the right-hand side and using
\[
-(1-\epsilon)a + \epsilon b = 0,
\]
we obtain
\begin{align*}
(1-\epsilon)(1-a)^2 + \epsilon(1+b)^2 - 1
&=
(1-\epsilon)a^2 + \epsilon b^2 \\
&=
\frac{1-\epsilon}{\epsilon}a^2.
\end{align*}
Hence,
\[
\E[\langle X,\widehat v\rangle^4 - 3 \mid \widehat v]
=
3\frac{1-\epsilon}{\epsilon}(1-1/\kappa)^2 t^4.
\]
By~\cref{lemma:reductioncorrelation} once again we have,
\[
|\langle \widehat\beta,v\rangle| \ge 0.9\delta
\qquad\text{and}\qquad
\|\widehat\beta\| \le 1.1\delta,
\]
and therefore
\[
t^2 = \langle \widehat v,v\rangle^2 \ge \left(\frac{0.9}{1.1}\right)^2 > \frac14,
\qquad\text{so}\qquad
t^4 \ge \frac{1}{16}.
\]
Since $\kappa>\sqrt d$, we have $(1-1/\kappa)^2=\Theta(1)$ in the regime of interest, and thus
\[
\E[\langle X,\widehat v\rangle^4 - 3 \mid \widehat v] = \Omega(1/\epsilon).
\]
Summing over the $n$ independent samples yields
\[
\E[f(z)\mid \widehat v] = \Omega(n/\epsilon).
\]
For the variance, we use the same eighth-moment bound as in the small-condition-number case:
\[
\E[\langle X,\widehat v\rangle^8 \mid \widehat v] = O(\epsilon^{-3}).
\]
Therefore,
\begin{align*}
\Var(f(z)\mid \widehat v)
&=
n\,\Var\!\left(\langle X,\widehat v\rangle^4-3 \mid \widehat v\right) \\
&\le
n\,\E[\langle X,\widehat v\rangle^8 \mid \widehat v] \\
&=
O(n/\epsilon^3).
\end{align*}
We now choose a threshold
\[
T := c\,\frac{n}{\epsilon}
\]
for a sufficiently small absolute constant $c>0$. Under the null, Chebyshev's inequality gives
\[
\Pr_{H_0}\bigl(|f(z)|>T \mid \widehat v\bigr)
\le
\frac{\Var(f(z)\mid \widehat v)}{T^2}
=
O\!\left(\frac{\epsilon^2}{n}\right).
\]
Under the alternative, since $\E[f(z)\mid \widehat v]\ge 2T$ for a suitable choice of $c$, another application of Chebyshev's inequality gives
\[
\Pr_{H_1}\bigl(f(z)\le T \mid \widehat v\bigr)
\le
\frac{\Var(f(z)\mid \widehat v)}{T^2}
=
O\!\left(\frac{1}{\epsilon n}\right).
\]
Hence it suffices to take
\[
n \gg \frac{1}{\epsilon}
\]
for this to yield a valid distinguisher and this is much fewer that what we require even information-theoretically for reasonable parameter choices of $\eps \gg \frac{1}{d}$ in the high-dimensional setting.
Combining this with the analysis in the regime $\kappa \le \sqrt d$ completes the proof.
\end{proof}

\begin{corollary}[Hardness of estimation]\label{corollary:hardnessofestimation}
Assuming~\cref{conj:hardnessfixedadversary}, for signal strength $\alpha=\Omega(1)$ it is computationally hard for efficient algorithms to output an estimator $\widehat\beta$ satisfying
\[
\left\|\Sigma^{1/2}(\widehat\beta-\beta)\right\| \le 0.1\alpha
\]
using
\[
n=o\!\left(\min\{d\epsilon^2\kappa^2,\ \epsilon^2 d^2\}\right)
\]
samples, up to $\poly(\log d)$ factors.
\end{corollary}

In particular, this suggests that even achieving constant-factor relative error in prediction norm may be computationally hard below the above sample complexity.

\section{Consequences for Private Regression}\label{sec:privateReg}

We now derive consequences of~\cref{corollary:hardnessofestimation} for differentially private regression via the connection between privacy and robustness. This connection has been studied extensively; see, for example,~\cite{dwork2009differential,georgiev2022privacy}. Throughout this section, we closely follow the analogous argument of~\cite{diakonikolas2025sos} for covariance-aware mean estimation.

Let $n_{\mathrm{private,eff}}(\alpha,\gamma,\varepsilon,\delta)$ denote the sample complexity of the best efficient $(\varepsilon,\delta)$-DP algorithm that, given $n$ i.i.d. samples $(X_i,Y_i)_{i=1}^n$ from the linear model
\[
Y=\langle X,\beta\rangle+\zeta,
\qquad
X\sim \mathcal N(0,\Sigma),
\qquad
\zeta\sim \mathcal N(0,1),
\]
where $\Sigma$ is unknown and $\zeta$ is independent of $X$, outputs an estimator $\widehat\beta$ satisfying
\[
\|\Sigma^{1/2}(\widehat\beta-\beta)\| \le \alpha
\]
with probability at least $1-\gamma$.

Similarly, let $n_{\mathrm{robust,eff}}(\alpha,\gamma,\eta)$ denote the sample complexity of the best efficient robust algorithm that, given an $\eta$-corrupted set of $n$ samples from the same linear model, outputs with probability at least $1-\gamma$ an estimator $\widehat\beta$ satisfying
\[
\|\Sigma^{1/2}(\widehat\beta-\beta)\| \le \alpha.
\]

\begin{conjecture}\label{conj:regression}
For any $\alpha \gtrsim \eta \log(1/\eta)$ and any $\eta\kappa \ge C$ for a sufficiently large absolute constant $C>0$,
\[
n_{\mathrm{robust,eff}}(\alpha,\gamma,\eta)
\gg
\min\{d\eta^2\kappa^2,\eta^2 d^2\},
\]
where $\kappa$ denotes the condition number of $\Sigma$.
\end{conjecture}
\cref{corollary:hardnessofestimation} provides evidence for~\cref{conj:regression}. In this section, we restrict attention to the regime
\[
\kappa \ge \sqrt d,
\]
in which the conjectured lower bound simplifies to
\[
n_{\mathrm{robust,eff}}(\alpha,\gamma,\eta)\gg \eta^2 d^2.
\]

\begin{proposition}\label{prop:privateregression}
Assume~\cref{conj:regression}. Then for any $\alpha \gtrsim \eta \log(1/\eta)$ and any $\kappa \ge \sqrt d$,
\[
n_{\mathrm{private,eff}}(\alpha,\gamma,\varepsilon,\delta)
\gg
\max_{t\in(0,1/2)}
\min\left\{
d^{2-2t},
\;
d^t\frac{\log(1/\gamma)}{\varepsilon},
\;
d^t\frac{\log(1/\delta)}{\varepsilon}
\right\}.
\]
\end{proposition}

We next give some context for~\cref{prop:privateregression} before turning to the proof.

Let $n_{\mathrm{private}}(\alpha,\gamma,\varepsilon,\delta)$ denote the \emph{information-theoretic} sample complexity of private regression, namely the minimum number of samples required by any $(\varepsilon,\delta)$-DP algorithm to achieve error $\alpha$ with success probability at least $1-\gamma$.

Specializing to the setting $\alpha=1$ and $\varepsilon=1$,~\cite[Corollary 4.16]{liu2022differential} gives
\[
n_{\mathrm{private}}(1,\gamma,1,\delta)
\lesssim
d+\log(1/\delta)+\log(1/\gamma).
\]
By contrast,~\cref{prop:privateregression} suggests that for efficient algorithms and $\kappa\ge \sqrt d$,
\[
n_{\mathrm{private,eff}}(1,\gamma,1,\delta)
\gg
\max_{t\in(0,1/2)}
\min\left\{
d^{2-2t},
\;
d^t \log(1/\gamma),
\;
d^t \log(1/\delta)
\right\}.
\]
The current state-of-the-art efficient algorithm of~\cite{anderson2025sample} uses
\[
n \lesssim d^2 + d + \log(1/\delta) + \log(1/\gamma)
\]
samples. While this matches the information-theoretic dependence on the privacy parameters, it requires $\Omega(d^2)$ samples. In the subquadratic-sample regime, the best currently known algorithm, due to~\cite{brown2024insufficient}, uses
\[
n \lesssim d + d\sqrt{\log(1/\delta)} + d(\log(1/\delta))^2
\]
samples.\footnote{Both algorithms apply for all $\kappa\ge 1$.}
Thus,~\cref{prop:privateregression} gives evidence, in the regime $\kappa\ge \sqrt d$, that any efficient algorithm using only $d^{2-\Omega(1)}$ samples must pay a polynomial factor in $d$ in front of either $\log(1/\gamma)$ or $\log(1/\delta)$. In particular, such algorithms cannot fully decouple the dimension from the logarithmic dependence on the privacy and failure parameters.

\begin{proof}
We use the privacy-to-robustness reduction of~\cite[Theorem 3.1]{georgiev2022privacy}. It implies that high-probability private estimation yields robustness to a nontrivial fraction of corruptions.

Concretely, suppose $\mathcal A$ is an efficient algorithm such that

\begin{enumerate}
    \item $\mathcal A$ is $(\varepsilon,\delta)$-DP, and
    \item whenever the input consists of $n$ i.i.d. samples $(X_i,Y_i)_{i=1}^n$ from the linear model
    \[
    Y=\langle X,\beta\rangle+\zeta,
    \qquad
    X\sim \mathcal N(0,\Sigma),
    \qquad
    \zeta\sim \mathcal N(0,1),
    \]
    with $\Sigma$ unknown and $\zeta$ independent of $X$, the algorithm outputs, with probability at least $1-\gamma$, an estimator $\widehat\beta$ satisfying
    \[
    \|\Sigma^{1/2}(\widehat\beta-\beta)\| \le \alpha.
    \]
\end{enumerate}

Then the same algorithm is robust to an $\eta$-fraction of corruptions, where
\[
\eta
=
\Theta\!\left(
\min\left\{
\frac{\log(1/\gamma)}{\varepsilon n},
\;
\frac{\log(1/\delta)}{\varepsilon n+\log n}
\right\}
\right),
\]
in the sense that on an $\eta$-corrupted sample set from the same model, it still outputs an estimator with error at most $\alpha$ with probability at least $1-\gamma^{\Omega(1)}$.

We now prove the proposition by contradiction. Fix $t\in(0,1/2)$, and suppose there exists an efficient $(\varepsilon,\delta)$-DP algorithm $\mathcal A$ using
\[
n_s
\ll
\min\left\{
d^{2-2t},
\;
d^t\frac{\log(1/\gamma)}{\varepsilon},
\;
d^t\frac{\log(1/\delta)}{\varepsilon}
\right\}
\]
samples. In particular,
\[
n_s \le d^t \cdot \frac{\log(1/\gamma)}{\varepsilon}
\qquad\text{and}\qquad
n_s \le d^t \cdot \frac{\log(1/\delta)}{\varepsilon}.
\]
Therefore the induced corruption level satisfies
\[
\eta_s
=
\Theta\!\left(
\min\left\{
\frac{\log(1/\gamma)}{\varepsilon n_s},
\;
\frac{\log(1/\delta)}{\varepsilon n_s+\log n_s}
\right\}
\right)
\gtrsim d^{-t}.
\]
Since we are in the regime $\kappa\ge \sqrt d$,~\cref{conj:regression} gives
\[
n_{\mathrm{robust,eff}}(\alpha,\gamma,\eta_s)
\gg
\eta_s^2 d^2
\gtrsim
d^{2-2t}.
\]
But $\mathcal A$, viewed as a robust estimator through the privacy-to-robustness reduction, uses only $n_s \ll d^{2-2t}$ samples, which results in a contradiction. This proves our claim.
\end{proof}

\section{Conclusion}

We study robust linear regression under strong contamination with Gaussian covariates of unknown covariance, focusing on the tradeoffs between sample complexity, condition number, runtime and prediction error. We give a near-linear-time algorithm which requires $\widetilde{O}(d)$ samples under mild conditions ($\epsilon\kappa\lesssim 1$), with error $O(\sigma \sqrt{\varepsilon \kappa})$ for Gaussian distributions, improving upon prior works, and answering an open problem by~\cite{jambulapati2021robust}-- improved error guarantees for fast algorithms do not require an increased sample complexity. 

We complement this with a statistical query lower bound that provides evidence that when $\eps \kappa \lesssim 1$, achieving error $o(\sigma \sqrt{\eps \kappa})$ might be hard unless $\Omega(d^2)$ samples are used. We further provide a low-degree lower bound that gives evidence that without assumptions such as $\eps \kappa \lesssim 1$, efficient algorithms may require $n = \tilde{\Omega}(\min\left\{d \eps^2 \kappa^2, \eps^2 d^2\right\})$ samples to perform substantially better than the trivial estimator that always guesses zero. We leave open the problem of whether there is an efficient algorithm that uses $n = \tilde{O}(d \eps^2 \kappa^2)$ samples when $\kappa \leq \sqrt{d}$ and achieves non-trivial error whenever $\eps \kappa \gg 1$.

\section*{Acknowledgements}
  We thank David Steurer for insightful discussions, and Stefan Tiegel for enlightening conversations regarding prior work. We are also thankful to the anonymous reviewers for their comments regarding presentation and pointing us to relevant prior work.

\newpage 
\newpage
\phantomsection
\bibliographystyle{amsalpha}
\bibliography{refs}

\addcontentsline{toc}{section}{References}
\section*{Appendix} 
\appendix
\crefalias{section}{appendix} %

\section{Covariance Estimation to Regression}\label{sec:covariancetoreg}
In this section we will provide a formal approach to solving robust regression problems given access to a robust covariance estimation algorithm that can estimate the underlying covariance to good accuracy in spectral norm. Specifically, consider the following algorithm $\cA$. 
$\cA$ takes as input $\epsilon-$corrupted samples $\{X_1, X_2, \dots, X_n\} \sim N(0, \Sigma)$ and outputs $\hat{\Sigma}$ such that 
\begin{align*}
    0.9 \Sigma \preceq \hat{\Sigma} \preceq 1.1 \Sigma
\end{align*}
Such a matrix can be efficiently computed~\cite{kothari2018robust, diakonikolas2025sos} using $O(\eps^2 d^2)$ samples. We can draw $2n$ samples from our regression model and use the first $n$ samples to obtain the preconditioner and use the remaining $n$ samples for regression. Abusing notation and reusing the indices, we apply the following preconditioner to the datapoints $\{X_1, X_2, \dots, X_n\}$. We obtain new samples by performing the following linear operation $X_i' \coloneqq \hat{\Sigma}^{-1/2} X_i$. We then consider the new regression instance 
\begin{align*}
    y_i = \langle X_i', \hat{\Sigma}^{1/2} \beta \rangle + \eta
\end{align*}
We observe that this is the same as the model 
\[
    y_i = \langle X_i, \beta \rangle + \eta 
\]
We now run a robust regression algorithm e.g.~\cref{thm:GaussSampleInf} to obtain an estimate $\widetilde{\beta}$ such that $\Vert \widetilde{\beta} - \hat{\Sigma}^{1/2} \beta \Vert \lesssim \sqrt{\eps}$ on this \emph{well-conditioned} instance. We finally output $\hat{\beta} = \hat{\Sigma}^{-1/2} \widetilde{\beta}$. Observe that we have for all $u \in \R^d$ that
$u^T \Sigma u \leqslant 2 \cdot u^T \hat{\Sigma} u$. In particular for $u = \hat{\beta} - \beta$, it suffices to show that $ u^T \hat{\Sigma} u$ is small. To show this, we observe that
\begin{align*}
     (\hat{\beta} - \beta) \hat{\Sigma}     (\hat{\beta} - \beta) = \Vert \hat{\Sigma}^{1/2}(\hat{\beta} - \beta)\Vert^2 = \Vert \hat{\Sigma}^{1/2}(\hat{\Sigma}^{-1/2} \widetilde{\beta} - \beta)\Vert^2 = \Vert \widetilde{\beta} - \hat{\Sigma}^{1/2} \beta \Vert^2 \lesssim \eps
\end{align*}
Therefore we have that $\Vert \hat{\beta} - \beta \Vert_{\Sigma} \lesssim \sqrt{\eps}$.

\section{Fast covariance-aware mean estimation}\label{sec:meanestimation}
In this section we will describe a fast algorithm for covariance-aware mean estimation where given $\eps$-corrupted samples from $\cN(\mu, \Sigma)$ the goal is to output an estimator $\widehat{\mu}$ such that $\Vert \Sigma^{-1/2}(\widehat{\mu} - \mu) \Vert$ is small assuming $\Sigma$ is unknown.
Assuming that we are given corrupted samples from $\cN(\mu, \Sigma)$, for $M \cdot I_d \preceq \Sigma \preceq L \cdot I_d$ with condition number $\kappa = L/M$ , we first run a fast robust mean estimation procedure for distributions with bounded covariance \cite[Theorem 1.1]{dong2019quantum}. In our case since $\Sigma \preceq L \cdot I_d$ this gives us an approach that uses $\widetilde{O}(d)$ samples, runs in time $\widetilde{O}(nd)$ and outputs an estimate $\widehat{\mu}$ that satisfies
\[
\Vert \widehat{\mu} - \mu \Vert \leq O\left(\sqrt{L\eps}\right).
\]
We now pay for the norm conversion penalty. Using the same estimator we observe that we get
\[
\Vert \Sigma^{-1/2}(\widehat{\mu} - \mu) \Vert \leq \Vert \Sigma^{-1/2} \Vert \cdot  \Vert \widehat{\mu} - \mu \Vert \leq \frac{1}{\sqrt{M}} \cdot \sqrt{L \eps} \leq O(\sqrt{\eps \kappa}).
\]
Observe that the above error can grow arbitrarily with the condition number. 

\section{Code for Verification}\label{sec:code}
The relevant Mathematica and SymPy code for verification is available at the following repository: 

\texttt{https://github.com/sdeepaknarayanan/RegressionCOLT26}.

\end{document}